\documentclass[12pt,a4paper]{article}

\usepackage[utf8]{inputenc}
\usepackage[T1]{fontenc}
\usepackage{lmodern}
\usepackage[margin=1in]{geometry}
\usepackage{bm}
\usepackage{bigints}
\usepackage{orcidlink}
\usepackage{afterpage}
\usepackage{placeins}
\usepackage{enumitem}
\usepackage{stackengine}
\usepackage{amssymb}
\usepackage{booktabs}
\usepackage{changepage}
\usepackage{amssymb}
\usepackage{amsthm}
\usepackage{setspace}
\usepackage{float}

\usepackage{hyperref}
\hypersetup{
	colorlinks=true,
	linkcolor=blue,   % Internal links
	citecolor=blue,   % Citation year/number color
	urlcolor=blue     % URL color
}

\usepackage[round]{natbib}
\newcommand{\tb}{\tiny$\bullet$}

\newcommand{\bbt}{\bm{\beta}}
\newcommand{\bth}{\bm{\theta}}
\newcommand{\ba}{\bm{\alpha}}

\newcommand{\me}{\mathrm{e}}
\newcommand{\Yi}{\mathcal{Y}_i}
\newcommand{\Bi}{\mathcal{B}_i}
\newcommand{\Bzi}{\mathcal{B}_{0i}}
\newcommand{\Ei}{\mathcal{E}_{i}}
\newcommand{\Si}{\Sigma_i}

\newcommand{\J}{\mathcal{J}}

\newcommand{\dist}{\mathcal{D}}

\newcommand{\bTh}{\bm{\Theta}}
\newcommand{\bLam}{\bm{\Lambda}}
\newcommand{\bTau}{\bm{\mathcal{T}}}
\newcommand{\blam}{\bm{\lambda}}
\newcommand{\btau}{\bm{\tau}}
\newcommand{\bmu}{\bm{\mu}}

\newcommand{\byid}{\bm{y}_{id}}

\newcommand{\E}{\text{E}}

\newcommand{\vect}{\text{vec}}
\newcommand{\T}{\top}
\newcommand{\tr}{\text{tr}}

\newcommand{\tms}{\!\times\!}

\newcommand{\bv}{\,\big\vert\,}
\renewcommand{\v}{\,\vert\,}
\newcommand{\N}{\text{N}}
\newcommand{\IG}{\text{IG}}
\newcommand{\U}{\text{U}}
\newcommand{\Dir}{\text{Dir}}
\newcommand{\IW}{\text{IW}}

\newcommand{\bz}{\bm{z}}
\newcommand{\ziz}{{z_{i_0}}}
\newcommand{\zio}{{z_{i_1}}}
\newcommand{\zizLs}{{z_{i_0}^{L_s}}}
\newcommand{\zioLs}{{z_{i_1}^{L_s}}}

\newcommand{\zizLm}{{z_{i_0}^{L_m}}}
\newcommand{\zioLm}{{z_{i_1}^{L_m}}}
\newcommand{\znew}{{z_{new}}}
\newcommand{\ziLs}{{z_i^{L_s}}}
\newcommand{\ziLm}{{z_i^{L_m}}}
\newcommand{\zizs}{{z_{i_0}^s}}
\newcommand{\zios}{{z_{i_1}^s}}

\newcommand{\zizm}{{z_{i_0}^m}}
\newcommand{\ziom}{{z_{i_1}^m}}

\newcommand{\Ell}{\mathcal{L}}
\renewcommand{\d}{\,\text{d}}

\newtheorem{theorem}{Theorem}[section]

\newtheorem{proposition}{Proposition}[section]

\usepackage{authblk}
\title{Bayesian repulsive mixture model for multivariate\\functional data}
\author{
	Ricardo Cunha Pedroso\\
	\vspace{-.15in}
	Departamento de Estat{\'i}stica, Universidade Federal de Minas Gerais, Brazil
	\\
	\vspace{.15in}
	Fernando Andr{\'e}s Quintana\\
	Departamento de Estad{\'i}stica, Pontificia Universidad Cat{\'o}lica de Chile, Santiago
	\\
	\vspace{.15in}
	Rosangela Helena Loschi\\
	Departamento de Estat{\'i}stica, Universidade Federal de Minas Gerais, Brazil
}

\date{August, 2026}

\begin{document}

%\begin{adjustwidth}{-1cm}{-1cm}
%\maketitle
%\end{adjustwidth}
%\thispagestyle{empty}

\newgeometry{top=.8in, bottom=.1in, left=.7in, right=.7in}
\maketitle

\begin{abstract}
%We introduce a repulsive mixture model to cluster observation units represented by multivariate functional curves, based on curve shape similarity and individual covariates. We propose a repulsive prior density for the component location parameters that depends on a B-spline curve-tailored distance, extending existing repulsive prior densities to the context of multivariate functional data. The model thus favors the identification of well-differentiated clusters, avoiding the presence of redundant ones. To sample from the posterior distribution, we propose a Markov Chain Monte Carlo algorithm containing a novel split-merge step, which accounts for the repulsive interaction between the mixture components, significantly improving the chain mixing. The effects of repulsion and covariates in the clustering are evaluated through simulation. The proposed model is fitted to a real data set consisting of individuals with Chronic Ankle Instability, with a focus on identifying similar types of physical dysfunctions based on the similarity of movement patterns.
We introduce a repulsive mixture model to cluster observation units represented by multivariate functional curves, based on similarity of curve shapes and individual-specific covariates. The model features a repulsive prior distribution for component-specific location parameters, which depends on a B-spline curve-tailored distance, extending existing repulsive prior distributions to the context of multivariate functional data. The model favors identification of well-differentiated clusters, avoiding redundancy. Posterior simulation is implemented through a Markov Chain Monte Carlo algorithm containing a novel split-merge step that considers the repulsive interactions between mixture components, significantly improving the chain mixing. Different features of the proposed model, including the effects of repulsion and covariates in the clustering, are evaluated through simulation. The model is applied to data concerning multivariate sequences of observations characterizing movement patterns of individuals with Chronic Ankle Instability, with a focus on identifying similar types of physical dysfunctions based on the similarity of movement patterns.
\end{abstract}


\restoregeometry

\thispagestyle{empty}

\newpage
\section{Introduction}\label{sec:intro}

Due to technological advancements, data have been collected almost continuously in temporal intervals, areas or regions, leading to a fast growth of functional data analysis \citep{ramsay2005}. Functional clustering models can be useful to accommodate data heterogeneity and to identify individuals with similar features in functional data analysis, providing pattern recognition and improving predictions. See, e.g., \cite{petrone2009,page2015,zhang2023review,zhang2023,toto2025}, and references therein. Finite mixture models (FMMs) are a popular approach for identifying clusters. A FMM for observation units $\bm{y}_i\in\mathbb{R}^{^p}$, $i=1,\dots,m$, may be hierarchically represented as

\begin{equation*}\label{eq_intro_fmm}
	\begin{aligned}
		\bm{y}_i\v\bz,\Theta &\overset{\text{ind}}{\sim} f(\bm{y}_i\v\bth_{z_i}),\\
		z_i\v\bm{\pi} &\overset{\text{iid}}{\sim} \text{Multinomial}_J(\bm{\pi}),
	\end{aligned}
\end{equation*}

\noindent where $J$ is a fixed known number of mixture components, $\bm{\pi}=(\pi_1,\dots,\pi_J)$ is a vector of non-negative mixing weights, such that ${\sum_{j=1}^{J}\pi_j=1}$, and $\Theta=\{\bth_1,\dots,\bth_J\}$ is a set of the $J$ parameters characterizing the $J$ mixture component densities $f(\bm{y}_i\v\bth_j)$. Cluster identification is achieved through the latent allocation vector $\bz=(z_1, \dots, z_m)$, so that $z_i=j$ if observation $\bm{y}_i$ was drawn from $f(\cdot\v\bth_j)$ and the number of clusters corresponds to the number of occupied components. Assuming a non-degenerate prior distribution for $J$ typifies the class of mixture of finite mixture models (MFMMs)
%\citep{richardson1997,	miller2018, sylvia2021}
\citep[see, e.g.,][]{miller2018, sylvia2021} and assuming $J=\infty$ leads to nonparametric mixture models %\citep{gnedin2006,muller2011}
\citep[see, e.g.][]{ferguson1973,lo1984,hartigan1990,quintana2003,muller2015bayesian}. Typically, it is \textit{a priori} assumed that the mixing weights $\bm{\pi}$ have a Dirichlet distribution and the component-specific parameters in $\Theta$ are independent. The  assumption of independence between component parameters may favor low separation between clusters, possibly leading to the presence of redundant overlapping clusters. One strategy to overcome the redundancy problem is to assume a repulsive interaction between the component-specific parameters. Options to do so in the context of FMMs have been discussed in \cite{petralia2012}, \cite{quinlan2018}, \cite{fuquene2019}, \cite{hayashida2025repulsive}, among others. In particular, \cite{petralia2012} introduced a class of prior densities with a repulsive interaction between the mixture components, based on pairwise distances between component location parameters, that was later generalized by \cite{xie2020} and \cite{quinlan2021}. Alternative constructions are based on point processes \citep[see, e.g.,][]{beraha2022, ghilotti2024, cremaschi2024, song2025repulsive}.

Regarding functional data, Bayesian clustering models have been proposed by \cite{petrone2009}, \cite{page2015}, %\cite{hopkins2019},
\cite{rigon2023},
%\cite{zhang2023},
\cite{zhong2024}, \cite{liang2024}, \cite{toto2025}, among others. Particularly, \cite{page2015} proposed a product partition model \citep{hartigan1990} to cluster univarite B-spline curves based on the similarity of the shapes of the curves, and also allowing individual-specific covariates to influence clustering. Later, \cite{hopkins2019} extended such model to bivariate B-spline curves.

\subsection{Our contributions}

Our model was motivated by the case study discussed in \cite{hopkins2019}, which set out to identify similar types of physical dysfunctions in a set of individuals with Chronic Ankle Instability (CAI), based on movement patterns represented by multivariate functional curves. A critical point in their analysis was the identification of different clusters with very low differentiation between their curve shapes, characterizing a redundancy problem. To mitigate this issue, we introduce a FMM for multivariate B-spline curves and propose a repulsive joint prior distribution for the location parameters of the mixture components. The proposed repulsive prior depends on a B-spline curve-tailored distance, in order to favor the identification of well-differentiated clusters with regard to the shape of the multivariate curves. This prior was inspired by the class of repulsive prior densities introduced by \cite{petralia2012}. To the best of our knowledge, the functional clustering models proposed so far have not explored a repulsive interaction between mixture components to influence cluster identification.
%\trs{in this multivariate context}. 
Additionally, the proposed model allows for the influence of individual-specific covariates on clustering, to make individuals with similar covariate values more likely to co-cluster {\it a priori}. To this effect, we assume a stick-breaking representation of covariate-dependent mixing weights, as discussed by \cite{rigon2021}, which allows for exact sampling of the mixing weights through the P{\'o}lya-Gamma data augmentation method proposed by \cite{polson2013}. We refer to the proposed model as MFRMMx, which abbreviates \textit{Multivariate Functional Repulsive Mixture Model with covariates}. Model estimation is performed through a Markov Chain Monte Carlo (MCMC) algorithm that includes a novel split-merge step, which is an adaptation of the original method introduced by \cite{jain2007} to the case of FMMs with a repulsive dependency between the mixture components. As far as we know, the split-merge method had never been considered in clustering models with a dependency structure between clusters.

% As far as we know, the split-merge method had never been used in functional clustering models or clustering models that assume dependence between clusters.

\subsection{Article organization}

This work is organized as follows. Section~\ref{sec:model} details the proposed repulsive FMM for multivariate functional data. The MCMC strategy for posterior inference of the model is discussed in Section~\ref{sec:posterior}. In Section~\ref{sec:simulation}, we evaluate the effects of repulsion, covariates and other model settings in the clustering, through a simulation study. Section~\ref{sec:real} discusses the application of the proposed model to the same data set of individuals with CAI analyzed by \cite{hopkins2019}. Our main conclusions and some ideas for future research related to our work are discussed in Section~\ref{sec:final}. Additional results and a detailed description of the proposed MCMC algorithm for posterior simulation are provided in the Appendix.

\FloatBarrier
\section{The proposed MFRMM\lowercase{x} model}\label{sec:model}

Consider a set of $m$ individuals where each individual $i$, $i=1,\dots,m$, is represented by $D$ sequences of observations of common size $n_i$, denoted by the column vector $\bm{y}_{id}\!=\!(y_{id 1},\dots,y_{id n_i})^\T$, where $y_{idt_i}$ denotes the measurement for individual $i$ in each data sequence $d=1,\dots,D$, at times $t_i=1,\dots,n_i$. Denote by $\Yi$ the $n_i\tms D$ matrix with columns $\bm{y}_{id}$, that is, $\Yi=\left[\bm{y}_{i1}\,,\dots,\,\bm{y}_{iD}\right]$. We assume that $\Yi$ is modeled by
\begin{equation}\label{eq_yi}
	\Yi = \Bzi + H_i\Bi + \Ei, \quad i=1,\dots,m,
\end{equation}
\noindent where $H_i$ is the $n_i\tms p$ design matrix of a cubic B-spline curve with $p$ equally spaced inner knots \citep{fahrmeir2011}, $\Bi=\left[\bbt_{i1}\,,\dots,\,\bbt_{iD}\right]$ is a $p\,\tms\,D$ matrix which column vectors $\bbt_{id}$ are B-spline coefficients associated to the respective data sequences $\bm{y}_{id}$, and the matrix $\Bzi$ is equal to ${1}_{n_i}\bbt_{0i}^\T$, where ${1}_{n_i}$ denotes the $n_i \times 1$ column vector of $1$s, and $\bbt_{0i}=(\beta_{0i1},\dots,\beta_{0iD})^\T$. Then, $H_i\Bi$ is a $n_i\tms D$ matrix where each column $d$ is a smooth representation of the respective data sequence $\bm{y}_{id}$ in $\Yi$, and the elements $\beta_{0id}$ in $\bbt_{0i}$ account for the general level of the curves in $H_i\Bi$. Typically, the curves are misaligned, making knot selection difficult. To align the curves, we transform ``time'' into a unit interval, by considering transformed times $t_i^* = t_i/n_i$, for $t_i=1,\dots,n_i$, as proposed by \cite{page2015}. Thus, the set of B-spline knots will be selected such that the unit intervals are divided into equal subintervals, independently of $n_i$. The error term $\Ei=\left[\bm{\epsilon}_{i1},\dots,\bm{\epsilon}_{iD}\right]$ is a $n_i\tms D$ matrix with column vectors  $\bm{\epsilon}_{id}=(\epsilon_{id 1},\dots,\epsilon_{id n_i})^\T$, $d=1\dots,D$. We assume the $\Ei$'s conditionally independent and, to account for dependence among the $D$ curves of each individual, each $\Ei$ is assumed to follow a matrix Normal distribution with a $n_i\tms D$ zero mean matrix and covariance matrix $\Si\otimes I_{n_i}$, where $\Si$ is a $D\tms D$ covariance matrix with entries $\sigma_{i,dd'}$, $d,\!d'\!=\!1,\dots,D$, $I_{n_i}$ is the $n_i\tms n_i$ identity matrix, and $\otimes$ is the Kronecker product. For $d\!=d'$, $\sigma_{i,dd'}$ is the variance $\sigma^2_{id}$. Under these assumptions, the density of the matrix $\Ei$ is 

\begin{equation}\label{eq_Ei}
	f(\Ei) = (2\pi)^{-\frac{n_i D}{2}} |\Si|^{-\frac{n_i}{2}}
	\exp\left\{-\frac{1}{2}\,\tr(\Si^{-1} \Ei^\T\Ei)\right\},
	\quad i=1,\dots,m.
\end{equation}

\noindent Let $\vect(\Ei)$ be the vector of length $n_iD$ resulting from stacking the columns of $\Ei$. It follows that $\vect(\Ei)$ has a $n_iD$-variate Normal distribution with zero mean vector of dimension $n_iD$ and covariance matrix $\Si\otimes I_{n_i}$. These model assumptions imply that the likelihood function is simply a product of Normal densities. Conditionally on $\Bzi$ and $\Bi$, the covariance matrix $\Si\otimes I_{n_i}$ implies the following dependence structure in each observational unit $\Yi$, $i\!=\!1,\dots,m$: for all $\,t,t'\in\{1,\dots,n_i\}\,$ and $\,d,d'\in\{1,\dots,D\}$, we have ${\text{Cov}(y_{idt}\,,y_{id't'}\!)\!=\!\sigma_{i,dd'}}$ if $t\!=\!t'$ and, ${\text{Cov}(y_{idt}\,,y_{id't'}\!)\!=\!0}$ if $t\!\ne\!t'$. Thus, observations at different times are conditionally independent, while the observations of the $D$ sequences at each time $t$ have a common covariance structure  $\Si$ for all $t=1,\dots,n_i$. 

We assume, \textit{a priori}, that $\Si\overset{\text{iid}}{\sim}\IW(\Sigma_0,\omega)$ and $\beta_{0id}\overset{\text{iid}}{\sim}\N(\mu_0,\sigma_0^2)$, $i=1,\dots,m$, $d=1,\dots,D$. To identify clusters of curves with similar shapes, we propose a finite mixture distribution for the B-spline coefficients in $\Bi$.
For $d=1,\dots,D$ and $j=1,\dots,J$, let $\Theta_d=\{\bth_{1d},\dots,\bth_{Jd}\}$ be the set of column vectors $\bth_{jd}\in\mathbb{R}^p$, $\blam_d=\{\lambda^2_{1d},\dots,\lambda^2_{Jd}\}$ be the set of variances $\lambda^2_{jd}$, and $\bm{\pi}(\bm{x}_i)=(\pi_1(\bm{x}_i),\dots,\pi_J(\bm{x}_i))$ be the vector of covariate-dependent mixing weights $\pi_j(\bm{x}_i)>0$, with $\sum_{j=1}^{J}\pi_j(\bm{x}_i)=1$, where $\bm{x}_i$ is a $\Ell\tms1$ vector of covariates of individual $i$. Conditionally on $\bTh=\{\Theta_1,\dots,\Theta_D\}$, $\bLam=\{\blam_1,\dots,\blam_D\}$ and $\bm{\pi}(\bm{x}_i)$, the matrices $\Bi$, $i=1,\dots,m$, are independent and identically distributed, and the vectors $\bbt_{id}$ in $\Bi$ are independent across dimensions $d$, with prior distribution given by

%\begin{equation}\label{prior_beta}
%	\begin{aligned}
	%		&f(\bbt_{i1},\dots,\bbt_{iD} \v \bTh,\bLam,\bm{\pi}(\bm{x}_i))\\
	%		%\,\overset{\text{iid}}{=}\, \sum_{j=1}^{J}
	%		&\,=\, \sum_{j=1}^{J}
	%		\pi_j(\bm{x}_i)
	%		\prod_{d=1}^{D}\N_p(\bbt_{id}\,;\, \bth_{jd},\lambda^{2}_{jd}I_p),
	%		%\quad i=1,\dots,m,
	%	\end{aligned}
%\end{equation} 

\begin{equation}\label{prior_beta}
	f(\bbt_{i1},\dots,\bbt_{iD} \v \bTh,\bLam,\bm{\pi}(\bm{x}_i))
	\,=\,
	\sum_{j=1}^{J}
	\pi_j(\bm{x}_i) \prod_{d=1}^{D}\N_p(\bbt_{id}\,;\, \bth_{jd},\lambda^{2}_{jd}I_p),
	\quad i=1,\dots,m,
\end{equation} 

\noindent where $J$ is a known (fixed) number of mixture components, $\prod_{d=1}^{D}\N_p(\bbt_{id}\,;\, \bth_{jd},\lambda^{2}_{jd}I_p)$ is the $j$th component density and $\N_p(\bbt_{id}\,;\, \bth_{jd},\lambda^{2}_{jd}I_p)$ denotes the density of a $p$-variate Normal distribution with mean vector $\bth_{jd}$ and covariance matrix $\lambda^{2}_{jd}I_p$, evaluated at $\bbt_{id}$. Given \eqref{prior_beta}, the prior distribution for $\bbt_{id}$, for $i=1,\dots,m$ and $d=1,\dots,D$, may be hierarchically represented as
\begin{equation}\label{prior_beta_z}
	\begin{aligned}
		\bbt_{id} \mid z_i, \bth_{z_i d}, \lambda^2_{z_i d}
		&\,\overset{\text{ind}}{\sim}\,
		\N_p(\bth_{z_i d} \,, \lambda^2_{z_i d}I_p)\\
		z_i \v \bm{\pi}(\bm{x}_i)
		&\,\overset{\text{ind}}{\sim}\, \Pr(z_i\!=\!j\v\bm{x}_i)\!=\!\pi_j(\bm{x}_i),
	\end{aligned}
\end{equation}
\noindent where $z_i$ is a latent categorical variable assuming the value $j$ if unit $i$ is allocated to the mixture component $j$. This way, probabilistic cluster identification is obtained through the posterior distribution of $\bz\!=\!(z_1,\dots,z_m)$. The model assumptions in \eqref{prior_beta_z} imply that the curves $H_i\bbt_{id}$ of individuals belonging to the same cluster $j$ are characterized by varying around a common cluster-specific curve $H\bth_{jd}$, in every dimension $d=1,\dots,D$. As discussed in Section~\ref{sec:intro}, the location parameters of mixture component densities are commonly assumed to be independent, which does not encourage separation between mixture components, thus leading to cluster redundancy and overestimation of the number of clusters \citep{xu2016,bianchini2020,quinlan2021}. The proposed model innovates the analysis of functional data by assuming a repulsive joint prior to the component location vectors in $\Theta_d$, for every dimension $d=1,\dots,D$. The proposed prior is based on the repulsive strategy introduced by \cite{petralia2012}, and its purpose is to favor the identification of well-differentiated clusters in the context of multivariate functional curves, as will be described in \ref{sec:repulsion}. For the component-specific variance parameters $\lambda^{2}_{jd}$, we assume, \textit{a priori}, that
\begin{equation}\label{prior_lam2}
	\sqrt{\lambda^{2}_{jd}} \;\overset{\text{iid}}{\sim}\; \U(0,A_d),
	\quad j=1,\dots,J,\;\;d=1,\dots,D.
\end{equation}
\noindent The hyperparameter $A_d$ truncates the variability of the curves within each cluster, allowing for more (less) heterogeneous curves to co-cluster for higher (smaller) values of $A_d$.

To consider the influence of covariates on clustering, we assume a covariate-dependent stick-breaking representation of $\pi_j(\bm{x}_i)$ as in \cite{rigon2021}, that is,

\begin{equation}\label{EqLT}
	\pi_1(\bm{x}_i) = \nu_1(\bm{x}_i),\;\;\;\;\;\;
	\pi_j(\bm{x}_i) = \nu_j(\bm{x}_i)\prod_{\ell=1}^{j-1}\left(1\!-\!\nu_\ell(\bm{x}_i)\right),\, j\!=\!2,\dots,J-1,
\end{equation}

%\begin{equation}\label{EqLT}
%	\begin{aligned}
	%		&\pi_1(\bm{x}_i) = \nu_1(\bm{x}_i),\nonumber\\
	%		&\pi_j(\bm{x}_i) = %\nu_j(\bm{x}_i)\prod_{\ell=1}^{j-1}\left(1\!-\!\nu_\ell(\bm%{x}_i)\right),\, j\!=\!2,\dots,J-1,
	%	\end{aligned}
%\end{equation}
%
\noindent where $\nu_j(\bm{x}_i)=\Pr(z_i=j\mid z_i\ge j,\,\bm{x}_i)$, which is the conditional probability of the individual $i$ being allocated to mixture component $j$, given that it was not allocated to any of the preceding components ${\ell=1,\dots,j-1}$, and $\nu_J(\bm{x}_i)=1$, which restricts the original stick-breaking prior \citep{ishwaran2001} to a FMM. This strategy allows for exact sampling of $\pi_j(\bm{x}_i)$, based on the P{\'o}lya-Gamma data augmentation method \citep{polson2013}. Details are given in Section~\ref{supp:covariates} of the Appendix.
Specifically, we assume that $\pi_j(x_i)$ follows a stick-breaking construction with a logit transformation of the conditional probabilities $\nu_j(\bm{x}_i)$, given by $\bm{x}^\T_i\ba_j=\log\big(\nu_j(\bm{x}_i)/(1-\nu_j(\bm{x}_i))\big),$ where $\bm{x}_i$ are covariates related to the weights of individual $i$ and  $\ba_j$ is the $\Ell \times 1$ component-specific column vector of coefficients $\ba_j$, $i=1,\dots,m$ and $j=1,\dots,J$. Under this structure, the mixing weights are given by

\begin{equation}\label{def_weights_x}
	\pi_j(\bm{x}_i \v \ba_j)
	=
	\frac{\exp\Big(\bm{x}^\T_i\ba_j\Big)}{1+\exp\Big(\bm{x}^\T_i\ba_j\Big)}
	\prod_{\ell=1}^{j-1}\frac{1}{1+\exp\Big(\bm{x}^\T_i\ba_\ell\Big)}.
\end{equation}

\noindent The prior distribution of the covariate-dependent mixing weights $\pi_j(\bm{x}_i \v \ba_j)$ is specified indirectly  by eliciting  the prior distribution  for $\ba_j$. We assume that $\ba_j\overset{\text{ind}}{\sim} \N_\Ell(\bmu_\alpha,\Sigma_\alpha)$, for $j=1,\dots,J$. If covariates are judged unimportant to cluster identification, we assume independent mixing weights $\bm{\pi}=(\pi_1,\dots,\pi_J)$ to define ${\Pr(z_i=j\mid\bm{\pi})=\pi_j}$, for $j=1,\dots,J$. As discussed by \cite{rousseau2011}, assuming the number of components $J$ sufficiently large and a Dirichlet prior distribution for $\bm{\pi}=(\pi_1,\dots,\pi_J)$ with small parameters favors the elimination of empty extra components. In this case, we assume $\bm{\pi}\sim\Dir(1/J,\dots,1/J)$, as suggested in \cite{bda3}.

\FloatBarrier
\subsection{Repulsive prior density for multivariate functional curves}
\label{sec:repulsion}

To favor identification of well-differentiated clusters with respect to curve shapes, we introduce the following joint repulsive prior for the component location vectors in ${\Theta_d=\{\bth_{1d},\dots,\bth_{Jd}\}}$ and component-specific variances ${\btau_d=\{\tau^2_{1d},\dots,\tau^2_{Jd}\}}$, ${d=1,\dots,D}$:
%
%\begin{equation}\label{prior_Theta_d}
%	\begin{aligned}
	%		&f(\Theta_d,\btau_d \v \bm{\mu}_d) =\\
	%		&\hspace{-.1in}\frac{1}{\mathrm{C}_{Jd}}\!
	%		\left(
	%		\prod_{j=1}^{J}
	%		\!\N_p(\bth_{jd} \,; \bm{\mu}_d,\tau^2_{jd} K^{-1})
	%		\IG(\tau^2_{jd} \,; a_\tau,b_\tau)
	%		\!\!\right)
	%		\!h(\Theta_d),
	%	\end{aligned}
%\end{equation}
%
\begin{equation}\label{prior_Theta_d}
	f(\Theta_d,\btau_d \v \bm{\mu}_d)
	\;=\;
	\frac{1}{\mathrm{C}_{Jd}}
	\bigg(
	\prod_{j=1}^{J}
	\N_p(\bth_{jd} \,; \bm{\mu}_d,\tau^2_{jd} K^{-1})
	\,\IG(\tau^2_{jd} \,; a_\tau,b_\tau)
	\bigg)
	h(\Theta_d),
	%\quad d=1,\dots,D,
\end{equation}
\noindent where $C_{Jd}$ is the normalizing constant obtained by the integral
\begin{equation}\label{eq_CJd}
	%\bigintssss\!\!
	\int\!
	\bigg(
	\prod_{j=1}^{J}
	\N_p(\bth_{jd} \,; \bm{\mu}_d,\tau^2_{jd} K^{-1})
	\,\IG(\tau^2_{jd} \,; a_\tau,b_\tau)
	\bigg)
	h(\Theta_d)\d\btau_d\d\Theta_d,
\end{equation}
\noindent and $h(\Theta_d)$ is the functional repulsive factor, defined by
\begin{equation}\label{def_h} 
	h(\Theta_d) = \prod_{(j,\ell)\,\in\,\J} \exp\left(-\phi_d\,\dist(\bth_{jd},\bth_{\ell d})^{-\nu}\right),
	%\quad d=1,\dots,D,
\end{equation}
where $\J=\{(j,\ell):j,\ell=1,\dots,J,\,\ell<j\}$, $\dist(\bth_{jd},\bth_{\ell d})$ is a distance between the component-specific curves determined by the coefficients $\bth_{jd}$ and $\bth_{\ell d}$, and $\phi_d\in\mathbb{R}^{+}$ and $\nu\in\mathbb{N}$ are hyperparameters that calibrate the repulsion intensity. We note that other forms of the $h$ function in \eqref{def_h} can be equally adopted; see the discussion in, e.g., \cite{quinlan2021} and \cite{xie2020}.

Different prior distributions may be obtained by assuming distinct specifications of $\dist$. The distance between two continuous functions $f_j$ and $f_\ell$ is commonly measured by the $\text{L}_q$-norm $\left(\int\vert f_j(t)-f_\ell(t)\vert^q\, \d t\right)^{1/q}$. Since the B-spline curves are discretized representations of continuous functions, we consider the following approximation of the $\text{L}_q$-norm to specify $\dist$:
\begin{equation}\label{eq_distance_approx}
	\dist(\bth_{jd},\bth_{\ell d})
	= \left(\frac{1}{n}\sum_{t=1}^{n}\big\vert H\bth_{jd}(t)-H\bth_{\ell d}(t)\big\vert^q \right)^{1/q},
\end{equation}
\noindent for $(j,\ell)\in\J$ and $d=1,\dots,D$, where $H\bth_{jd}(t)$ and $H\bth_{\ell d}(t)$ are the values of the component-specific B-spline curves determined by a common design matrix $H$, of dimension $n\tms p$, and coefficients $\bth_{jd}$ and $\bth_{\ell d}$, evaluated at coordinate $t$. The length $n$ of the component-specific curves and the value of $q$ are preset. The choice of $n$ should be based on the length of the sequences in each application. It can be set, for example, to the maximum value in \{$n_1,\dots,n_m$\}. Higher values of $n$ produce more accurate approximations of the $\text{L}_q$-norm, but increase computational cost. A distance based on the vectors $\bth_{jd}$, rather than the funcional curves $H\bth_{jd}$, would also produce a repulsive model. However, the distance based on the functional curves provides a more precise differentiation between the curve shapes.

%Since the proposed model is defined to identify groups of functional curves based on shape similarity, the estimation of two or more cluster-specific curves that have very similar shapes would characterize redundancy, regardless of their vertical level distance. Unlike the design matrix $H_i$ defining the individual-specific B-spline curves in \eqref{eq_yi}, the design matrix $H$ considered to compute the distance in \eqref{eq_distance_approx} is standardized to make $\sum_{i=1}^{n} H\bth_{jd}(t_i)=0$. Consequently, the distance defined in \eqref{eq_distance_approx} ignores the vertical level distance between the curves $H\bth_{jd}$ and $H\bth_{\ell d}$ and represents only the difference in shape between them.

The repulsive factor $h(\Theta_d)$ penalizes the density \eqref{prior_Theta_d} at shorter distances between components in every dimension $d=1,\dots,D$. Unlike $H_i$, used to define individual curves according to \eqref{eq_yi}, the design matrix $H$ considered to compute distances using \eqref{eq_distance_approx} is standardized to make $\sum_{i=1}^{n} H\bth_{jd}(t_i)=0$. Therefore, that distance does not take into account the vertical level of the component-specific mean curves $H\bth_{jd}$ and represents only a difference between the shapes of the curves.

The hyperparameters $\phi_d$, $d=1,\dots,D$, and $\nu$ are fixed at known values. Each dimension $d$ has an independent repulsive factor $h(\Theta_d)$, with a specific parameter $\phi_d$. Specifying different values of $\phi_d$ across the dimensions $d$ can be a useful strategy, depending on the application. Setting $\phi_d=0$ implies that the component-specific mean vectors $\bth_{jd}$ are conditionally independent, given $\bmu$ and $\tau^2_{jd}$, and have a Normal prior distribution $\N_p(\bth_{jd}\,;\,\bm{\mu}_d,\tau^2_{jd} K^{-1})$. We further assume that the mean vectors $\bm{\mu}_d$ have independent normal priors, $\bm{\mu}_d\overset{\text{iid}}{\sim}\N_p(\bm{0},s^2_\mu I_p)$, $d=1,\dots,D$. Following \cite{page2015}, the variance matrix $K^{-1}$ follows from a first-order random walk assumption on the elements of $\bth_{jd}$, that is used to avoid overfitting when a moderate or large number of B-spline knots $p$ is chosen. Further details on this assumption are discussed in Section~\ref{supp:mfppmx} of the Appendix. It is also worth mentioning that the smoothness of the curves also influences clustering via the component-specific parameters $\tau^2_{jd}$'s. %We assume that $b_\tau\sim\text{G}(\xi,\varpi)$, where $E(b_\tau)=\xi/\varpi$.

Under the repulsive dependence between mixture components induced by $h(\Theta_d)$, the model loses the property of sample-size consistency \citep{quinlan2021}, which is verified if the elements in $\Theta_d$ are assumed independent. If the sample-size consistency is not verified, assuming a random or infinite $J$ imposes some challenges on model estimation. This motivates us to fix $J$ at a known, large enough value. Another cumbersome aspect of our model is the fact that the unknown constant $C_{Jd}$ in \eqref{prior_Theta_d} depends on the parameters $\phi_d$ and $\nu$, which makes estimation of these parameters non-trivial, justifying our strategy of fixing them to known values. Discussions on how to deal with this problem on non-functional repulsive mixture models can be found in \cite{beraha2022} and \cite{cremaschi2024}.

To guide the specification of $q$, $\phi_d$ and $\nu$, we evaluate the \textit{a priori} behavior of the pairwise distances between the component-specific mean curves through simulation. Assuming $n=30$, $a_\tau=1$, $b_\tau=1$, $\bmu=\bm{0}$, $J=10$ mixture components and $p=10$ inner knots defining the B-spline functions, we obtain values of $\dist$ defined in \eqref{eq_distance_approx} using values of $\Theta_d$ simulated from the repulsive prior distribution in \eqref{prior_Theta_d}.  Different combinations of $q$, $\phi_d$ and $\nu$ values are assumed. Considering the simulated values, Figure \ref{fig:priorTheta_prob3} shows the prior probability of the distance $\dist$ being greater than a value $\delta$, and Figure \ref{fig:priorTheta_dist_mean_p3} shows the impact of the values of $q$, $\phi_d$ and $\nu$ in the distance between component mean curves, \textit{a priori}.

\begin{figure}[h]%[t]
	\begin{minipage}[t]{.5\textwidth}
		\begin{adjustwidth}{-1.5cm}{-1cm}
		\centering
		\hspace{-.2in}\stackunder[1pt]{
			\raisebox{3.9cm}{\rotatebox[origin=c]{90}{\scriptsize $\Pr\big(\dist\!>\!\delta\big)$}}
			\includegraphics[width=8.0cm, height=8.0cm, trim=0.2cm 0 0 0]{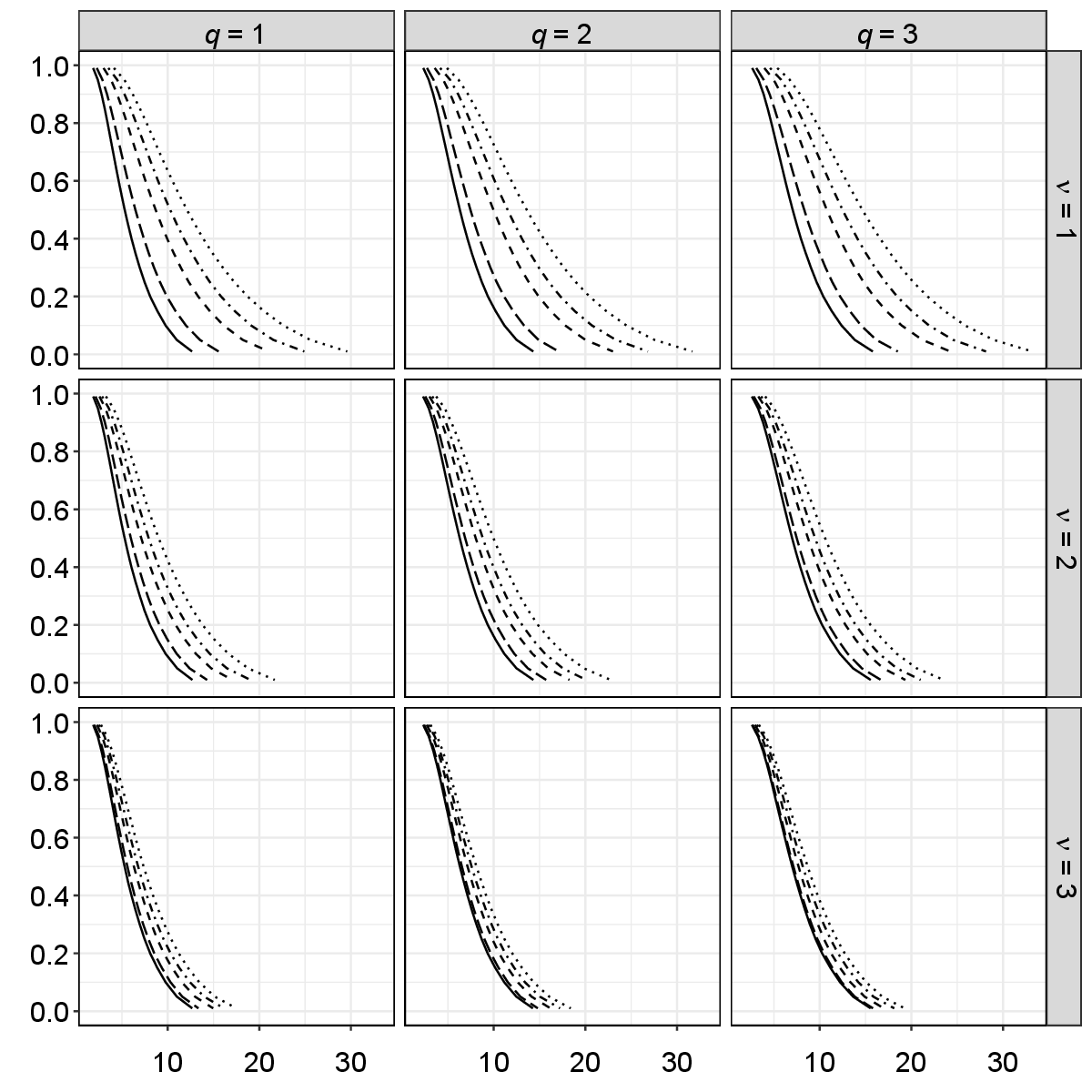}}
		{\scriptsize \hspace{.25in}$\delta$}
		\end{adjustwidth}
		\vspace{-.1in}
		\caption{Probability of the mean distance between component mean curves being greater than a value $\delta$, simulated \textit{a priori} for different values of $q$ and $\nu$, for $\phi\!=\!0$ (solid line), $\phi\!=\!1$ (long dashed line), $\phi\!=\!5$ (dashed line), $\phi\!=\!10$ (dash-dotted line) and $\phi\!=\!20$ (dotted line).}
		\label{fig:priorTheta_prob3}
	\end{minipage}
	%\hspace{0.1cm}
	\begin{minipage}[t]{.5\textwidth}
		\begin{adjustwidth}{-1.0cm}{}
		\centering
		\hspace{-.2in}\stackunder[1pt]{
			\raisebox{3.9cm}{\rotatebox[origin=c]{90}{\scriptsize Mean distance}}
			\includegraphics[width=8.0cm, height=8.0cm, trim=0.2cm 0 0 0]{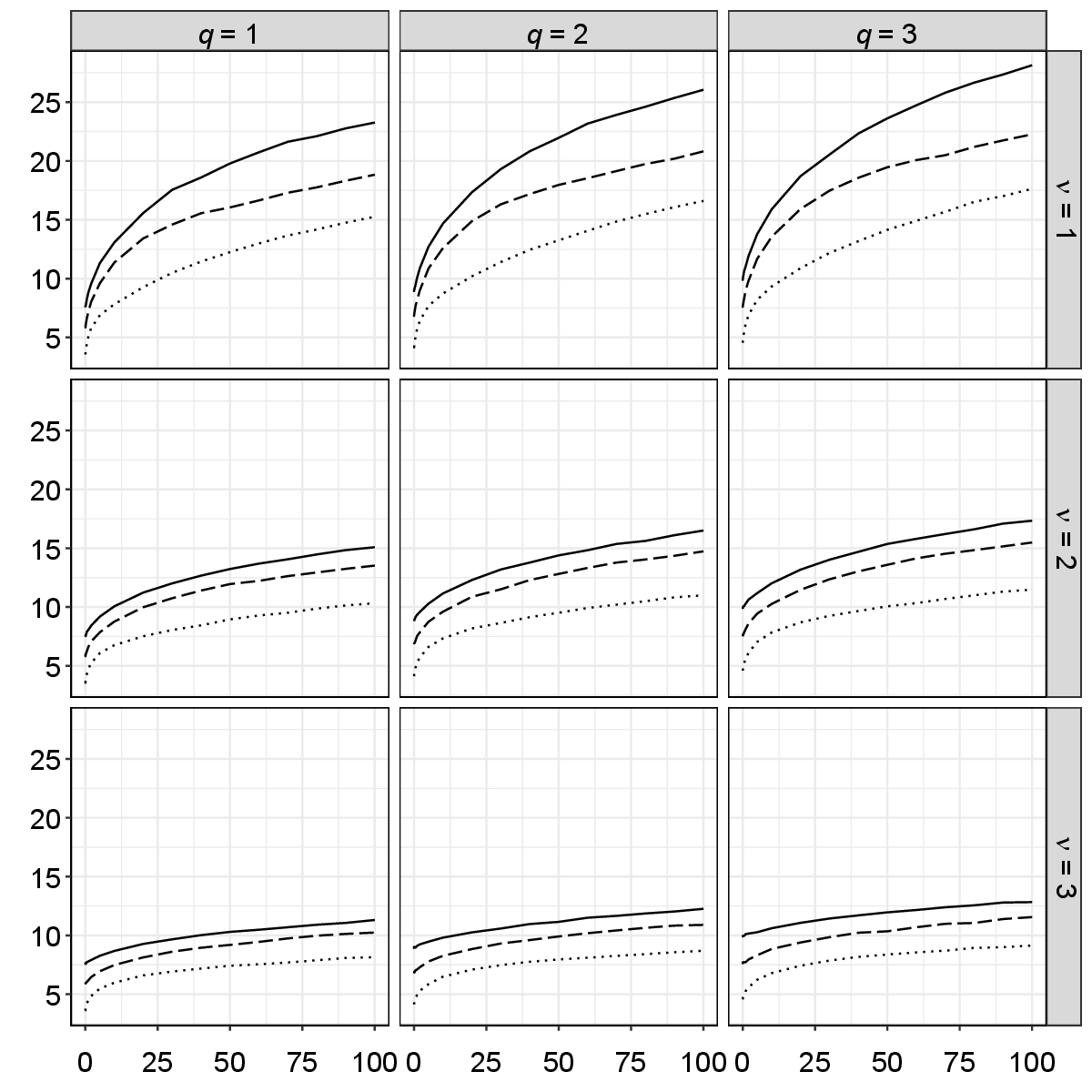}}
		{\scriptsize \hspace{.25in}$\phi$}
		\end{adjustwidth}
		\begin{adjustwidth}{-0.5cm}{}
		\centering
		\vspace{-.1in}
		\caption{Mean distance between component mean curves simulated \textit{a priori} for different values of $q$, $\nu$ and $\phi$, for $p\!=\!5$ (dotted line), $p\!=\!10$ (dashed line) and $p\!=\!15$ (solid line).}
		\label{fig:priorTheta_dist_mean_p3}
		\end{adjustwidth}
	\end{minipage}
\end{figure}

Based on the probability values shown in Figure \ref{fig:priorTheta_prob3}, we can calibrate the repulsion by choosing values of $q$, $\phi_d$ and $\nu$ that imply a high probability that $\dist$ will be greater than a distance $\delta$ that would characterize redundant clusters based on the opinion of experts in each specific application. For example, if distances between cluster mean curves lower than $10$ in one specific dimension $d$ were considered a redundancy problem, we could assume, \textit{a priori}, a high probability $\Pr(\dist>10\big)=0.90$ by assuming $\phi_d=5$ with $\nu=1$ and $q=2$. That same probability would also be obtained by assuming $\phi_d=10$ with $\nu=1$ and $q=1$. We can see from Figure \ref{fig:priorTheta_dist_mean_p3} that higher values of $q$ and $\phi_d$, and lower values of $\nu$, increase the repulsion, producing greater differentiation between the component mean curves. This conclusion could also be derived analytically from \eqref{def_h} and \eqref{eq_distance_approx}. We can also observe that higher number of knots $p$ implies higher differentiation between the component mean curves. In the applications of the proposed model to simulated and real data applications discussed in the following sections, we assume $q=2$ to define the $\text{L}_q$-norm approximation in \eqref{eq_distance_approx}. Regarding the hyperparameters $\phi_d$ and $\nu$, we follow the recommendation of \cite{petralia2012} to keep the value of $\nu$ constant and vary the repulsion intensity by varying the value of $\phi_d$. In particular, we assume $\nu=1$ and increase (or reduce) the repulsion intensity by increasing (or reducing) the value of $\phi_d$.

\FloatBarrier
\section{Posterior inference}\label{sec:posterior}

The proposed model is hierarchically represented as follows:

\begin{equation}\label{eq_model}
	\begin{aligned}
		\vect(\Yi) \v \Bzi,\Bi,\Si
		&\,\overset{\text{ind}}{\sim}\,
		\N_{n_i D}(\vect(\Bzi \!+\! H_i\Bi), \Si\!\otimes\!I_{n_i}),
		\\%[-.05in]
		\Si
		&\;\overset{\text{iid}}{\sim}\;
		\IW(\Sigma_0,\omega),
		\\%[-.05in]
		\beta_{0id} \v \mu_{0d},\sigma_{0d}^{2}
		&\,\overset{\text{ind}}{\sim}\,
		\N(\mu_{0d},\sigma_{0d}^{2}),
		\\%[-.05in]
		\bbt_{id} \v z_i,\bth_{z_i d},\lambda^{2}_{z_i d}
		&\,\overset{\text{ind}}{\sim}\,
		\N_p(\bth_{z_i d},\lambda^{2}_{z_i d} I_p),
		\\[.05in]
		\sqrt{\lambda^{2}_{z_i d}}
		&\;\overset{\text{iid}}{\sim}\;
		\U(0,A_d),
		\\[-.03in]
		f(\Theta_d\,,\btau_d \v \bm{\mu}_d)
		%\;\propto\;
		&\;\overset{\text{ind}}{\propto}\;
		h(\Theta_d)
		\prod_{j=1}^{J}
		\N_p(\bth_{jd} \,; \bm{\mu}_d,\tau^2_{jd} K^{-1})\,\IG(\tau^2_{jd} \,; a_\tau,b_\tau),
		\\%[-.05in]
		\bm{\mu}_d
		&\;\overset{\text{iid}}{\sim}\;
		\N_p(\bm{0},s^2_\mu I_p),
		\\[.02in]
		%	%
		%	\tau^2_{jd}
		%	&\overset{\text{iid}}{\sim}&
		%	\IG(a_\tau,b_\tau)
		%	\\[-.05in]
		%			%
		\Pr(z_i=j \v \bm{x}_i , \ba_j)
		&\;=\;
		\pi_j(\bm{x}_i \v \ba_j),
		\\[.02in]
		\ba_j
		&\;\overset{\text{iid}}{\sim}\;
		\N_{\Ell}(\bm{\mu}_\alpha,\Sigma_\alpha),
		\\%[-.05in]
		%
		%&\hspace{-1in}\text{for $i=1,\dots,m$, $j=1,\dots,J$ and $d=1,\dots,D$.}
		%\\%[.03in]
	\end{aligned}
\end{equation}

\noindent for ${i=1,\dots,m}$, ${j=1,\dots,J}$ and ${d=1,\dots,D}$. If covariates are not considered, the prior distribution of $z_i$ in \eqref{eq_model} is replaced by ${\Pr(z_i=j\v\bm{\pi})=\pi_j}$, with $\bm{\pi}\sim\Dir(1/J,\dots,1/J)$. In addition, the model is fully specified by the hyperpriors ${\mu_{0d}\overset{\text{iid}}{\sim}\N(0,s^2_0)}$ and  ${\sigma_{0d}^{2}\overset{\text{iid}}{\sim}\IG\left(a_0,b_0\right)}$, for $d=1,\dots,D$, and ${b_\tau\sim\text{G}(\xi,\varpi)}$. The closed form of the joint posterior of our model is not analytically obtained; see Equation \eqref{eq_posterior} in the Appendix. A suitable MCMC algorithm to sample from the posterior distribution is detailed in Section \ref{supp:posterior} of the Appendix. Some specific features regarding the proposed MCMC algorithm are briefly discussed next. The full conditional distribution of the variance parameters $\lambda^{2}_{jd}$ is a Inverse-Gamma truncated in the interval $(0,A_d^2)$, from which we sample using the Adaptive Rejection Sampling method (ARS) introduced by \cite{gilks1992}. Sampling from the component-specific mean vectors $\bth_{jd}$ is done jointly for all dimensions $d=1,\dots,D$, through the Metropolis-Hastings acceptance-rejection method \cite{metropolis1953,hastings1970}, by assuming independent Normal proposal distributions for each element $\bth_{jd}$, $d=1,\dots,D$. The proposal variance is adapted during an initial burn-in period, in order to achieve appropriate acceptance rates for all the mixture components $j=1,\dots,J$, since an empirical search for a suitable proposal variance could be very costly, especially when dealing with large data dimensions $D$. The full conditional distributions of the component-specific mean vectors $\bth_{jd}$ and the allocation vector $\bm{z}$ are marginalized with respect to $\mathcal{B}_1,\dots,\mathcal{B}_m$, which provided better mixing in the posterior chain. An \texttt{R} package containing the implementation of the MCMC algorithm is available at \url{https://github.com/rcpedroso/MHRMMx}. The algorithm is fully written in \texttt{c++}, which significantly speeds up execution. One important feature of our MCMC algorithm is the presence of a split-merge step that has proved essential for more accurate cluster identification for larger values of the dimension $D$. We discuss this split-merge step in the next section.

\FloatBarrier
\subsection{The split-merge step}

As discussed in \cite{jain2007}, including a split-merge step, in addition to the sequential update of the cluster allocation one individual at a time, prevents the chain from getting stuck in local modes, which frequently occurs when clustering high-dimensional data. While the split-merge step makes large modifications to the cluster allocation, sequentially moving one individual at a time provides small refinements to the clustering. Applications of the split-merge method can be found in \cite{martinez2014}, \cite{bouchard2017} and \cite{miller2018}, among others. To the best of our knowledge, the use of this method for functional clustering or for mixture models with dependent components has not been proposed previously, thus constituting an important contribution of this work.

At each MCMC iteration, a split of one cluster in two or a merge of two clusters in one is proposed, and is accepted or not based on the acceptance probability

\begin{equation}\label{eq_SM_acceptance}
	a(\gamma^\star,\gamma)
	\,=\, \min
	\left[1\,,\,
	\frac{ q(\gamma\v\gamma^\star) }{ q(\gamma^\star\v\gamma) }\,
	\frac{ P(\gamma^\star) }{ P(\gamma) }\,
	\frac{ L(\gamma^\star) }{ L(\gamma) }
	\right],
\end{equation}

\noindent where $\gamma$ and $\gamma^\star$ denote, respectively, the current and proposal states of the cluster allocation and respective cluster-specific parameters. Functions $P$ and $L$ represent the prior and the likelihood functions of the model, respectively, and function $q$ is determined by the Metropolis-Hastings proposal densities. The split-merge method was originally proposed for the nonconjugate Dirichlet process mixture model \citep{ferguson1983}, which assumes the parameters of different clusters are independent, \textit{a priori}. In that case, the ratio $P(\gamma^\star)/P(\gamma)$ in \eqref{eq_SM_acceptance} depends only on the parameters of those clusters being split or merged, and the proposal distributions $q$ are the full conditionals of the respective cluster parameters. Our split-merge procedure differs from the original one in two main aspects. First, the ratio $P(\gamma^\star)/P(\gamma)$ also depends on parameters of mixture components that are not being updated, which includes empty components. This is due to the repulsive dependence between different components, implied by the repulsive factors $h(\Theta_d)$, $d=1,\dots,D$, which appear in the numerator and denominator of this ratio; see Equations \eqref{prior_split_rate} and \eqref{prior_merge_rate} in the Appendix. Second, we do not consider the full conditional distributions of all component-specific parameters to create the update proposal. The $\tau^2_{jd}$ and $\lambda^2_{jd}$ variances are proposed from their corresponding full conditionals, but the elements in $\Theta_d$ are proposed from the $p$-variate Normal distribution obtained when assuming $h(\Theta_d)\!=\!1$ in the full conditional distribution in \eqref{full_theta_dep}. Otherwise, the proposal distribution of $\bth_{jd}$ would not have a closed form, making the split-merge step very complex to carry out. The algorithm of our split-merge procedure is described in detail in Section \ref{supp:SM} of the Appendix. %A \texttt{C++} implementation of the full MCMC algorithm is available in an \texttt{R} package at \url{https://github.com/rcpedroso/MFRMMx}.

\FloatBarrier
\section{Simulation study}\label{sec:simulation}

In this section, we discuss the application of the proposed model to simulated data, evaluating the performance of the proposed model under various levels of repulsion and the influence of individual-specific covariates on clustering. For comparison purposes, we also consider a multivariate version of the (non-repulsive) functional clustering model introduced by \cite{page2015}, which is described in Section~\ref{supp:mfppmx} of the Appendix. We will refer to their model as MFPPMx (which abbreviates \textit{Multivariate Functional Product Partition Model with covariates}).

We consider 25 data sets of $m=40$ individuals, with each individual being represented by a multivariate data sequence with $D=4$. For each individual $i=1,\dots,m$ and dimension $d=1,\dots,D$, we generated data sequences $\byid$ of common size $n_i=30$, as follows. First, we simulated four different cluster-specific means $\bth_{jd}$, $j=1,\dots,4$, from \eqref{prior_Theta_d}, with $p=4$ knots, precision matrix $K=I_p$, $\bmu_d=\bm{0}$, $\tau^2_{jd}=200$, and assuming no repulsion by letting $h(\Theta_d)=1$. Then, we generated individual-specific curves by simulating $\Bi$ from \eqref{prior_beta}, centered in the respective simulated cluster-specific means $\bth_{jd}$, such that two clusters have 8 individuals and the other two clusters have 12 individuals, assuming $\lambda^2_{jd}=100$, for all $j$. Finally, the curves $H_i\Bi$ were considered to generate correlated data sequences $\bm{y}_{i1},\dots,\bm{y}_{i4}$ through \eqref{eq_yi}, by simulating $\beta_{0id}\overset{\text{iid}}{\sim}\N(10,4)$ and $\vect(\Ei)\overset{\text{iid}}{\sim}N_{n_iD}(\bm{0},\Si\otimes I_{n_i})$. We assume the variance $\sigma^2_{i,d}\overset{\text{iid}}{\sim}\U(0.1,4)$, $d=1,\dots,D$,  a positive correlation 0.9 between observations $y_{i1t}$ and $y_{i2t}$, a negative correlation $-0.9$ between observations $y_{i1t}$ and $y_{i3t}$ and observations $y_{i2t}$ and $y_{i3t}$, for all $t=1,\dots,n_i$,  and no correlation between $\bm{y}_{i4}$ and the other data sequences. Figure \ref{fig:sim4_mean_curves} shows an example of simulated cluster-specific mean curves and Figure \ref{fig:sim4_all} shows examples of simulated individual-specific curves and data sequences of one cluster. Another view, with the individual-specific curves and data sequences of one simulated data set, is available in Section \ref{supp:sim4} of the Appendix.

\begin{figure*}[h]%[!htb]%[t]
	\begin{adjustwidth}{-0.6cm}{-0.6cm}
	\centering
	{\footnotesize$d=1$\hspace{3.2cm}$d=2$\hspace{3.2cm}$d=3$\hspace{3.2cm}$d=4$}\\[-.25in]
	\includegraphics[width=5cm, height=4.5cm, trim=0cm 0cm 0 0]{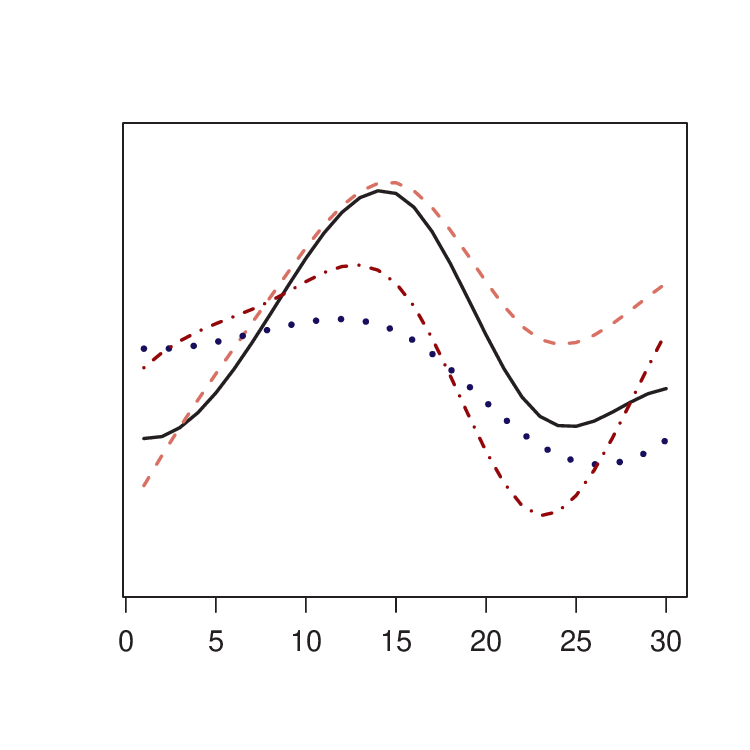}
	\label{fig_sim4_mean_curves_d1}\hspace{-.5in}
	\includegraphics[width=5cm, height=4.5cm, trim=0cm 0cm 0 0]{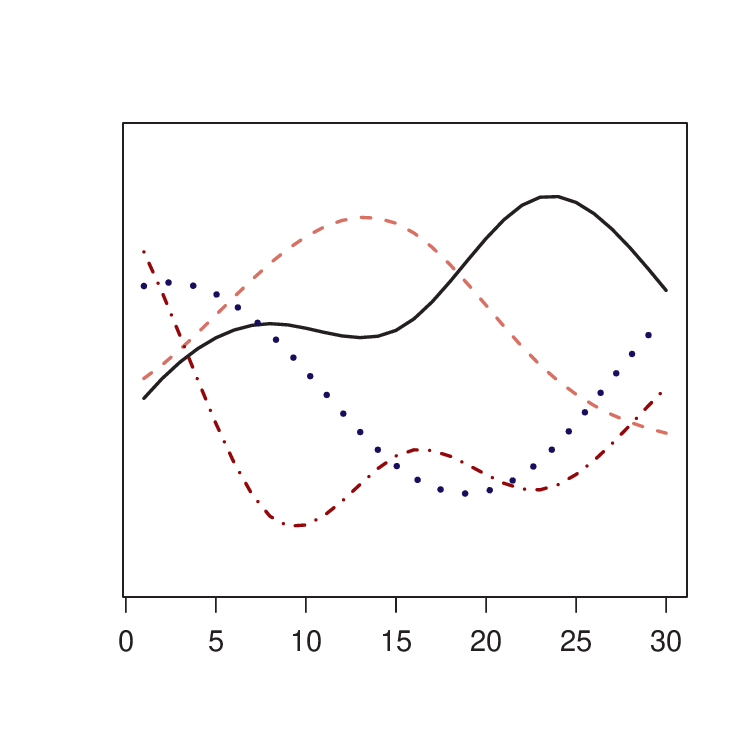}
	\label{fig_sim4_mean_curves_d2}\hspace{-.5in}
	\includegraphics[width=5cm, height=4.5cm, trim=0cm 0cm 0 0]{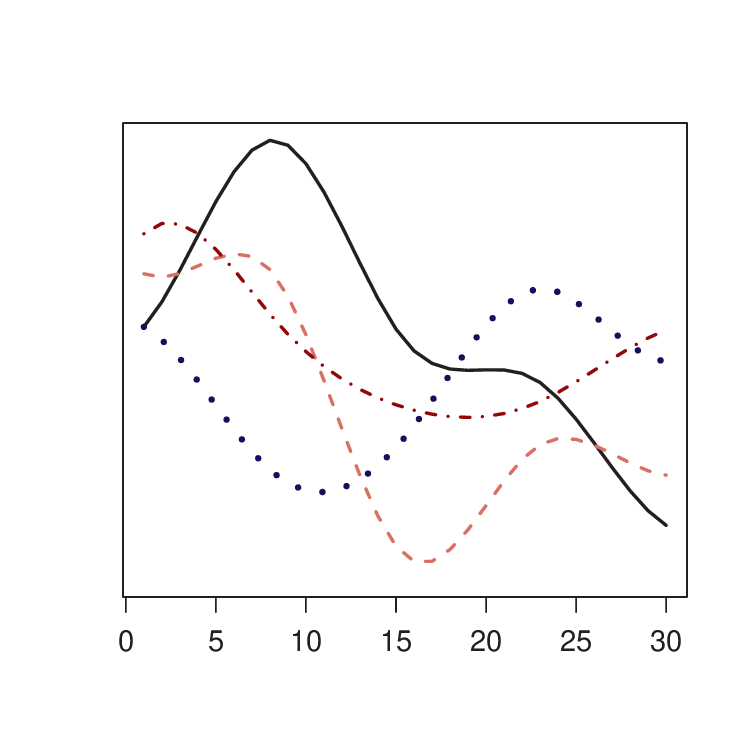}
	\label{fig_sim4_mean_curves_d3}\hspace{-.5in}
	\includegraphics[width=5cm, height=4.5cm, trim=0cm 0cm 0 0]{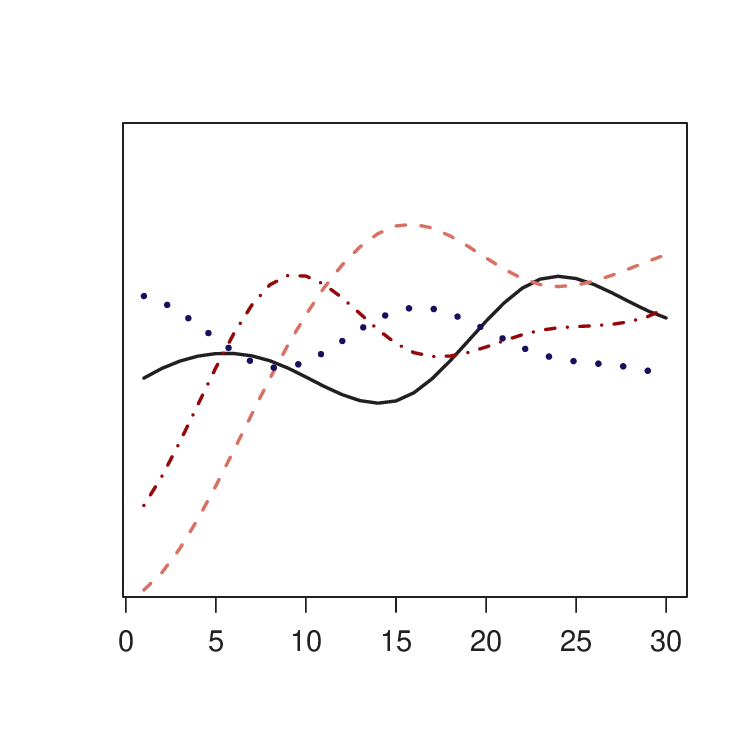}
	\label{fig_sim4_mean_curves_d4}
	\vspace{-.4in}
	\caption{Cluster-specific mean curves of one simulated data set, for each dimension $d\!=\!1,\dots,4$. Each cluster is identified by curves with the same line type and color across the four dimensions.}
	\label{fig:sim4_mean_curves}
	%\end{figure*}
	%
	%
	%\begin{figure*}[t]%[!htb]%[t]%[h]
	\centering
	{\footnotesize$d=1$\hspace{3.2cm}$d=2$\hspace{3.2cm}$d=3$\hspace{3.2cm}$d=4$}\\[-.25in]
	\includegraphics[width=5cm, height=4.5cm, trim=0cm 0cm 0 0]{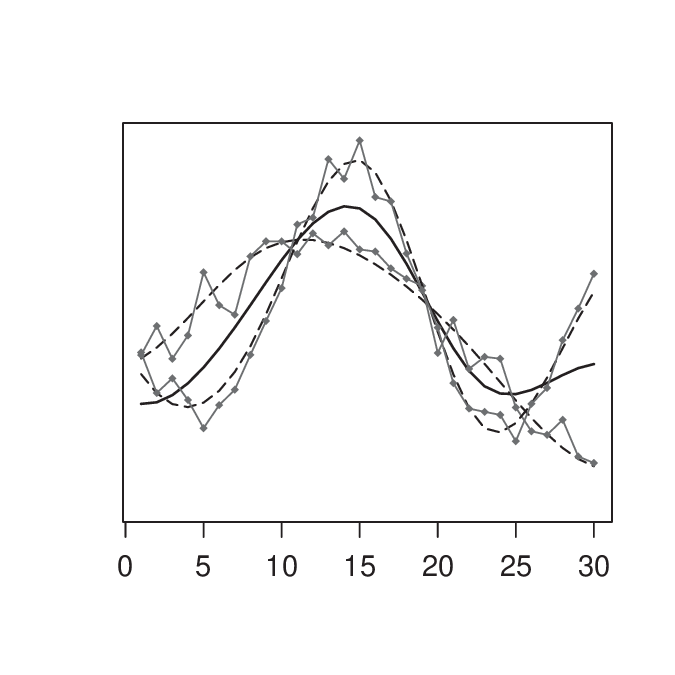}
	\label{fig_sim4_all_d1}\hspace{-.5in}
	\includegraphics[width=5cm, height=4.5cm, trim=0cm 0cm 0 0]{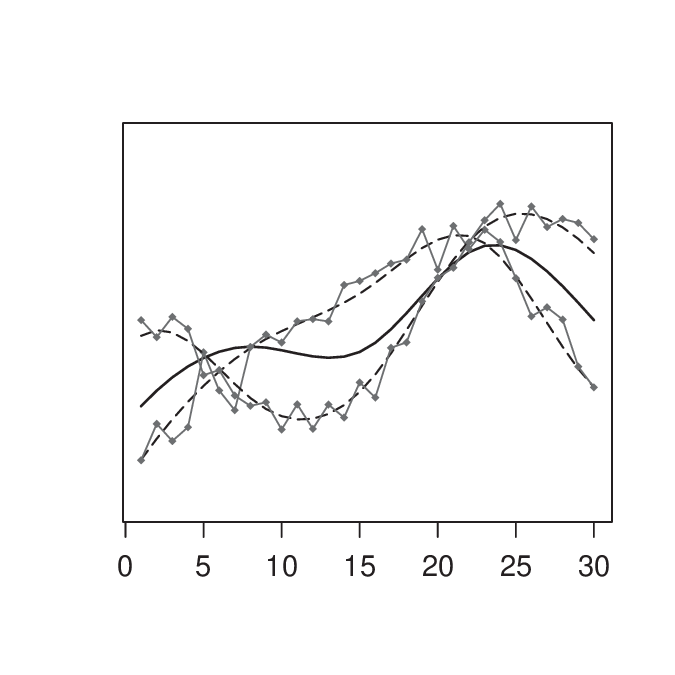}
	\label{fig_sim4_all_d2}\hspace{-.5in}
	\includegraphics[width=5cm, height=4.5cm, trim=0cm 0cm 0 0]{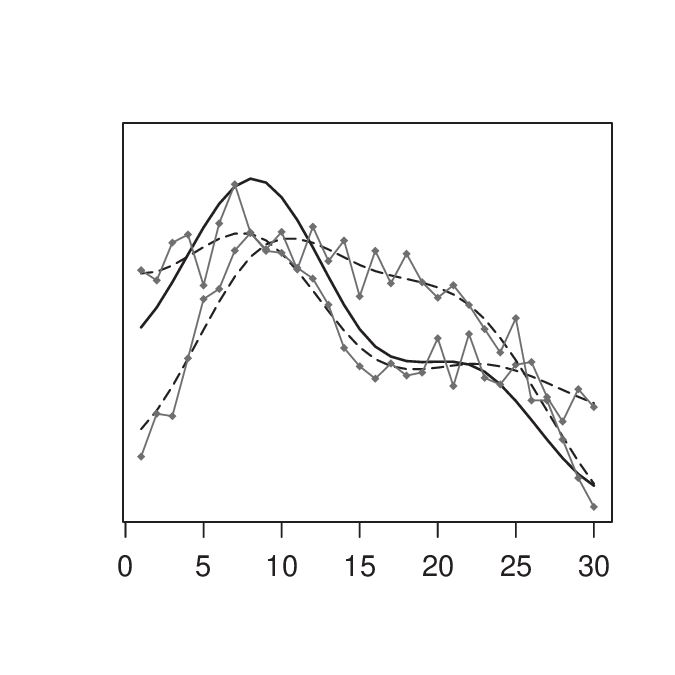}
	\label{fig_sim4_all_d3}\hspace{-.5in}
	\includegraphics[width=5cm, height=4.5cm, trim=0cm 0cm 0 0]{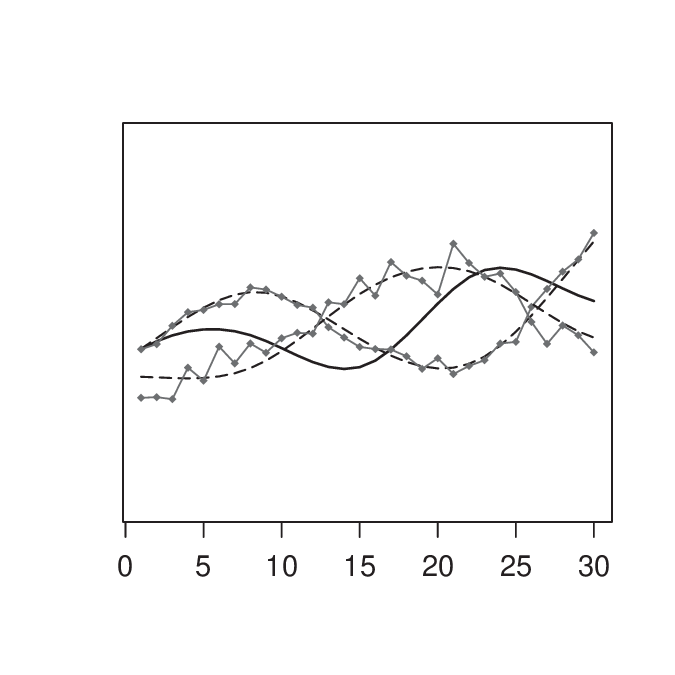}
	\label{fig_sim4_all_d4}
	\vspace{-.4in}	
	\caption{Cluster-specific mean curves (black solid line) and two examples of individual curves (gray dashed line) and respective individual data sequences (gray solid line with points), of one cluster of one simulated data set, for each dimension ${d\!=\!1,\dots,4}$.}
	\label{fig:sim4_all}
	\end{adjustwidth}
\end{figure*}

The proposed model was fitted to the simulated data sets with different levels of repulsion by assuming $\phi_d\!=\!\phi$ for all dimensions $d$, with $\phi\!=\!0,0.1,1,5,10,20,30,40,50$, without covariates, such that $\phi=0$ means no repulsion. Regarding the other parameters associated with repulsion, we set $q=2$ and $\nu=1$. Additionally, the following hyperparameter values are assumed: $a_0=1$, $b_0=1$, $a_\tau=1$, $\xi=1$, $\varpi=0.1$, $s^2_0=100^2$, $s^2_\mu=100^2$, $\omega=1$ and $\Sigma_0=I_p$. For the MFPPMx model, we define the cohesion prior with $M=1$. We consider the number of B-spline knots $p=4,7,10$ and the variance truncation parameter $A_d\!=\!A$ for all $d$, with $A\!=\!5,10$. The MCMC algorithm was executed by collecting 1,000 posterior samples with thinning by 100, after discarding the 700,000 iterates as burn-in, of which the first 400,000 were used to calibrate the Metropolis-Hastings step, as discussed in Section \ref{supp:mcmc} of the Appendix.

\subsection{Simulation results}

As shown in Table \ref{tab:sim4_nclr_rand_NS}, on average, the number of clusters and the number of singletons (clusters containing exactly one individual) reduce as the repulsion increases, for all $p$ and $A$ values. It is also noticeable that higher $A$ and lower $p$ produce less clusters (and less singletons). The importance of $p$, $A$, and $\phi$ in the accuracy of cluster identification can also be evaluated through the adjusted Rand-index \citep{hubert1985}. A very precise identification of clusters is obtained if the models are fitted assuming $p=4$ knots and $\phi\leq 20$ if $A=5$ or $\phi\leq 10$ if $A=10$. This is expected, since the data sets were generated considering $p=4$ knots and no repulsion; values of $\phi$ greater than those seem to excessively influence clustering. If $p=7$, it is interesting to note that better accuracy is obtained for values of $\phi$ between 5 and 20; in this case, we have the repulsive model producing better cluster identification than the models without repulsion. When $p=10$, larger values of $\phi$ improve cluster identification; however, this would be verified up to a certain point after which the force of repulsion becomes excessively strong, similarly to what was observed when $\phi>20$ and $p=4$. In conclusion, choosing the repulsive parameter $\phi$ will depend on the values selected for $p$ and $A$. Considering, for instance, higher prior uncertainty (higher $A$) about the component-specific variance $\lambda^2_{jd}$, we need stronger repulsion to obtain better accuracy.

\begin{table*}[!htb]%[t]
	\begin{adjustwidth}{-1.5cm}{0.0cm}
	\centering
	\footnotesize
	\begin{tabular*}{1.0\textwidth}{
			l@{\hskip .2in}
			r@{\hskip .1in}r@{\hskip .1in}r@{\hskip .2in}
			r@{\hskip .1in}r@{\hskip .1in}r@{\hskip .2in}
			r@{\hskip .1in}r@{\hskip .1in}r@{\hskip .3in}
			r@{\hskip .1in}r@{\hskip .1in}r@{\hskip .2in}
			r@{\hskip .1in}r@{\hskip .1in}r@{\hskip .2in}
			r@{\hskip .1in}r@{\hskip .1in}r
		}
		%\toprule\\[-.12in]
		\cmidrule(l{-2pt}r{2pt}){1-19}\\[-.15in]
		\multicolumn{1}{c}{} &
		\multicolumn{9}{c}{\hspace{-.4in}$A=5$} &
		\multicolumn{9}{c}{\hspace{-.3in}$A=10$}\\[-.01in]
		\cmidrule(l{-3pt}r{20pt}){2-10}\cmidrule(l{-3pt}r{5pt}){11-19}\\[-.15in]
		\multicolumn{1}{c}{} &
		\multicolumn{3}{c}{\hspace{-.25in}$p=4$} &
		\multicolumn{3}{c}{\hspace{-.2in}$p=7$} &
		\multicolumn{3}{c}{\hspace{-.2in}$p=10$} &
		\multicolumn{3}{c}{\hspace{-.2in}$p=4$} &
		\multicolumn{3}{c}{\hspace{-.2in}$p=7$} &
		\multicolumn{3}{c}{\hspace{-.02in}$p=10$}\\%[-.08in]
		\cmidrule(l{-3pt}r{12pt}){2-4}
		\cmidrule(l{-3pt}r{12pt}){5-7}
		\cmidrule(l{-2pt}r{20pt}){8-10}
		\cmidrule(l{-3pt}r{12pt}){11-13}
		\cmidrule(l{-3pt}r{12pt}){14-16}
		\cmidrule(l{-2pt}r{5pt}){17-19}\\[-.15in]
		Model &
		\multicolumn{3}{c}{\hspace{-.6cm}NC\hspace{.3cm}NS\hspace{.3cm}RI} &
		\multicolumn{3}{c}{\hspace{-.6cm}NC\hspace{.3cm}NS\hspace{.3cm}RI} &
		\multicolumn{3}{c}{\hspace{-.8cm}NC\hspace{.3cm}NS\hspace{.3cm}RI} &
		\multicolumn{3}{c}{\hspace{-.5cm}NC\hspace{.3cm}NS\hspace{.3cm}RI} &
		\multicolumn{3}{c}{\hspace{-.6cm}NC\hspace{.3cm}NS\hspace{.3cm}RI} &
		\multicolumn{3}{c}{\hspace{-.25cm}NC\hspace{.3cm}NS\hspace{.3cm}RI}\\[-.02in]
		\cmidrule(l{-2pt}r{4pt}){1-19}\\[-.15in]
		MFPPMx 		 & 4.0 &   - & 0.98 & 5.8 &   - & 0.85 &  7.9 &   - & 0.72 & 4.0 &   - & 1.00 & 4.7 &   - & 0.92 &  7.0 &   - & 0.76 \\%[-.08in]
		%\midrule\\[-.12in]
		\cmidrule(l{-2pt}r{4pt}){1-19}\\[-.12in]
		MFRMMx   	 & \multicolumn{12}{c}{} \\%[-.15in]
		no repulsion & 4.3 & 0.2 & 0.99 & 7.5 & 2.2 & 0.84 & 11.9 & 5.8 & 0.69 & 4.0 &   - & 0.99 & 6.4 & 1.0 & 0.86 & 10.6 & 4.5 & 0.73 \\%[-.14in]
		$\phi=0.1$   & 4.2 & 0.2 & 0.98 & 6.5 & 1.4 & 0.87 & 11.2 & 4.9 & 0.69 & 4.0 &   - & 0.99 & 5.9 & 0.7 & 0.88 &  9.7 & 2.9 & 0.72 \\%[-.14in]
		$\phi=1$     & 4.0 &   - & 0.99 & 5.6 & 0.9 & 0.91 & 10.9 & 3.5 & 0.71 & 4.0 &   - & 1.00 & 4.6 & 0.2 & 0.95 &  8.0 & 1.5 & 0.76 \\%[-.14in]
		$\phi=5$     & 4.0 &   - & 1.00 & 4.4 & 0.2 & 0.95 &  6.1 & 0.6 & 0.83 & 4.0 &   - & 0.99 & 4.0 &   - & 0.96 &  5.8 & 0.5 & 0.85 \\%[-.14in]
		$\phi=10$    & 4.0 &   - & 0.99 & 4.5 & 0.3 & 0.95 &  6.4 & 1.0 & 0.86 & 3.8 &   - & 0.96 & 4.0 & 0.1 & 0.95 &  5.6 & 0.8 & 0.87 \\%[-.14in]
		$\phi=20$    & 4.0 &   - & 0.99 & 4.8 & 0.7 & 0.93 &  6.2 & 1.2 & 0.85 & 3.6 &   - & 0.90 & 4.0 &   - & 0.96 &  5.0 & 0.5 & 0.91 \\%[-.14in]
		$\phi=30$    & 3.8 &   - & 0.93 & 4.3 & 0.3 & 0.93 &  6.6 & 1.8 & 0.85 & 3.3 &   - & 0.80 & 3.8 &   - & 0.93 &  4.7 & 0.4 & 0.93 \\%[-.14in]
		$\phi=40$    & 3.6 &   - & 0.88 & 4.6 & 0.5 & 0.93 &  6.2 & 1.4 & 0.86 & 3.2 &   - & 0.77 & 3.8 &   - & 0.90 &  4.6 & 0.4 & 0.93 \\%[-.14in]
		$\phi=50$    & 3.4 &   - & 0.83 & 4.5 & 0.4 & 0.92 &  5.9 & 1.4 & 0.88 & 3.1 &   - & 0.73 & 3.8 &   - & 0.91 &  4.4 & 0.2 & 0.94 \\%[-.14in]
		%\bottomrule
		\cmidrule(l{-2pt}r{2pt}){1-19}
	\end{tabular*}
	\vspace{-.1in}
	\end{adjustwidth}
	\begin{adjustwidth}{-1.0cm}{-1.0cm}
	\caption{Average of the number of clusters (columns NC), number of singletons (columns NS) and adjusted Rand-index (columns RI) estimated by fitting the MFRMMx and MFPPMx models to the simulated data sets, considering different values for $p$, $A$ and $\phi$. The true number of clusters is four.}
	\end{adjustwidth}
	\label{tab:sim4_nclr_rand_NS}
	\vspace{-.3in}
\end{table*}

The variations in the number of clusters and in the accuracy of the model shown in Table \ref{tab:sim4_nclr_rand_NS} are associated with greater or lower differentiation between the mean cluster curves, obtained under stronger or weaker repulsion, respectively. This differentiation is measured by the pairwise distance between cluster-specific mean curves defined in \eqref{eq_distance_approx}. Figure \ref{fig:sim4_dist_mean} shows the evolution of average distance between cluster mean curves across increasing values of $\phi$, estimated with $p=4$ and $A=5$. This figure considers only the 12 simulated datasets in which four clusters were estimated for all values of $\phi$, so that the distances are not affected by variations in the number of clusters across different $\phi$ values, and a smooth increasing trend in those distances are observed. The clusters were estimated through the least-squares method proposed by \cite{dahl2006}.

\begin{figure*}[!htb]
	\centering
	\stackunder[1pt]{
		\raisebox{2.6cm}{\rotatebox[origin=c]{90}{\scriptsize Mean distance}}
		\includegraphics[width=15cm, height=4.5cm, trim=0 0.0cm 0 0]{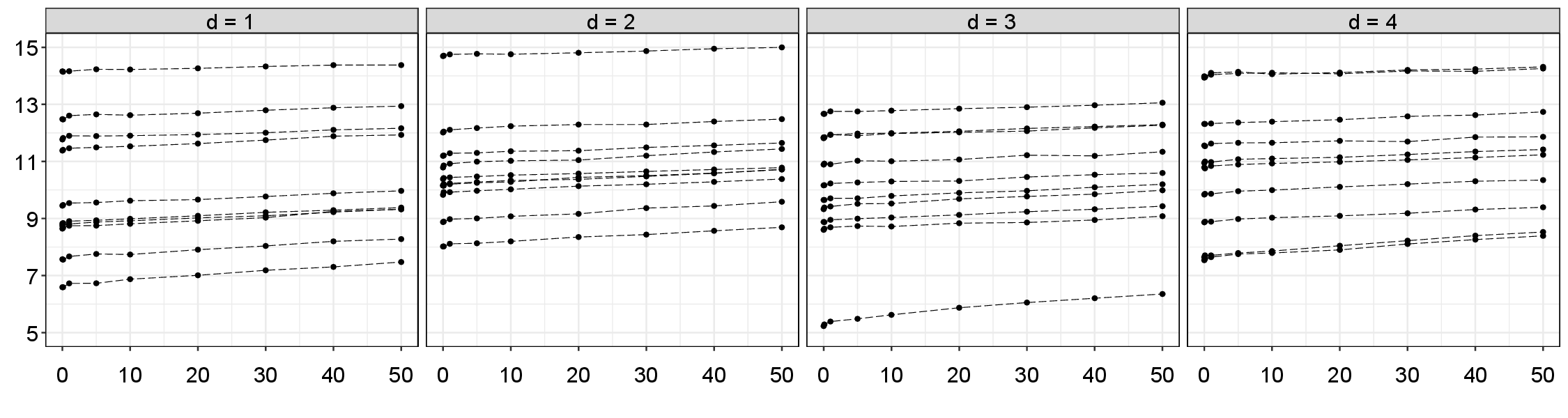}}
	{}%{\scriptsize \hspace{.25in}$\phi$}
	\vspace{-.25in}
	\caption{Mean distances between the cluster-specific mean curves for each each dimension $d$, obtained by fitting the MFRMMx model considering the $\phi=0,0.1,1,5,10,20,30,40,50$, $p\!=\!4$ and $A\!=\!5$, for 12 simulated data sets.}
	\label{fig:sim4_dist_mean}
\end{figure*}

Table \ref{tab:sim4_nclr_rand_NS} also shows that the (non-repulsive) MFPPMx model provided more accurate cluster identification than the MFRMMx model without repulsion $(\phi=0)$, for most of the $p$ and $A$ values.  However, in all scenarios, there is always some level of repulsion for which the MFRMMx model has equal or greater accuracy to that of the MFPPMx. Figure \ref{fig:sim4_ccmat_phi-dep-p10-A10} shows some examples where the MFRMMx model, assuming $p=10$, $A=10$ and different levels of repulsion, outperforms the MFPPMx model.

\begin{figure*}[!htb]
%\begin{adjustwidth}{-1cm}{-1cm}
\centering
	{\scriptsize
	 \hspace{0.7cm} MFPPMx
	 \hspace{1.2cm} MFRMMx ($\phi\!=\!0$)
	 \hspace{0.5cm} MFRMMx ($\phi\!=\!10$)
	 \hspace{0.5cm} MFRMMx ($\phi\!=\!30$)
	 \hspace{0.4cm} MFRMMx ($\phi\!=\!50$)
	 }\\[-.05in]
	 \stackunder[2pt]{
		%\raisebox{1.2cm}{\rotatebox[origin=c]{90}{\normalsize $A=10$}}
		\includegraphics[width=3cm, height=2.8cm, trim=0 0cm 0 0]{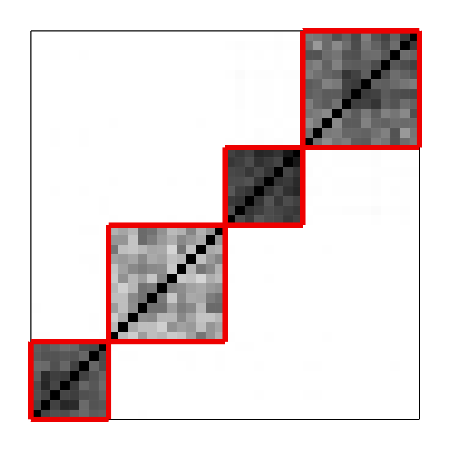}}
		{\scriptsize (0.45 , 0.76 , 0.93)}%\hspace{-.05in}
	\stackunder[2pt]{
		\includegraphics[width=3cm, height=2.8cm, trim=0 0cm 0 0]{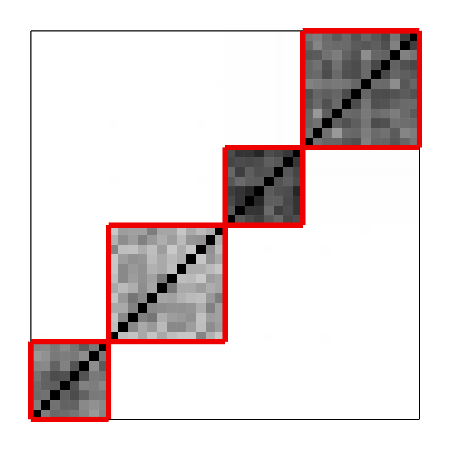}}
		{\scriptsize (0.44 , 0.73 , 0.92)}%\hspace{-.05in}
	\stackunder[2pt]{
		\includegraphics[width=3cm, height=2.8cm, trim=0 0cm 0 0]{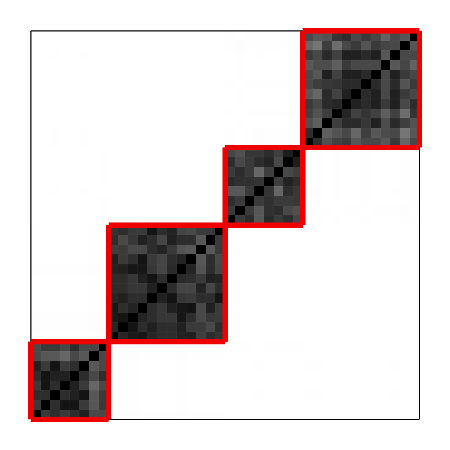}}
		{\scriptsize (0.58 , 0.87 , 1.00)}%\hspace{-.05in}
	\stackunder[2pt]{
		\includegraphics[width=3cm, height=2.8cm, trim=0 0cm 0 0]{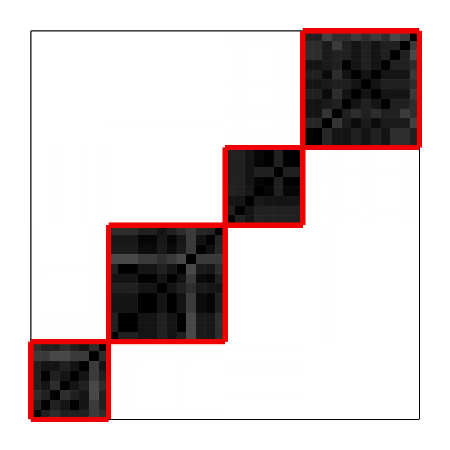}}
		{\scriptsize (0.66 , 0.94 , 1.00)}%\hspace{-.05in}
	\stackunder[2pt]{
		\includegraphics[width=3cm, height=2.8cm, trim=0 0cm 0 0]{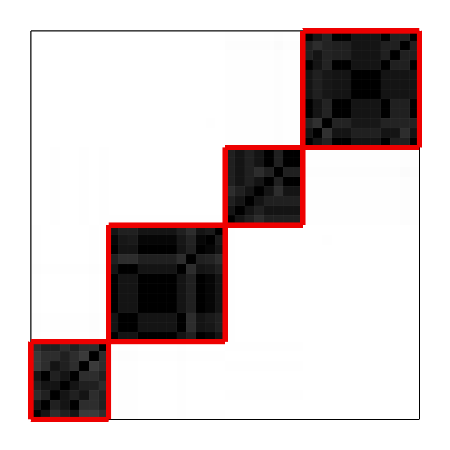}}
		{\scriptsize (0.85 , 0.98 , 1.00)}
	\\%[.01in]
%	%
%	{\normalsize (a) $p=10 , A=10$}\\[-.03in]
%	%
%\end{adjustwidth}
\vspace{-.05in}
\caption{Average co-clustering matrices estimated over the 25 simulated data sets by the MFPPMx and MFRMMx models, assuming $p=10$, $A=10$ and a few different values of $\phi$. The true clusters are identified by the red boxes. The three values bellow each co-clustering matrix are the minimum, mean and maximum values of the Rand-index computed over the simulated data sets.}
\label{fig:sim4_ccmat_phi-dep-p10-A10}
%\vspace{-.1in}
\end{figure*}

\FloatBarrier
\subsection{Effects of covariates on clustering}
To verify the influence of covariates in clustering, we generated an individual-specific continuous covariate from $\N(z_i,0.1)$, independently for each individual $i=1,\dots,m$, where $z_i$ is the cluster index assumed for individual $i$ in the simulated data sets, and considered a categorical covariate taking on identical values for individuals in the same cluster. We fitted the MFRMMx model considering these covariates to influence the clustering through the covariate-dependent mixing weights, and assuming $p=4$, $A=10$ and $\phi\ge20$, which are the settings with worst performance among those considered in the analysis without covariates, as can be seen in Table \ref{tab:sim4_nclr_rand_NS}. The co-clustering matrices and the adjusted Rand-index presented in Figure \ref{fig:sim4_ccmat_p4_phi-dep-covariates} show that considering the simulated individual-specific covariates favored, as expected, the correct allocation of the inviduals to their true clusters, working against the repulsive effect that seemed excessive in the case of those larger values of $\phi$ when $p=4$.

\begin{figure*}[h]
\begin{adjustwidth}{-1cm}{-1cm}
\centering
	{\scriptsize\hspace{0.2cm}$\phi=20$\hspace{1.8cm}$\phi=25$\hspace{2.0cm}$\phi=50$\hspace{2.0cm}$\phi=75$\hspace{2.0cm}$\phi=100$}
	\\[-.02in]
	%
	%%% A=10, p=4 (no covariates)
	%
	\stackunder[5pt]{
		\raisebox{1.1cm}{\rotatebox[origin=c]{90}{\scriptsize no covariates}}
		\includegraphics[width=2.8cm, height=2.5cm, trim=0.2cm 0.5cm 0 0]{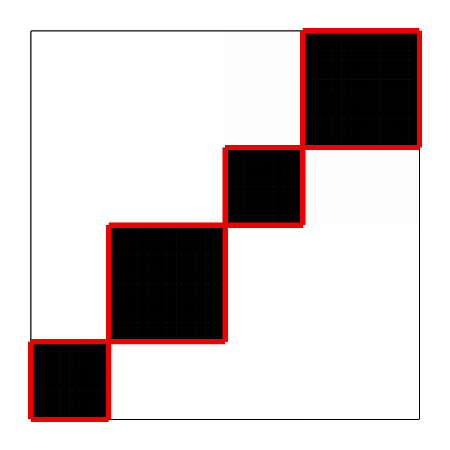}}
		{\scriptsize\hspace{0.2cm} (0.71 , 0.96 , 1.00)}%\hspace{-.1in}
	\stackunder[5pt]{
		\includegraphics[width=2.8cm, height=2.5cm, trim=0.2cm 0.5cm 0 0]{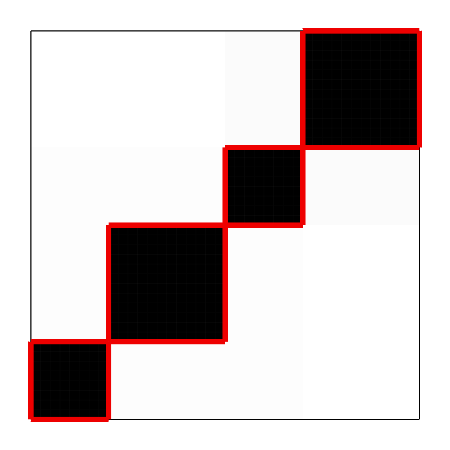}}
		{\scriptsize\hspace{0.0cm} (0.71 , 0.90 , 1.00)}%\hspace{-.1in}
	\stackunder[5pt]{
		\includegraphics[width=2.8cm, height=2.5cm, trim=0.2cm 0.5cm 0 0]{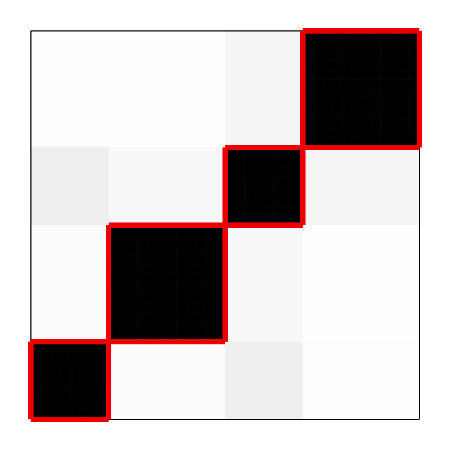}}
		{\scriptsize\hspace{0.0cm} (0.27 , 0.80 , 1.00)}%\hspace{-.1in}
	\stackunder[5pt]{
		\includegraphics[width=2.8cm, height=2.5cm, trim=0.2cm 0.5cm 0 0]{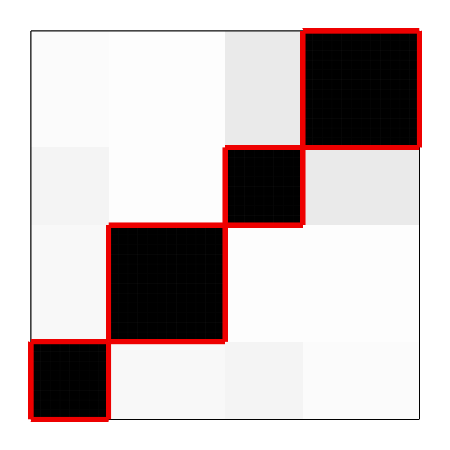}}
		{\scriptsize\hspace{0.0cm} (0.27 , 0.77 , 1.00)}%\hspace{-.1in}
	\stackunder[5pt]{
		\includegraphics[width=2.8cm, height=2.5cm, trim=0.2cm 0.5cm 0 0]{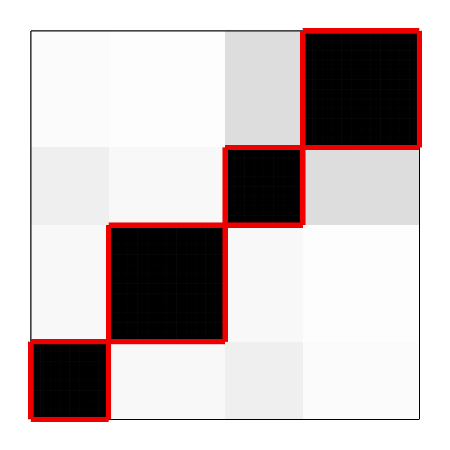}}
		{\scriptsize\hspace{0.0cm} (0.27 , 0.73 , 1.00)}
	\\[.01in]
	%
	%%% A=10, p=4 (covariates)
	%
	\stackunder[5pt]{
	\raisebox{1.2cm}{\rotatebox[origin=c]{90}{\scriptsize covariates}}
		\includegraphics[width=2.8cm, height=2.5cm, trim=0.2cm 0.5cm 0 0]{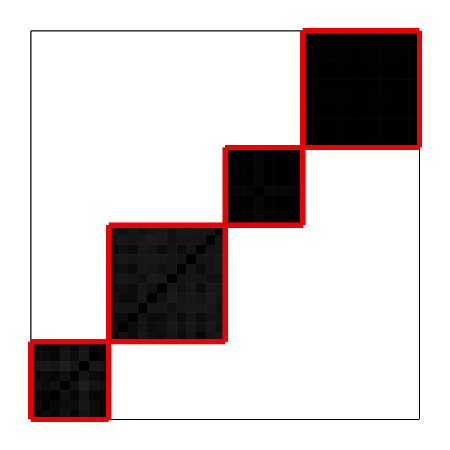}}
		{\scriptsize\hspace{0.2cm} (0.91 , 0.99 , 1.00)}%\hspace{-.1in}
	\stackunder[5pt]{
		\includegraphics[width=2.8cm, height=2.5cm, trim=0.2cm 0.5cm 0 0]{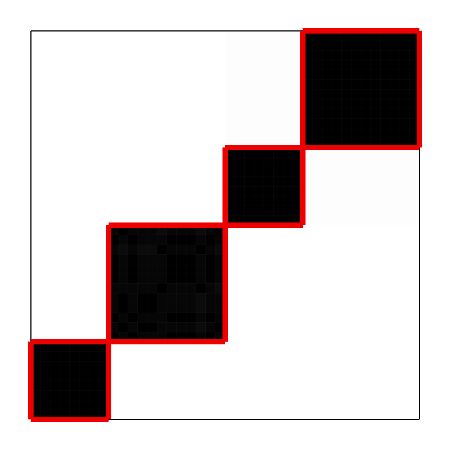}}
		{\scriptsize\hspace{0.0cm} (0.71 , 0.98 , 1.00)}%\hspace{-.1in}
	\stackunder[5pt]{
		\includegraphics[width=2.8cm, height=2.5cm, trim=0.2cm 0.5cm 0 0]{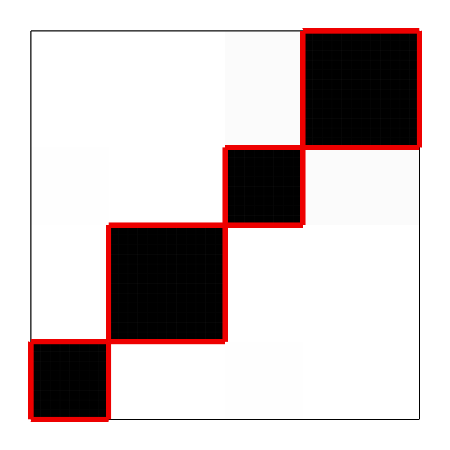}}
		{\scriptsize\hspace{0.0cm} (0.71 , 0.96 , 1.00)}%\hspace{-.1in}
	\stackunder[5pt]{
		\includegraphics[width=2.8cm, height=2.5cm, trim=0.2cm 0.5cm 0 0]{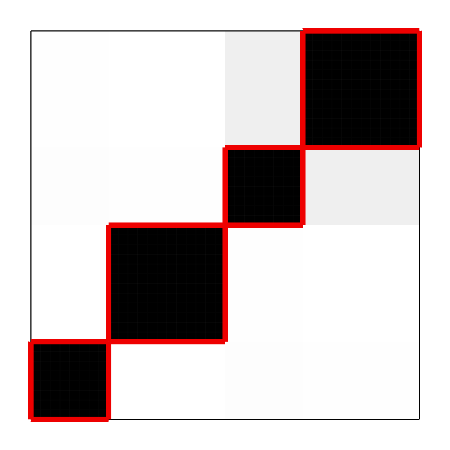}}
		{\scriptsize\hspace{0.0cm} (0.39 , 0.90 , 1.00)}%\hspace{-.1in}
	\stackunder[5pt]{
		\includegraphics[width=2.8cm, height=2.5cm, trim=0.2cm 0.5cm 0 0]{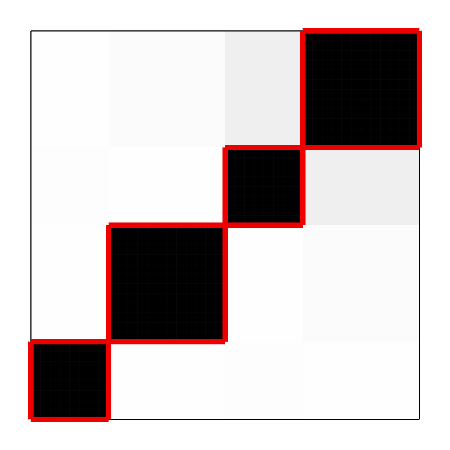}}
		{\scriptsize\hspace{0.0cm} (0.27 , 0.84 , 1.00)}
	\\[-.051in]
\end{adjustwidth}
\vspace{-.1in}
\caption{Average co-clustering matrices computed by fitting the MFRMMx over the 25 simulated data sets, by assuming $p=4$, $A=10$ and $\phi=10,20,30,40,50$, without covariates (first line) and with covariates (second line). The true clusters are identified by the red boxes. The three values bellow each co-clustering matrix are the minimum, mean and maximum values of the Rand-index computed over the simulated data sets.}
\label{fig:sim4_ccmat_p4_phi-dep-covariates}
%\vspace{-.1in}
\end{figure*}

\FloatBarrier
\subsection{Effects of the split-merge step on clustering}

The split-merge step in the MCMC algorithm provided better chain mixing and was essential for efficient posterior sampling in the proposed model. To illustrate this fact, Figure \ref{fig:sim4_SM} exhibits the mean values of the adjusted Rand-index obtained for the simulated data sets after running the MCMC algorithm with and without the split-merge step. The model was fitted assuming $p=4$, $A=5,10$, and different repulsion levels $\phi$.

It is noticeable that cluster identification is severely compromised if we sample the posterior without the split-merge step. When the split-merge step is considered, we have very accurate cluster identification for the cases with no or low repulsion. The loss of accuracy observed for values of $\phi\ge30$ when $A=5$ and $\phi\ge10$ when $A=10$ seems to be an exclusive consequence of the increase in repulsion. On the other hand, if the split-merge step is not considered, the adjusted Rand-Index is zero if $\phi\ge20$ and $A=10$ for all the data sets. In our case, it means that all the individuals are allocated to the same cluster. Since we initialize the chain assuming that all individuals are allocated to the same mixture component, this indicates that the cluster allocation did not change during the iterations of the MCMC algorithm and the parametric space is not being adequately explored. This analysis exemplifies that the split-merge step is essential for the proposed model to properly identify clusters, even though it increases the computational cost of the MCMC algorithm.

\begin{figure}[h]%[t]
	\centering
	\stackunder[1pt]{
		\raisebox{2.0cm}{\rotatebox[origin=c]{90}{\footnotesize Adjusted Rand-index}}
		\includegraphics[width=6cm, height=4cm, trim=0.2cm 0 0 0]{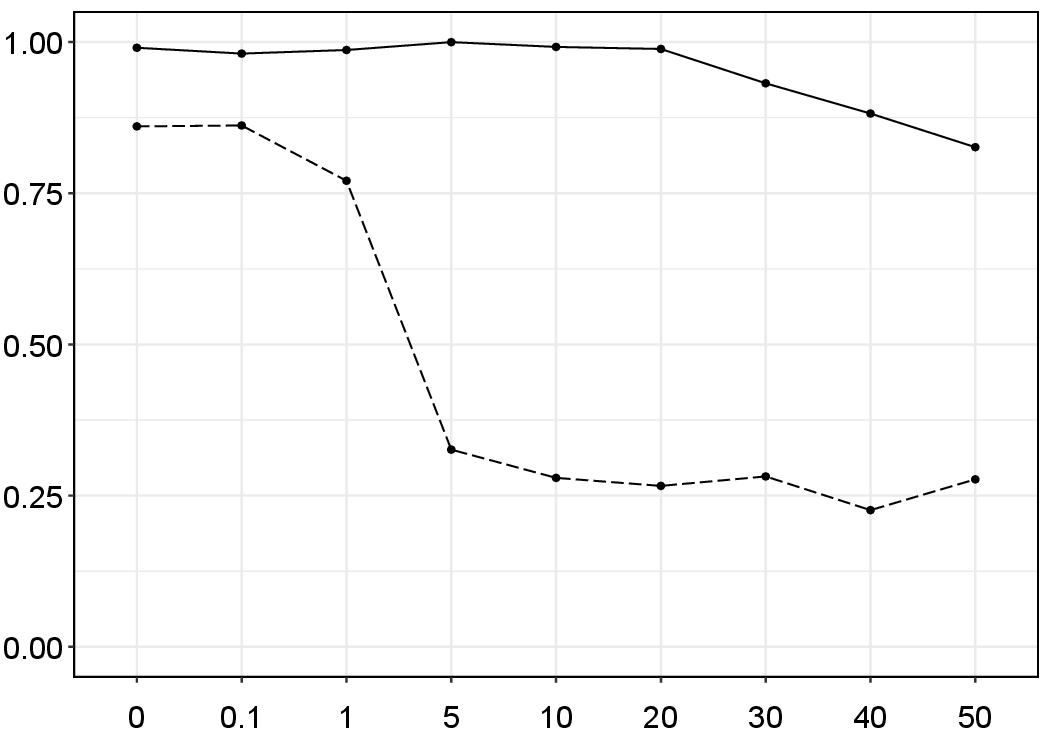}}
	{\scriptsize \hspace{.25in}$\phi$}
	\hspace{.5cm}
	\stackunder[1pt]{
		\raisebox{2.0cm}{\rotatebox[origin=c]{90}{\footnotesize Adjusted Rand-index}}
		\includegraphics[width=6cm, height=4cm, trim=0.2cm 0 0 0]{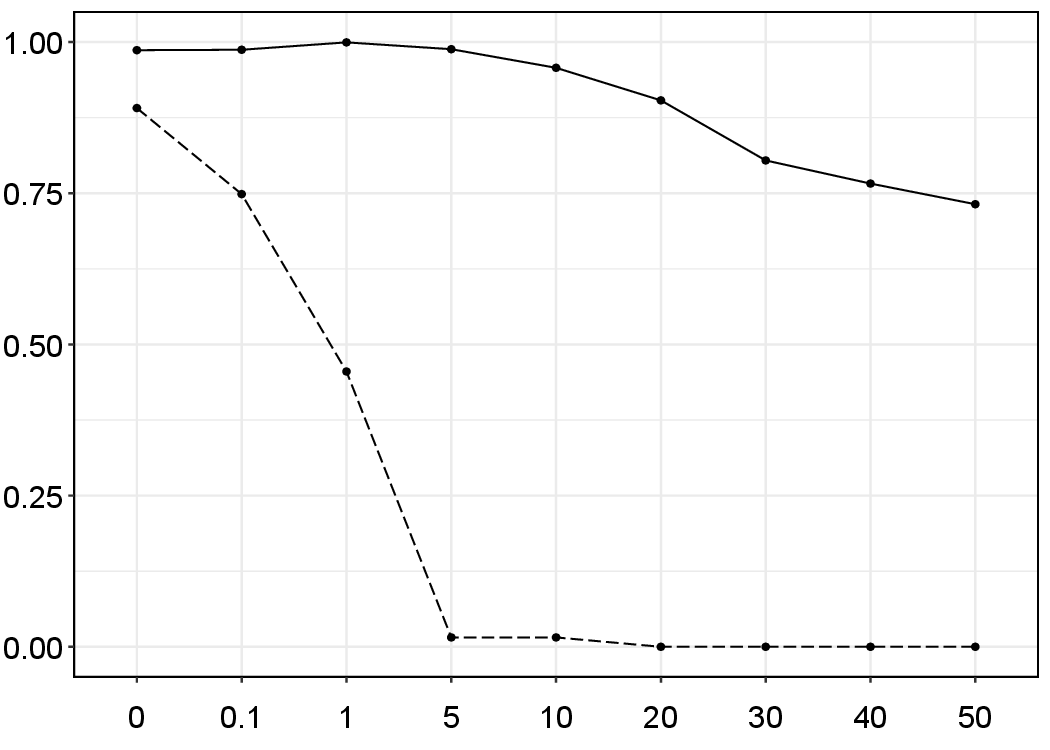}}
	{\scriptsize \hspace{.25in}$\phi$}
	\vspace{-.25cm}
	\caption{Average of adjusted Rand-index obtained by executing the MCMC algorithm with (solid line) and without (dashed line) the split-merge step, for $p=4$, $A=5$ (left) and $A=10$ (right), and different levels of repulsion.}
	\label{fig:sim4_SM}
\end{figure}

\FloatBarrier
%\newpage
\vspace{-.2in}
\section{Real data application: CAI dataset}\label{sec:real}

The CAI dataset is composed of multivariate sequences of angle measures observed for three body joints (ankle, knee, and hip) from two planes (frontal and sagittal) during a jumping-landing-cutting movement task performed by individuals with some type of Chronic Ankle Instability (CAI). This dataset was first analyzed by \cite{hopkins2019}. Their main goal was to identify subsets of individuals with similar types of physical dysfunctions, based on the similarity of movement patterns characterized by the formats of the observed angle sequences. In their study, they reduced the multivariate data sequence of each individual from total dimension $D\!=\!6$ to $D\!=\!2$ by transforming the three sequences of joints of each plane to only one representative curve associated to the respective plane, through functional principal component analysis \citep{ramsay2005}. We will analyze the same format of the CAI dataset considered in the analysis of \cite{hopkins2019}, with $D=2$ data sequences by individual. Our main goal is to evaluate the ability of the proposed model to overcome the cluster redundancy problem reported their study. In fact, the accuracy of diagnoses and predictions may be compromised when different clusters with similar curve shapes are identified. The data sequences in the CAI dataset are very smooth, which means that there is almost no noise between those data sequences and the individual-specific B-spline curves $H_i\Bi$ defined in \eqref{eq_yi}. Thus the individual-specific curves will almost interpolate the data sequences. Figure \ref{fig:CAI2_data} shows the data sequences for a sample of $50$ individuals of the CAI dataset.

\vspace{-.2in}
\begin{figure}[h]%[!htb]
	\centering
	\stackunder[-10pt]{
		\includegraphics[width=5cm, height=4.8cm, trim=0 0cm 0 0]{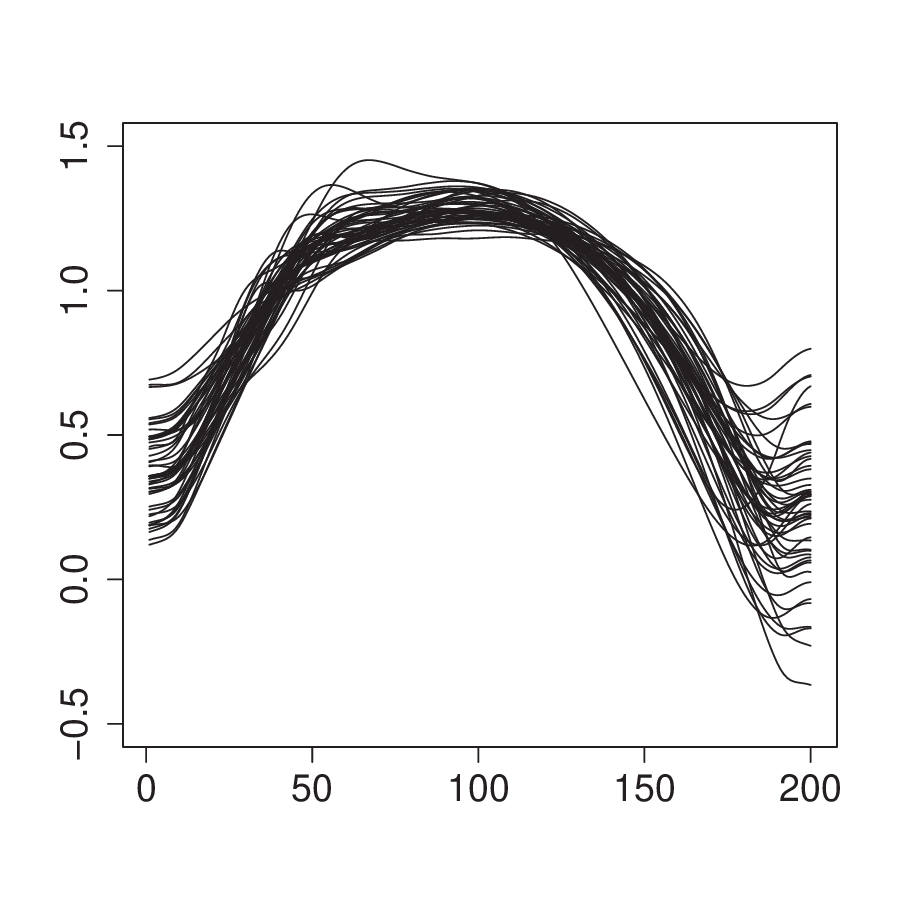}}
	{\footnotesize\hspace{0.3cm} Sagittal plane}\hspace{-0.5cm}
	\stackunder[-10pt]{
		\includegraphics[width=5cm, height=4.8cm, trim=0 0cm 0 0]{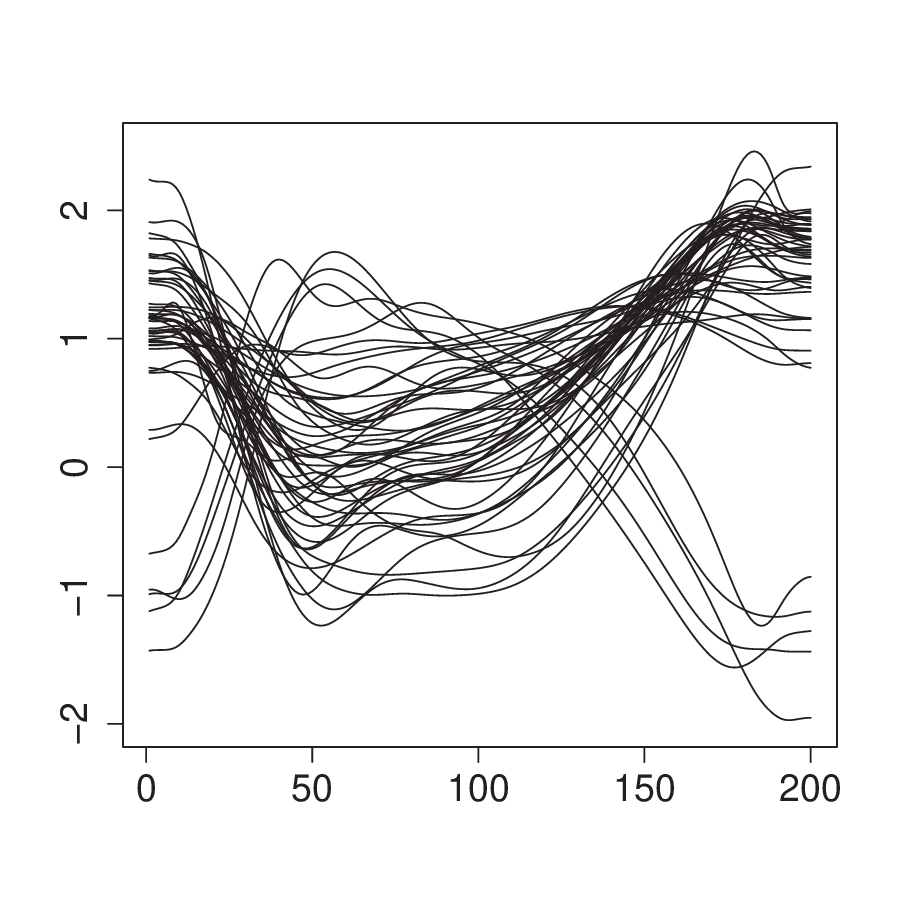}}
	{\footnotesize\hspace{0.2cm} Frontal plane}
	\vspace{-.2in}
	\caption{Observed angle measures of the sagittal and frontal planes of $50$ individuals of the CAI dataset.}
	\label{fig:CAI2_data}
\end{figure}

The MFRMMx and MFPPMx models are fitted to the CAI dataset assuming no covariates to influence the clustering. We consider the number of knots $p=10,15$, and the values of the truncation parameter $A=1,10,100$. We fit the MFRMMx assuming $J=40$ mixture components and $\phi_d=\phi$ for $d=1,2$, with no repulsion ($\phi=0$) and different levels of repulsion determined by $\phi=0.01,0.05,0.1,0.5,1,2,5,10,20,50$, where dimensions $d=1$ and $d=2$ refer to the sagittal and frontal planes, respectively. The other hyperparameter values are the same considered in Section \ref{sec:simulation}. All models were fit by collecting 1,000 MCMC iterates after discarding the first 700,000 as burn-in (the first 400,000 used to calibrate the Metropolis-Hastings step) and thinning by 100. Similar to what was observed in the simulation study, we can see from Table \ref{tab:CAI_nclr} that higher repulsion (higher $\phi$ values) implies a reduction in the posterior mean number of clusters for all choices of $p$ and $A$. Furthermore, a larger number of B-spline nodes ($p=15$) or a smaller value of $A$ tends to identify a larger number of clusters.

\begin{table}[!htb]%[t]%[h]
	\centering
	\footnotesize
	\begin{tabular}{l@{\hskip .1in}r@{\hskip .03in}r@{\hskip .1in}r@{\hskip .22in}r@{\hskip .1in}r@{\hskip .22in}r@{\hskip .1in}r}
		\toprule\\[-.15in]
		\multicolumn{1}{c}{} &
		\multicolumn{1}{c}{} &
		\multicolumn{2}{c}{\hspace{-.15in}$A\!=\!1$} &
		\multicolumn{2}{c}{\hspace{-.20in}$A\!=\!10$}&
		\multicolumn{2}{c}{\hspace{-.05in}$A\!=\!100$}\\[-.04in]
		\cmidrule(l{2pt}r{14pt}){3-4}\cmidrule(l{0pt}r{14pt}){5-6}\cmidrule(l{0pt}r{5pt}){7-8}\\[-.15in]
		Model &$p\!=$&
		\multicolumn{2}{c}{\hspace{-.5cm}10\hspace{.5cm}15} &
		\multicolumn{2}{c}{\hspace{-.5cm}10\hspace{.5cm}15} &
		\multicolumn{2}{c}{\hspace{-.2cm}10\hspace{.5cm}15}\\[-.04in]
		\midrule \\[-.15in]
		MFPPMx 		 && 20.0 & 26.0 & 15.6 & 24.5 & 15.4 & 22.1 \\%[-.08in]
		\midrule \\[-.15in]
		MFRMMx   	 && \multicolumn{6}{c}{} \\[-.04in]
		no repulsion && 23.7 & 34.5 & 18.6 & 24.9 & 16.0 & 18.0 \\%[-.14in]
		$\phi=0.01$  && 20.4 & 28.2 & 14.2 & 19.0 & 11.0 & 15.0 \\[-.02in]
		$\phi=0.05$  && 13.0 & 18.0 & 11.0 & 16.0 &  8.0 & 14.5 \\[-.02in]
		$\phi=0.1$   && 11.0 & 15.0 &  7.5 & 14.0 &  5.0 & 11.0 \\[-.02in]
		$\phi=0.5$   &&  4.2 &  6.0 &  4.0 &  6.0 &  4.0 &  6.0 \\[-.02in]
		$\phi=1$     &&  4.0 &  5.0 &  4.0 &  5.0 &  3.0 &  5.0 \\[-.02in]
		$\phi=2$     &&  3.0 &  4.0 &  3.0 &  4.0 &  3.0 &  4.0 \\[-.02in]
		$\phi=5$     &&  3.0 &  3.0 &  3.0 &  3.0 &  2.0 &  3.0 \\[-.02in]
		$\phi=10$    &&  3.0 &  3.0 &  2.0 &  3.0 &  2.0 &  3.0 \\[-.02in]
		$\phi=20$    &&  2.0 &  2.0 &  2.0 &  2.0 &  2.0 &  2.0 \\[-.02in]
		$\phi=50$    &&  1.0 &  2.0 &  1.0 &  2.0 &  1.0 &  2.0 \\[-.02in]
		\bottomrule
	\end{tabular}
	%\vspace{-.05in}
	\caption{Posterior mean of the number of clusters obtained by fitting the MFRMMx and MFPPMx models to the CAI dataset with $D=2$ curves by individual, without the influence of covariates, for different values of $p$, $A$ and $\phi$.}
	\label{tab:CAI_nclr}
\end{table}

Figure \ref{fig:CAI2_mean_curves_A100_p15_phi5} shows the cluster-specific mean curves estimated by the MFRMMx model for dimensions $d=1,2$ in the interval where the values of $\phi$ increase from 0.1 to 5, when $A=100$ and $p=15$. In that figure, we can observe that the clusters with the most similar mean curves gradually merge as repulsion increases.

\begin{figure*}[h]%[!htb]%[t]
\begin{adjustwidth}{-1cm}{-1cm}
\centering
	{\scriptsize$\phi=0.1$\hspace{2.2cm}$\phi=0.5$\hspace{2.5cm}$\phi=1$\hspace{2.5cm}$\phi=2$\hspace{2.5cm}$\phi=5$}
	\\[.05in]
	%
	%%% A=10, p=15, d=1
	%
	\stackunder[5pt]{
		\raisebox{1.4cm}{\rotatebox[origin=c]{90}{\scriptsize Sagittal plane}}
		\includegraphics[width=3.0cm, height=2.8cm, trim=0cm 0cm 0 0]{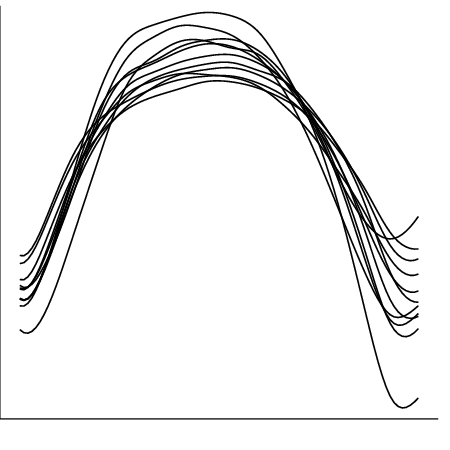}}{}\hspace{.05in}
	\stackunder[5pt]{
		\includegraphics[width=3.0cm, height=2.8cm, trim=0cm 0cm 0 0]{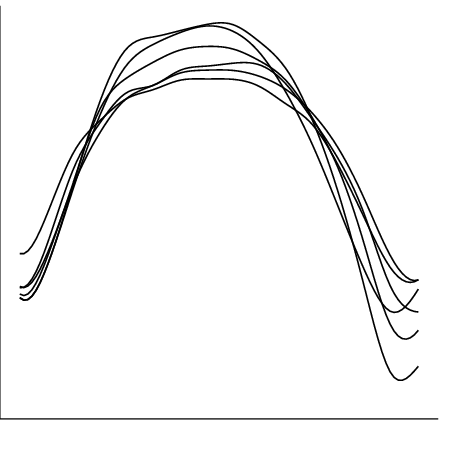}}{}\hspace{.05in}
	\stackunder[5pt]{
		\includegraphics[width=3.0cm, height=2.8cm, trim=0cm 0cm 0 0]{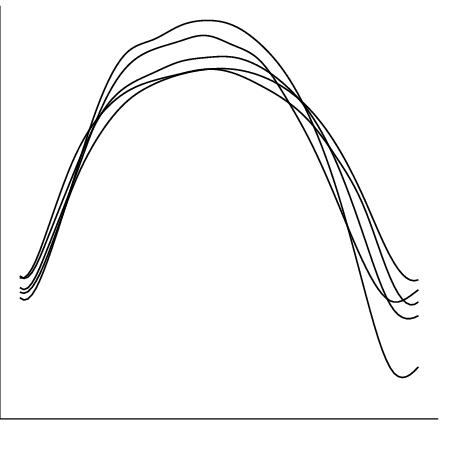}}{}\hspace{.05in}
	\stackunder[5pt]{
		\includegraphics[width=3.0cm, height=2.8cm, trim=0cm 0cm 0 0]{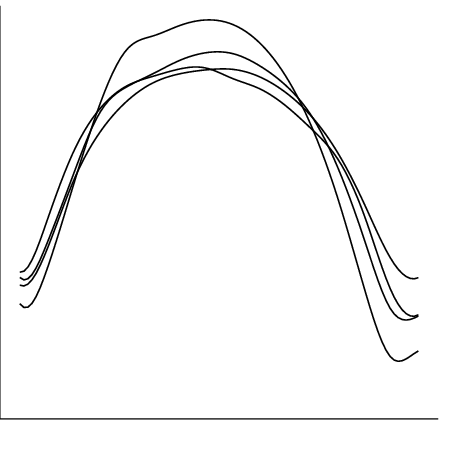}}{}\hspace{.05in}
	\stackunder[5pt]{
		\includegraphics[width=3.0cm, height=2.8cm, trim=0cm 0cm 0 0]{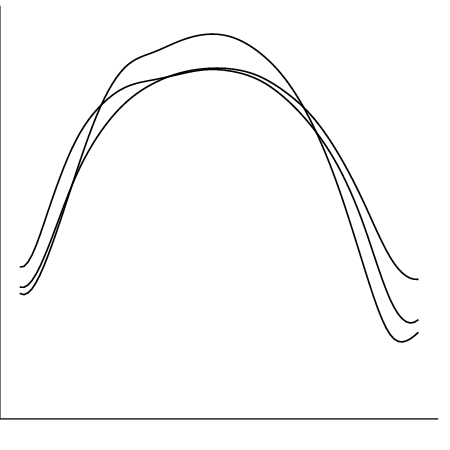}}{}
	\\[-.05in]
	%
	%%% A=10, p=15, d=2
	%
	\stackunder[5pt]{
		\raisebox{1.4cm}{\rotatebox[origin=c]{90}{\scriptsize Frontal plane}}
		\includegraphics[width=3.0cm, height=2.8cm, trim=0cm 0cm 0 0]{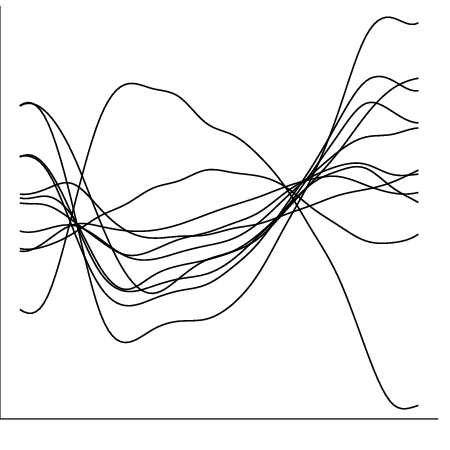}}{}\hspace{.05in}
	\stackunder[5pt]{
		\includegraphics[width=3.0cm, height=2.8cm, trim=0cm 0cm 0 0]{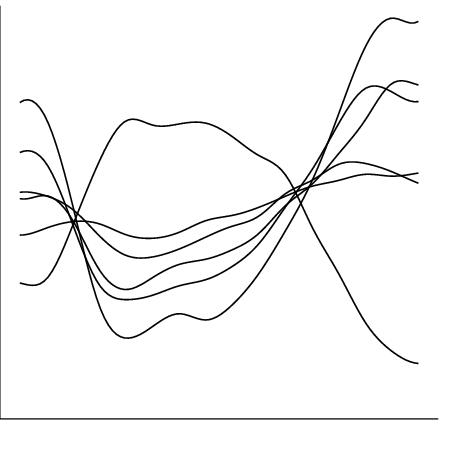}}{}\hspace{.05in}
	\stackunder[3pt]{
		\includegraphics[width=3.0cm, height=2.8cm, trim=0cm 0cm 0 0]{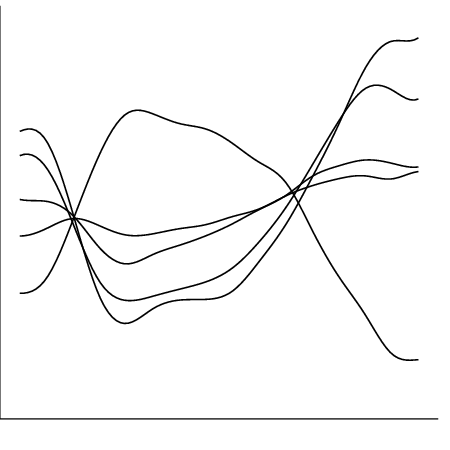}}{}\hspace{.05in}
	\stackunder[5pt]{
		\includegraphics[width=3.0cm, height=2.8cm, trim=0cm 0cm 0 0]{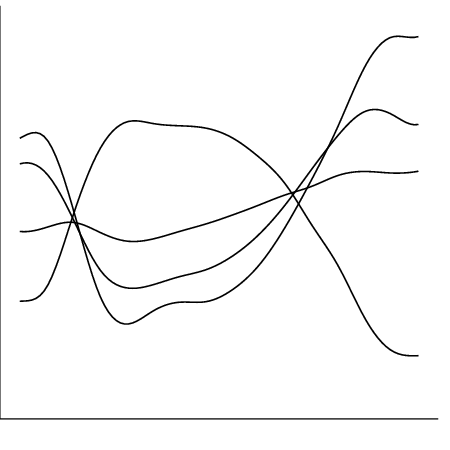}}{}\hspace{.05in}
	\stackunder[5pt]{
		\includegraphics[width=3.0cm, height=2.8cm, trim=0cm 0cm 0 0]{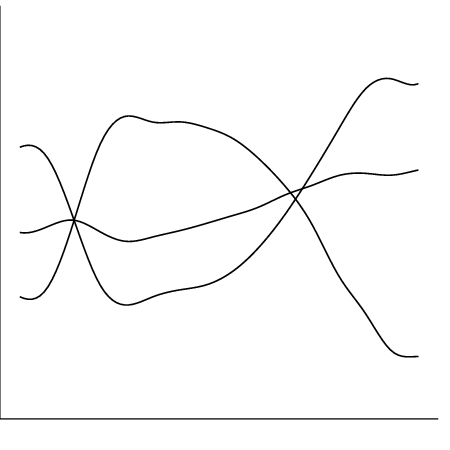}}{}
	\\
\end{adjustwidth}
\vspace{-.1in}
\caption{Cluster-specific mean curves obtained by fitting the MFRMMx model for the CAI dataset, with $A\!=\!100$, $p\!=\!15$ and $\phi\!=0.1,0.5,1,2,5$, for the sagittal and frontal planes.}
\label{fig:CAI2_mean_curves_A100_p15_phi5}
%\vspace{-.1in}
\end{figure*}

Figures \ref{fig:CAI2_A100_p15}(a) and \ref{fig:CAI2_A100_p15}(b) show the individual data sequences in the clusters identified by assuming $A\!=\!100$, $p\!=\!15$, and $\phi\!=\!1$ and $\phi\!=\!5$, respectively. The clusters are identified using the least-square method proposed by \cite{dahl2006}. Based on those figures, it is very clear that individuals with similar curves tend to be allocated to the same cluster, and all identified clusters have a non-negligible number of individuals. The cluster whose curves differ most from the others on the frontal plane, $j\!=\!5$ when $\phi\!=\!1$ and $j\!=\!3$ when $\phi\!=\!5$, remains unchanged as repulsion increases from $\phi\!=\!1$ to $\phi\!=\!5$. As for the others, four clusters merge into two, based on the similarity of their curve shapes. The repulsive prior is thus a technical tool to favor similar clusters to merge when selecting a value for $\phi$ that adequately translates what experts understand as curves being sufficiently different, which may suggest treatments that are specific to each group of the CAI dataset.

\begin{figure}
\begin{adjustwidth}{-0.8cm}{-0.5cm}
	\begin{minipage}{.5\textwidth}
		\centering
		\input{fig_CAI2_p15_A100_phi1-minipage.tex}{\hspace{.0in} (a) $\phi=1$}
	\end{minipage}
	\hspace{.1in}
	\begin{minipage}{.5\textwidth}
		\centering
		\input{fig_CAI2_p15_A100_phi5-minipage.tex}{\hspace{.0in} (b) $\phi=5$}
	\end{minipage}
	\vspace{-.1in}
	\caption{Clusters identified by the MFRMMx for the CAI dataset with $D=2$ curves by individual, with truncation parameter $A\!=\!100$, $p\!=\!15$ knots and repulsive parameter values $\phi=1$ (a) and $\phi=5$ (b). The gray lines are the individual data sequences, the black lines are the cluster-specific mean curves and $n_j$ is the number of individuals in each respective cluster $j$.}
	\label{fig:CAI2_A100_p15}
\end{adjustwidth}
\end{figure}

An application of the proposed model to analyze a univariate and ``noisier'' data set (i.e. with more variability of the data sequences around the individual-specific curves), previously analyzed by \cite{page2015}, can be found in Section \ref{supp:NBA} of the Appendix.

\FloatBarrier
\section{Final comments}\label{sec:final}
To alleviate the problem of overlapping and/or redundant clusters in the analysis of multivariate functional data, we proposed a repulsive FMM that has the ability to produce well separated groups in terms of specific curve shapes. The proposed model also allows the use of individual-specific covariates to influence the clustering process, through a  stick-breaking formulation of covariate-dependent mixing weights. The repulsive prior is defined in terms of component-specific location vectors of B-spline coefficients that define component-specific mean curves. The repulsive factor in the proposed prior depends on a B-spline curve-tailored distance. The analysis of the proposed repulsive prior behavior, presented in Section \ref{sec:repulsion}, can guide the specification of those parameters that calibrate the repulsion ($q$, $\nu$ and $\phi_d$) based on a predefined distance value that would characterize a redundancy problem. However, finding an adequate value of $\phi_d$ for the application at hand may require trying a few different values.
%it may be necessary to fit the proposed model using different values of $\phi_d$ in order to select a repulsive level suitable for each application.}
Another important contribution was the introduction of a suitable MCMC algorithm with a novel split-merge step adapted to a functional FMM model that assumes a repulsive dependence between mixture components. This strategy improved the mixing of the chains, allowing to more efficiently sample from the posterior distribution.

Adding repulsion helped addressing a critical point raised by \cite{hopkins2019}: the clusters now identified for the CAI data showed high differentiation between their mean curve shapes, which facilitates analysis and diagnosis. Corroborating some findings in the literature for non-functional data, the clustering was also improved by considering covariate-dependent mixing weights, as demonstrated in the simulation studies. We also observed an interesting interaction between the number of B-spline knots $p$ and the repulsion parameter $\phi_d$. Given the observed increase in the number of clusters when we increase the number of B-spline knots $p$, our findings suggest that increasing $p$ would require choosing larger values for $\phi_d$ to achieve the desired repulsive effect.

Future extensions of this work include the construction of different classes of repulsive priors for functional data, and the estimation of parameters involved in the repulsive prior. A possible starting point for these extensions would be to integrate the repulsive strategy proposed by \cite{cremaschi2024} into our proposed model.

%\vspace{.5in}
\newpage
\noindent\textbf{\Large Acknowledgments}
\vskip6pt
\noindent The authors would like to thank the following professors at Brigham Young University, USA: Professor Ty Hopkins from the Department of Exercise Sciences and his team for granting access to the original CAI data set, and Professor Garritt L. Page from the Department of Statistics for helpful insights. The authors would also like to thank Professor Marcos Oliveira Prates, from the Department of Statistics of the Federal University of Minas Gerais, Brazil, for the computational support that enabled the simulation studies discussed in this work.

\vskip24pt
\noindent\textbf{\Large Funding information}
\vskip6pt
\noindent The first author was supported by CAPES, grants 88887.499199/2020-00 and\linebreak 88887.803176/2023-00. The second author was supported by ANID Fondecyt Regular, grants 1220017 and 1260012. The third author is partially supported by CNPq, grants 317790/2025-0, 405025/2021-1, 304268/2021-6, and FAPEMIG, grants APQ00674-24, APQ01748-24.

\vskip24pt
\noindent\textbf{\Large Data availability}
\vskip6pt
\noindent The CAI dataset is not public, and access to it should be requested from the team that collected it.

%\vskip24pt
%\noindent\textbf{\Large Conflict of interest}
%\vskip6pt
%\noindent The authors declare no potential conflict of interests.

%%% References %%%
\bibliography{refs}

@article{metropolis1953,
	author = {Metropolis, Nicholas and Rosenbluth, Arianna W. and Rosenbluth, Marshall N. and Teller, Augusta H. and Teller, Edward},
	title = {Equation of State Calculations by Fast Computing Machines},
	journal = {The Journal of Chemical Physics},
	volume = {21},
	number = {6},
	pages = {1087-1092},
	year = {1953},
	month = {06},
	issn = {0021-9606},
	doi = {10.1063/1.1699114},
	url = {https://doi.org/10.1063/1.1699114},
	eprint = {https://pubs.aip.org/aip/jcp/article-pdf/21/6/1087/18802390/1087_1_online.pdf}
}

@article{hastings1970,
	author = {Hastings, W. K.},
	title = {Monte Carlo sampling methods using Markov chains and their applications},
	journal = {Biometrika},
	volume = {57},
	number = {1},
	pages = {97-109},
	year = {1970},
	month = {04},
	issn = {0006-3444},
	doi = {10.1093/biomet/57.1.97},
	url = {https://doi.org/10.1093/biomet/57.1.97},
	eprint = {https://academic.oup.com/biomet/article-pdf/57/1/97/23940249/57-1-97.pdf}
}

@article{ferguson1973,
	author = {Thomas S. Ferguson},
	title = {A {B}ayesian Analysis of Some Nonparametric Problems},
	volume = {1},
	journal = {The Annals of Statistics},
	number = {2},
	publisher = {Institute of Mathematical Statistics},
	pages = {209 -- 230},
	year = {1973},
	doi = {10.1214/aos/1176342360}
}

@incollection{ferguson1983,
	title = {{B}ayesian Density Estimation by Mixtures of Normal Distributions},
	editor = {M. Haseeb Rizvi and Jagdish S. Rustagi and David Siegmund},
	booktitle = {Recent Advances in Statistics},
	publisher = {Academic Press},
	pages = {287-302},
	year = {1983},
	author = {Thomas S. Ferguson},
	isbn = {978-0-12-589320-6},
	doi = {https://doi.org/10.1016/B978-0-12-589320-6.50018-6}
}

@article{lo1984,
	author = {Albert Y. Lo},
	title = {On a Class of {B}ayesian Nonparametric Estimates: {I. D}ensity Estimates},
	volume = {12},
	journal = {The Annals of Statistics},
	number = {1},
	publisher = {Institute of Mathematical Statistics},
	pages = {351 -- 357},
	year = {1984},
	doi = {10.1214/aos/1176346412}
}

@article{hubert1985,
	title={Comparing partitions},
	author={Hubert, Lawrence and Arabie, Phipps},
	journal={Journal of Classification},
	volume={2},
	number={1},
	pages={193--218},
	year={1985},
	publisher={Springer},
	doi={10.1007/BF01908075}
}

@article{hartigan1990,
	author = {J.A. Hartigan},
	title = {Partition models},
	journal = {Communications in Statistics - Theory and Methods},
	volume = {19},
	number = {8},
	pages = {2745--2756},
	year = {1990},
	publisher = {Taylor \& Francis},
	doi = {10.1080/03610929008830345}
}

@article{gilks1992,
	author = {W. R. Gilks and P. Wild},
	journal = {Journal of the Royal Statistical Society. Series C (Applied Statistics)},
	number = {2},
	pages = {337--348},
	publisher = {[Wiley, Royal Statistical Society]},
	title = {Adaptive Rejection Sampling for {G}ibbs Sampling},
	volume = {41},
	year = {1992}
}

@article{ishwaran2001,
	author = {Hemant Ishwaran and Lancelot F James},
	title = {{G}ibbs Sampling Methods for Stick-Breaking Priors},
	journal = {Journal of the American Statistical Association},
	volume = {96},
	number = {453},
	pages = {161--173},
	year = {2001},
	publisher = {Taylor \& Francis},
	doi = {10.1198/016214501750332758},
	eprint = {https://doi.org/10.1198/016214501750332758}
}

@article{hollinger2002,
	title={{Pro Basketball Prospectus 2002. Pro Basketball Forecast}},
	author={Hollinger, J},
	journal={Free Press},
	volume={380},
	pages={384},
	year={2002}
}

@article{quintana2003,
	author = {Quintana, Fernando A. and Iglesias, Pilar L.},
	title = {{Bayesian clustering and product partition models}},
	journal = {Journal of the Royal Statistical Society: Series B (Statistical Methodology)},
	volume = {65},
	number = {2},
	pages = {557-574},
	year = {2003},
	doi = {https://doi.org/10.1111/1467-9868.00402}
}

@article{lang2004,
	title={{Bayesian P-splines}},
	author={Lang, Stefan and Brezger, Andreas},
	journal={Journal of Computational and Graphical Statistics},
	volume={13},
	number={1},
	pages={183--212},
	year={2004},
	publisher={Taylor \& Francis},
	doi={10.1198/1061860043010}
}

@book{ramsay2005,
	author={Ramsay, James O and Silverman, B W},
	publisher={Springer},
	title={{Functional Data Analysis}},
	year={2005}
}

@inbook{dahl2006,
	place={Cambridge},
	title={Model-Based Clustering for Expression Data via a Dirichlet Process Mixture Model},
	booktitle={{Bayesian Inference for Gene Expression and Proteomics}},
	publisher={Cambridge University Press},
	author={David B. Dahl},
	editor={Do, Kim-Anh and Müller, Peter and Vannucci, MarinaEditors},
	year={2006},
	pages={201–218}
}

@article{bouchard2017,
	author  = {Alexandre Bouchard-C{{\^o}}t{{\'e}} and Arnaud Doucet and Andrew Roth},
	title   = {Particle {G}ibbs Split-Merge Sampling for {B}ayesian Inference in Mixture Models},
	journal = {Journal of Machine Learning Research},
	year    = {2017},
	volume  = {18},
	number  = {28},
	pages   = {1--39}
}

@article{jain2007,
	author = {Sonia Jain and Radford M. Neal},
	title = {{Splitting and merging components of a nonconjugate Dirichlet process mixture model}},
	volume = {2},
	journal = {Bayesian Analysis},
	number = {3},
	publisher = {International Society for Bayesian Analysis},
	pages = {445 -- 472},
	year = {2007},
	doi = {10.1214/07-BA219}
}

@article{petrone2009,
	author = {Petrone, Sonia and Guindani, Michele and Gelfand, Alan E.},
	title = {{Hybrid Dirichlet mixture models for functional data}},
	journal = {Journal of the Royal Statistical Society: Series B (Statistical Methodology)},
	volume = {71},
	number = {4},
	pages = {755-782},
	year = {2009},
	doi = {https://doi.org/10.1111/j.1467-9868.2009.00708.x}
}

@article{rousseau2011,
	author = {Judith Rousseau and Kerrie Mengersen},
	journal = {Journal of the Royal Statistical Society. Series B (Statistical Methodology)},
	number = {5},
	pages = {689--710},
	publisher = {[Royal Statistical Society, Wiley]},
	title = {{Asymptotic behaviour of the posterior distribution in overfitted mixture models}},
	urldate = {2025-02-17},
	volume = {73},
	year = {2011},
	ISSN = {13697412, 14679868}
}

@book{fahrmeir2011,
	author = {Fahrmeir, Ludwig and Kneib, Thomas},
	title = {Bayesian {S}moothing and {R}egression for {L}ongitudinal, {S}patial and {E}vent {H}istory {D}ata},
	publisher = {Oxford University Press},
	year = {2011},
	month = {04},
	isbn = {9780199533022},
	doi = {10.1093/acprof:oso/9780199533022.001.0001}
}

@article{muller2011,
	title={A product partition model with regression on covariates},
	author={M{\"u}ller, Peter and Quintana, Fernando and Rosner, Gary L},
	journal={Journal of Computational and Graphical Statistics},
	volume={20},
	number={1},
	pages={260--278},
	year={2011},
	publisher={Taylor \& Francis},
	doi={10.1198/jcgs.2011.09066}
}

@book{bda3,
	title={{Bayesian Data Analysis}},
	author={Gelman, Andrew and Carlin, John B and Stern, Hal S and Dunson, David B and Vehtari, Aki and Rubin, Donald B},
	year={2013},
	publisher={Chapman and Hall/CRC},
	isbn={13:978-1-4398-9820-8}
}

@inproceedings{petralia2012,
	author={Petralia, Francesca and Rao, Vinayak and Dunson, David},
	booktitle={{Advances in Neural Information Processing Systems}},
	editor={F. Pereira and C.J. Burges and L. Bottou and K.Q. Weinberger},
	publisher={Curran Associates, Inc.},
	title={Repulsive Mixtures},
    volume={25},
	year={2012}
}

@article{polson2013,
	author = {Nicholas G. Polson and James G. Scott and Jesse Windle},
	title = {Bayesian Inference for Logistic Models Using {P}ólya–Gamma Latent Variables},
	journal = {Journal of the American Statistical Association},
	volume = {108},
	number = {504},
	pages = {1339--1349},
	year = {2013},
	publisher = {Taylor \& Francis},
	doi = {10.1080/01621459.2013.829001},
	eprint = {https://doi.org/10.1080/01621459.2013.829001}
}

@article{martinez2014,
	author = {Asael Fabian Mart{\'i}nez and Rams{\'e}s H. Mena},
	title = {On a nonparametric change point detection model in {M}arkovian regimes},
	volume = {9},
	journal = {Bayesian Analysis},
	number = {4},
	publisher = {International Society for Bayesian Analysis},
	pages = {823 -- 858},
	year = {2014},
	doi = {10.1214/14-BA878}
}

@article{page2015,
	author = {Garritt L. Page and Fernando A. Quintana},
	title = {Predictions Based on the Clustering of Heterogeneous Functions via Shape and Subject-Specific Covariates},
	volume = {10},
	journal = {Bayesian Analysis},
	number = {2},
	publisher = {International Society for Bayesian Analysis},
	pages = {379 -- 410},
	year = {2015},
	doi = {10.1214/14-BA919}
}

@article{xu2016,
	author = {Xu, Yanxun and Müller, Peter and Telesca, Donatello},
	title = {{B}ayesian Inference for Latent Biologic Structure With Determinantal Point Processes (DPP)},
	journal = {Biometrics},
	volume = {72},
	number = {3},
	pages = {955-964},
	year = {2016},
	month = {02},
	issn = {0006-341X},
	doi = {10.1111/biom.12482}
}

@article{makalic2016,
	title     = {High-Dimensional Bayesian Regularised Regression with the BayesReg Package},
	author    = {Enes Makalic and Daniel F. Schmidt},
	journal   = {arXiv:1611.06649},
	year      = {2016},
	doi       = {10.48550/arXiv.1611.06649}
}

@article{page2018,
	title={Calibrating covariate informed product partition models},
	author={Page, Garritt L and Quintana, Fernando A},
	journal={Statistics and Computing},
	volume={28},
	pages={1009--1031},
	year={2018},
	publisher={Springer}
}

@article{miller2018,
	author = {Jeffrey W. Miller and Matthew T. Harrison},
	title = {Mixture Models With a Prior on the Number of Components},
	journal = {Journal of the American Statistical Association},
	volume = {113},
	number = {521},
	pages = {340--356},
	year = {2018},
	publisher = {ASA Website},
	doi = {10.1080/01621459.2016.1255636}
}

@article{quinlan2018,
	author = {José J. Quinlan and Garritt L. Page and Fernando A. Quintana},
	title = {Density regression using repulsive distributions},
	journal = {Journal of Statistical Computation and Simulation},
	volume = {88},
	number = {15},
	pages = {2931--2947},
	year = {2018},
	publisher = {Taylor \& Francis},
	doi = {10.1080/00949655.2018.1491578}
}

@article{fuquene2019,
	author = {Fúquene, Jairo and Steel, Mark and Rossell, David},
	title = {On Choosing Mixture Components via Non-Local Priors},
	journal = {Journal of the Royal Statistical Society Series B: Statistical Methodology},
	volume = {81},
	number = {5},
	pages = {809-837},
	year = {2019},
	month = {08},
	issn = {1369-7412},
	doi = {10.1111/rssb.12333}
}

@article{hopkins2019,
	title={Characterization of multiple movement strategies in participants with chronic ankle instability},
	author={Hopkins, J Ty and Son, S Jun and Kim, Hyunsoo and Page, Garritt and Seeley, Matthew K},
	journal={Journal of Athletic Training},
	volume={54},
	number={6},
	pages={698--707},
	year={2019},
	publisher={National Athletic Trainers Association},
	doi={10.4085/1062-6050-480-17}
}

@article{xie2020,
	title={Bayesian repulsive {G}aussian mixture model},
	author={Xie, Fangzheng and Xu, Yanxun},
	journal={Journal of the American Statistical Association},
	volume={115},
	number={529},
	pages={187--203},
	year={2020},
	publisher={Taylor \& Francis},
	doi={10.1080/01621459.2018.1537918}
}

@article{bianchini2020,
	author = {Ilaria Bianchini and Alessandra Guglielmi and Fernando A. Quintana},
	title = {Determinantal Point Process Mixtures Via Spectral Density Approach},
	volume = {15},
	journal = {Bayesian Analysis},
	number = {1},
	publisher = {International Society for Bayesian Analysis},
	pages = {187 -- 214},
	year = {2020},
	doi = {10.1214/19-BA1150}
}

@article{quinlan2021,
	title={On a class of repulsive mixture models},
	author={Quinlan, Jos{\'e} J and Quintana, Fernando A and Page, Garritt L},
	journal={TEST},
	volume={30},
	number={2},
	pages={445--461},
	year={2021},
	publisher={Springer},
	doi={10.1007/s11749-020-00726-y}
}

@article{rigon2021,
	title = {{Tractable Bayesian density regression via logit stick-breaking priors}},
	journal = {Journal of Statistical Planning and Inference},
	volume = {211},
	pages = {131-142},
	year = {2021},
	issn = {0378-3758},
	author = {Tommaso Rigon and Daniele Durante},
	doi = {https://doi.org/10.1016/j.jspi.2020.05.009}
}

@article{sylvia2021,
	author = {Sylvia Fr{\"u}hwirth-Schnatter and Gertraud Malsiner-Walli and Bettina Gr{\"u}n},
	title = {Generalized Mixtures of Finite Mixtures and Telescoping Sampling},
	volume = {16},
	journal = {Bayesian Analysis},
	number = {4},
	publisher = {International Society for Bayesian Analysis},
	pages = {1279 -- 1307},
	year = {2021},
	doi = {10.1214/21-BA1294}
}

@article{beraha2022,
	author = {Mario Beraha and Raffaele Argiento and Jesper Møller and Alessandra Guglielmi},
	title = {{MCMC} Computations for{ B}ayesian Mixture Models Using Repulsive Point Processes},
	journal = {Journal of Computational and Graphical Statistics},
	volume = {31},
	number = {2},
	pages = {422--435},
	year = {2022},
	publisher = {ASA Website},
	doi = {10.1080/10618600.2021.2000424}
}

@article{zhang2023,
	author = {Bohai Zhang and Huiyan Sang and Zhao Tang Luo and Hui Huang},
	title = {{Bayesian clustering of spatial functional data with application to a human mobility study during COVID-19}},
	volume = {17},
	journal = {The Annals of Applied Statistics},
	number = {1},
	publisher = {Institute of Mathematical Statistics},
	pages = {583 -- 605},
	year = {2023},
	doi = {10.1214/22-AOAS1643}
}

@article{rigon2023,
	author = {Rigon, Tommaso},
	title = {An enriched mixture model for functional clustering},
	journal = {Applied Stochastic Models in Business and Industry},
	volume = {39},
	number = {2},
	pages = {232-250},
   	year = {2023},
	doi = {https://doi.org/10.1002/asmb.2736}
}

@article{zhang2023review,
	author = {Zhang, Mimi and Parnell, Andrew},
	title = {Review of Clustering Methods for Functional Data},
	journal   = {ACM Transactions on Knowledge Discovery from Data},
	year = {2023},
	publisher = {Association for Computing Machinery},
	address = {New York, NY, USA},
	volume = {17},
	number = {7},
	pages     = {1-34},
	issn = {1556-4681},
	doi = {10.1145/3581789},
	articleno = {91}
}

@article{ghilotti2024,
	author = {Ghilotti, L and Beraha, M and Guglielmi, A},
	title = {{B}ayesian clustering of high-dimensional data via latent repulsive mixtures},
	journal = {Biometrika},
	pages = {asae059},
	year = {2024},
	month = {11},
	issn = {1464-3510},
	doi = {10.1093/biomet/asae059}
}

@article{zhong2024,
	title     = {Bayesian spatial functional data clustering: applications in disease surveillance},
	author    = {Ruiman Zhong and Erick A. Chac{\'o}n-Montalv{\'a}n and Paula Moraga},
	journal   = {arXiv:2407.12633},
	year      = {2024},
	doi       = {10.48550/arXiv.2407.12633}
}

@article{liang2024,
	title={A {B}ayesian nonparametric approach for clustering functional trajectories over time},
	author={Liang, Mingrui and Koslovsky, Matthew D and H{\'e}bert, Emily T and Kendzor, Darla E and Vannucci, Marina},
	journal={Statistics and Computing},
	volume={34},
	number={6},
	pages={215},
	year={2024},
	publisher={Springer}
}

@article{cremaschi2024,
	author = {Cremaschi, Andrea and Wertz, Timothy M and De Iorio, Maria},
	title = {Repulsion, chaos, and equilibrium in mixture models},
	journal = {Journal of the Royal Statistical Society Series B: Statistical Methodology},
    volume = {87},
	number = {2},
	pages = { 389--432},
	year = {2025},
	month = {10},
	issn = {1369-7412},
	doi = {10.1093/jrsssb/qkae096}
}

@article{toto2025,
	title     = {Bayesian local clustering of functional data via semi-Markovian random partitions},
	author    = {Giovanni Toto and Antonio Canale},
	journal   = {arXiv:2503.08881},
	year      = {2025},
	doi       = {10.48550/arXiv.2503.08881}
}

@book{muller2015bayesian,
	title     = {Bayesian Nonparametric Data Analysis},
	author    = {Peter M{\"u}ller and Fernando Andr{\'e}s Quintana and Alejandro Jara and Tim Hanson},
	series    = {Springer Series in Statistics},
	publisher = {Springer},
	year      = {2015},
	doi       = {10.1007/978-3-319-18968-0},
	isbn      = {978-3-319-18967-3},
	edition   = {1},
	pages     = {xiv+193}
}

@article{song2025repulsive,
	title   = {Repulsive Mixture Model with Projection Determinantal Point Process},
	author  = {Song, Ziyi and Camerlenghi, Federico and Shen, Weining and Guindani, Michele and Beraha, Mario},
	journal = {arXiv:2510.08838},
	year    = {2025},
	doi     = {10.48550/arXiv.2510.08838}
}

@article{hayashida2025repulsive,
	title   = {Repulsive g-Priors for Regression Mixtures},
	author  = {Hayashida, Yuta and Sugasawa, Shonosuke},
	journal = {arXiv:2512.16276},
	year    = {2025},
	doi     = {10.48550/arXiv.2512.16276}
}

%%% Appendix %%%
%\vspace{.5in}
\newpage
\noindent\textbf{\Large Appendix}

\appendix

%\newpage
\section{Posterior sampling of the mixing weights}\label{supp:covariates}

Following \cite{rigon2021}, we assume that the covariate-dependent mixing weights $\pi_j(x_i)$ follow a stick-breaking construction with a logit transformation of the conditional probabilities $\nu_j(\bm{x}_i)$, given by

\begin{equation}
	\bm{x}^\T_i\ba_j=\log\big(\nu_j(\bm{x}_i)/(1-\nu_j(\bm{x}_i))\big),
\end{equation}

\noindent where $\bm{x}_i$ are covariates related to the weights of individual $i$ and  $\ba_j$ is the $\Ell \times 1$ component-specific column vector of coefficients $\ba_j$, $i=1,\dots,m$ and $j=1,\dots,J$. Under this structure, the mixing weights are given by

\begin{equation}\label{def_weights_x}
	\pi_j(\bm{x}_i \v \ba_j)
	=
	\frac{\exp\Big(\bm{x}^\T_i\ba_j\Big)}{1+\exp\Big(\bm{x}^\T_i\ba_j\Big)}
	\prod_{\ell=1}^{j-1}\frac{1}{1+\exp\Big(\bm{x}^\T_i\ba_\ell\Big)}.
\end{equation}

\noindent The prior distribution of the covariate-dependent mixing weights $\pi_j(\bm{x}_i \v \ba_j)$ is specified indirectly  by eliciting  the prior distribution  for $\ba_j$. We assume that $\ba_j\overset{\text{ind}}{\sim} \N_\Ell(\bmu_\alpha,\Sigma_\alpha)$, for $j=1,\dots,J$. Samples of the posterior of $\pi_j(\bm{x}_i)$ are obtained by replacing  the generated values of $\ba_j$ in  \eqref{def_weights_x}. Such values are generated from
\eqref{full_alpha_j}, in turn obtained by marginalizing the joint posterior distribution of $(\ba_1,\dots,\ba_J)$:

\begin{equation}\label{full_alpha}
	\begin{aligned}
		f(\ba_1,\dots,\ba_J \v X,\bm{z})
		&\;\propto\;
		\prod_{i=1}^{m}\prod_{j=1}^{J}
		\pi_j(\bm{x}_i)^{\bm{1}(z_i=j)}
		\;\times\;
		\prod_{j=1}^{J}
		\,\N_\Ell(\ba_j \,;\, \bm{\mu}_\alpha,\Sigma_\alpha)
		\\[.1in]
		&\hspace{-3cm}\;\propto\;
		\prod_{i=1}^{m}\prod_{j=1}^{J}
		\left[
		\frac{\exp\left(\bm{x}^\T_i\ba_j\right)^{\bm{1}(z_i=j)}}
		{\big(1+\exp\left(\bm{x}^\T_i\ba_j\right)\big)^{\bm{1}(z_i=j)}}
		\prod_{\ell=1}^{j-1}\frac{1}
		{\big(1+\exp\left(\bm{x}^\T_i\ba_\ell\right)\big)^{\bm{1}(z_i=j)}}
		\right]
		\;\times\;
		\prod_{j=1}^{J}
		\,\N_\Ell(\ba_j \,;\, \bm{\mu}_\alpha,\Sigma_\alpha),
		\\[.1in]
		&\hspace{-3cm}\;=\;
		\prod_{i=1}^{m}\prod_{j=1}^{J}
		\left[
		\frac{\exp\left(\bm{x}^\T_i\ba_j\right)^{\bm{1}(z_i=j)}}
		{\big(1+\exp\left(\bm{x}^\T_i\ba_j\right)\big)^{\bm{1}(z_i\ge j)}}
		\right]
		\;\times\;
		\prod_{j=1}^{J}
		\,\N_\Ell(\ba_j \,;\, \bm{\mu}_\alpha,\Sigma_\alpha),
	\end{aligned}
\end{equation}

\noindent where $X$ is a $m\tms\Ell$ matrix which the $i$th row is the individual-specific covariate vectors $\bm{x}_i$ for individual $i$. Consequently, for $j=1,\dots,J$, the posterior full conditional of $\ba_j$ is

\begin{equation}\label{full_alpha_j}
	\begin{aligned}
		f(\ba_j \v X,\bm{z})
		&\;\propto\;
		\prod_{i=1}^{m}
		\frac{\exp\left(\bm{x}^\T_i\ba_j\right)^{\bm{1}(z_i=j)}}
		{\big(1+\exp\left(\bm{x}^\T_i\ba_j\right)\big)^{\bm{1}(z_i\ge j)}}
		\;\times\;
		\exp\left(-\frac{1}{2}(\ba_j\!-\bmu_\alpha)^\T\Sigma_\alpha^{-1}(\ba_j\!-\bmu_\alpha)\right).
	\end{aligned}
\end{equation}

As described in \cite{rigon2021}, the Pólya-Gamma (PG) data augmentation method introduced by \cite{polson2013} {is useful to efficiently sample} from the full conditional distribution in \eqref{full_alpha}. The method is based on the following theorem:

%\begin{theorem}\textup{\textbf{(Theorem 1 of }\cite{polson2013}\textbf{)}}\label{theorem_polson_1}
%\begin{theorem}[{\cite[Theorem~1]{polson2013}}]\label{theorem_polson_1}
\begin{theorem}[\cite{polson2013}, Theorem~1]\label{theorem_polson_1}
	Let $f(\omega)$ denote the density of the random variable
	$\omega\sim\textup{PG}(b,0), b>0$. Then the following integral identity holds for all $a\in\mathbb{R}\!:$
	
	\begin{equation}\label{theorem_polson_1_eq}
		\begin{aligned}
			\frac{\left(\exp(\psi)\right)^a}
			{\left(1+\exp(\psi)\right)^b}
			&=
			2^{-b}\exp(k\psi)
			\int_{0}^{\infty}
			\!\!\exp(-\omega\psi^2/2)f(\omega) \,\textup{d} \omega,
			%\\[.1in]
		\end{aligned}
	\end{equation}
	
	\noindent where $k=a-b/2$.
\end{theorem}

From Theorem \ref{theorem_polson_1}, it follows that the posterior full conditional distribution \eqref{full_alpha} is

\begin{equation}\label{full_alpha_pg}
	\begin{aligned}
		f(\ba_j \v \bm{z},X)
		&\propto
		f(\ba_j)
		\prod_{i=1}^{m}
		\left[
		\int_{0}^{\infty}
		\!\!\exp\!\left(k_i\bm{x}^\T_i\ba_j - \frac{1}{2}\omega_i(\bm{x}^\T_i\ba_j)^2\right)
		f(\omega_i)
		\d\omega_i
		\right]^{\bm{1}(z_i\ge j)}
		%\\[.1in]
	\end{aligned}
\end{equation}

\noindent where $k_i \,=\, \bm{1}(z_i=j)-0.5\bm{1}(z_i\ge j)$ is equal to  $1/2$, if  $z_i=j$ and  is equal to $-1/2$, if $z_i>j$. Consequently,  we can sample from \eqref{full_alpha} by sequentially sampling from  ${\omega_i \v z_i=j,\ba_j,\bm{x}_i\sim\text{PG}(1,\bm{x}^\T_i\ba_j)}$ whose density is given by

\begin{equation}\label{full_wi_alpha}
	\begin{aligned}
		f(\omega_i \v z_i=j, \ba_j,\bm{x}_i)
		&\propto
		\exp\!\left(-\frac{1}{2}\omega_i(\bm{x}^\T_i\ba_j)^2\right)
		f(\omega_i),
		\\[.1in]
	\end{aligned}
\end{equation}

\noindent for all $i$ such that $z_i\ge j$ and from $\ba_j \v \bm{z},X,\{\omega_i\!:\!zi\ge j\}\sim\N(m_\alpha,V_\alpha)$ where $m_\alpha=V_\alpha\left(0.5(\sum_{i:z_i=j}\bm{x}_i-\!\sum_{i:z_i>j}\bm{x}_i)+\Sigma^{-1}_\alpha\mu_\alpha\right)$ and $V_\alpha=\left(\sum_{i:z_i\ge j}\omega_i\bm{x}_i\bm{x}^\T_i+\Sigma^{-1}_\alpha\right)^{-1}$, with density

\begin{equation}\label{full_alpha_wi}
	\begin{aligned}
		f(\ba_j \v \bm{z},X,\{\omega_i\!:\!zi\ge j\})
		&\propto
		f(\ba_j)
		\prod_{i=1}^{m}
		\left[
		\exp\!\left(k_i\bm{x}^\T_i\ba_j - \frac{1}{2}\omega_i(\bm{x}^\T_i\ba_j)^2\right)
		\right]^{\bm{1}(z_i\ge j)}
		\\[.1in]
		&\hspace{-1.5in}\propto
		\exp\!\left[
		-\frac{1}{2}\,
		{\ba_j}^\T\!
		\left(\sum_{i:z_i\ge j}\!\!\omega_i\bm{x}_i\bm{x}^\T_i+\Sigma^{-1}_\alpha\right)
		\ba_j
		+{\ba_j}^\T\!
		\left(\frac{\sum_{i:z_i=j}\bm{x}_i-\!\sum_{i:z_i>j}\bm{x}_i}{2}+\Sigma^{-1}_\alpha\mu_\alpha\right)
		\right].
		\\[.1in]
	\end{aligned}
\end{equation}

%From equation (5) in \cite{polson2013}, it follows  that the distribution in \eqref{full_wi_alpha} is ${\omega_i \v z_i=j,\ba_j,\bm{x}_i\sim\text{PG}(1,\bm{x}^\T_i\ba_j)}$ . The full conditional density in \eqref{full_alpha_wi} is that of a $\text{N}(m_\alpha,V_\alpha)$ where $V_\alpha=\left(\sum_{i:z_i\ge j}\!\!\omega_i\bm{x}_i\bm{x}^\T_i+\Sigma^{-1}_\alpha\right)^{-1}$ and $m_\alpha=V_\alpha\left(0.5(\sum_{i:z_i=j}\bm{x}_i-\!\sum_{i:z_i>j}\bm{x}_i)+\Sigma^{-1}_\alpha\mu_\alpha\right).$

\noindent To sample from the full conditional distribution of $w_i$ in \eqref{full_wi_alpha} we use the Polya-Gamma sampler implementation available in the \texttt{R} package \texttt{pgdraw} \citep{makalic2016}.

\vskip24pt
\noindent {\bf Remark:} If no covariates are considered to model the mixing weights, the following modifications to equation \eqref{eq_posterior} must be made: the  mixing weights $\pi_j(\bm{x}_i\v\ba_j)$ must be replaced by $\pi_j$ and the distributions for $\ba_j$, $\N_{\Ell}(\ba_j\,;\,\bm{\mu}_\alpha,\Sigma_\alpha)$, must be replaced by the Dirichlet kernel $\pi_j^{\gamma_j-1}$.

%\newpage
\section{Sampling from the joint posterior}\label{supp:posterior}

Denote the set of all model parameters by
$$
\bm{\Psi}=
\left(\bigcup_{d=1}^{D}\Big\{
\Theta_d,\btau_d,\blam_d,\bm{\mu}_d,\mu_{0d},\sigma^2_{0d}
\Big\}\right)
\bigcup
\left(\bigcup_{i=1}^{m}\Big\{
\Bzi,\Bi,\Si,z_i
\Big\}\right)
\bigcup \Big\{\ba_1,\dots,\ba_J\Big\}.
$$

\noindent The joint posterior distribution of the proposed model is given by

\begin{equation}\label{eq_posterior}
	\begin{aligned}
		f(\bm{\Psi} \v \mathcal{Y}_1,\dots,\mathcal{Y}_m,\bm{x}_1,\dots,\bm{x}_m)
		%		f(\bm{\Psi} \v \bY,\bm{x}_1,\dots,\bm{x}_m)
		&\propto\;
		\Bigg[\prod_{i=1}^{m}
		\,\N_{n_i D}(\vect(\Yi) \,;\, \vect(\Bzi \!+\! H_i\Bi) \,,\, \Si)
		\Bigg]\\%[.1in]
		&\hspace{-2.1in}\times\;
		\Bigg[\prod_{i=1}^{m}
		\,\IW(\Si \,;\, \Sigma_0,\omega)
		\Bigg]
		\times
		\Bigg[\prod_{d=1}^{D}
		\Bigg(\prod_{i=1}^{m}\,\N(\beta_{0id} ;\, \mu_{0d},\sigma^{2}_{0d})\Bigg)
		\,\N(\mu_{0d} ;\, 0,s^2_0)
		\,\IG(\sigma^{2}_{0d} ;\, a_0,b_0)
		\Bigg]
		\\[.1in]
		&\hspace{-2.10in}\times\;
		\Bigg[\prod_{i=1}^{m}\prod_{j=1}^{J}
		\left(\,\prod_{d=1}^{D}
		\N_p(\bbt_{id} ;\, \bth_{jd} \,, \lambda^{2}_{jd} I_p)
		\right)^{\bm{1}(z_i=j)}
		\mkern-25mu\pi_j(\bm{x}_i\v\ba_j)^{\bm{1}(z_i=j)}
		\Bigg]
		\times
		\Bigg[\prod_{j=1}^{J}
		\,\N_{\Ell}(\ba_j \,;\, \bm{\mu}_\alpha,\Sigma_\alpha)
		\Bigg]
		\\[.1in]
		&\hspace{-2.10in}\times\;
		\Bigg[\prod_{d=1}^{D}
		\Bigg(\prod_{j=1}^{J}
		\,\N_p(\bth_{jd} ;\, \bm{\mu}_d ,\, \tau^2_{jd} K^{-1})\,
		\,f(\lambda^{2}_{jd})\,
		\,\IG(\tau^2_{jd} ;\, a_\tau,b_\tau)
		\Bigg)
		h\!\left(\Theta_d\right)
		\N_p(\bm{\mu}_d ;\, \bm{0},s^2_\mu I_p)
		\Bigg],
		\\[.0in]
	\end{aligned}
\end{equation}
where $f(\lambda^2_{jd})\propto(\lambda^2_{jd})^{-1/2}\,\bm{1}\left(\lambda^2_{jd}\!\in\!(0,A_d^2)\right)$, which follows from \eqref{prior_lam2}, and $h\!\left(\Theta_d\right)$  is given in equations \eqref{def_h} and \eqref{eq_distance_approx} in the main text. If no covariates are considered to influence the mixing weights, the following modifications to \eqref{eq_posterior} must be made: the  mixing weights $\pi_j(\bm{x}_i\v\ba_j)$ must be replaced by $\pi_j$ and the distributions for $\ba_j$, $\N_{\Ell}(\ba_j\,;\,\bm{\mu}_\alpha,\Sigma_\alpha)$, must be replaced by the Dirichlet kernel $\pi_j^{\gamma_j-1}$. In what follows, we give a detailed description of the  proposed MCMC algorithm to sample from the posterior distribution in \eqref{eq_posterior}. One important feature of the proposed algorithm is the split-merge step, that has proved essential for more accurate cluster identification, and is described in Section \ref{supp:SM}.

\subsection{MCMC algorithm}\label{supp:mcmc}

For the MCMC algorithm described next, consider the following notation:

\begin{itemize}
	
	\item[\tb] $\Theta_j=[\bth_{j1},\dots,\bth_{jD}]$,

	\vspace{.05in}
	\item[\tb] $\btau_j=(\tau^2_{j1},\dots,\tau^2_{jD})$
	\quad and \quad
	$T_j=\text{diag}(\btau_j)\otimes I_p$,
	
	\vspace{.05in}
	\item[\tb] $\blam_j=(\lambda^2_{j1},\dots,\lambda^2_{jD})$
	\quad and \quad
	$L_j=\text{diag}(\blam_j)\otimes K^{-1}$,
	
	\vspace{.05in}
	\item[\tb] $\bmu_0=(\mu_{01},\dots,\mu_{0D})$,
%	\quad and \quad
%	${\bm{\sigma}_0=(\sigma^{2}_{01},\dots,\sigma^{2}_{0D})}$

	\vspace{.05in}
	\item[\tb] $M=\left[\bmu_1,\dots,\bmu_D\right]$,
	
	\vspace{.05in}
	\item[\tb] $n_j\!=\displaystyle\sum_{i=1}^{m}\bm{1}(z_i=j)$, that is, the number of individuals allocated to mixture component $j$, for $j=1,\dots,J$,
	
	\vspace{.05in}
	\item[\tb] $\E_i=\Yi-\Bzi-H_i\Bi$,

\end{itemize}

\noindent {where we use $[\;]$ to denote matrices, $(\;)$ for vectors and diag(\textbf{v}) represents} a diagonal matrix with main diagonal equal to vector \textbf{v}. To sample from the joint posterior in \eqref{eq_posterior}, we propose the following MCMC algorithm:

\vskip12pt
\begin{enumerate}[leftmargin=*]
	
	\vskip12pt
	\item Execute one split-merge step described in Section \ref{supp:SM}.

	\vskip12pt
	\item In the presence of covariates:
	
	\begin{enumerate}
		
		\vspace{.05in}
		\item[i.] for $j=1,\dots,J$, update $\ba_j$ from \eqref{full_alpha_j} as described in Section \ref{supp:covariates};
		
		\vspace{.05in}
		\item[ii.] for $i=1,\dots,m$:
		
		\begin{itemize}
			\item[\tb] compute $\pi_j(\bm{x}_i\v\ba_j)$ from \eqref{def_weights_x}, for $j=1,\dots,J$;
			
			\vspace{.02in}
			\item[\tb] update $z_i$ from the full conditional (marginalized with respect to $\mathcal{B}_1,\dots,\mathcal{B}_m$),

			\vspace{-.1in}
			\begin{equation}\label{full_zi_Int}
			\hspace{-2cm}\begin{aligned}
				%&\Pr(z_i=j \v \Si, \Theta_j, \bm{\Lambda}, X, \ba_1,\dots,\ba_J)
				%&\Pr(z_i=j \v -)\\[.1in]
				\Pr(z_i=j \v -)
				&\,\propto\,
				\pi_j(\bm{x}_i\v\ba_j)
				\int
				\N(\vect(\Yi) ;\, \vect(\Bzi \!+\! H_i\Bi) ,\, \Si)
				\;\N(\vect(\Bi) ;\, \vect(\Theta_j) ,\, L_j)
				\;\d(\vect(\Bi))
				\\[.1in]
				&\,\propto\,
				\pi_j(\bm{x}_i\v\ba_j)
				\;|L_j|^{-\frac{1}{2}}
				\exp\left\{-\frac{1}{2} \,\vect(\Theta_j)^\T L_j^{-1}\vect(\Theta_j)\right\}
				\;|V_i|^\frac{1}{2}
				\exp\left\{\frac{1}{2} \,m_i^\T V_i \,m_i\right\},
			\end{aligned}
			\end{equation}
			
			\vspace{.1in}
			\noindent where \;$V_i^{-1}\!=\!(\Si^{-1}\!\otimes H_i^\T) + L_j^{-1}$\;  and \;$m_i\!=\!(\Si^{-1}\!\otimes H_i^\T)\,\vect(\Yi-\Bzi) + L_j^{-1}\vect(\Theta_j)$.

%			\noindent For the univariate or independent specifications of the proposed model, {we replace \eqref{full_zi_Int} by the following, cheaper to evaluate expression:}
%			
%			\vspace{-.3in}
%			\begin{equation}\label{full_zi_Int_d}
%				\begin{aligned}
%					%&\Pr(z_i=j \v \Bzi, \Si, \Theta_j, \bm{\Lambda}, X, \ba_1,\dots,\ba_J)
%					&\Pr(z_i=j \v -)
%					\\[.05in]
%					&\hspace{-0.3in}\,\propto\,
%					\pi_j(\bm{x}_i\v\ba_j)
%					\,\prod_{d=1}^{D}\int
%					\N\left( \byid ;\, \beta_{0id}\bm{1}_{n_i} \!+\! H_i\bbt_{id} ,\, \sigma_{id}^2 \right)
%					\,\N\left(\bbt_{id} ;\, \bth_{jd} ,\, \lambda_{jd}^2\right)
%					\,\d\bbt_{id}
%					\\[.05in]
%					&\hspace{-0.3in}\,\propto\,
%					\pi_j(\bm{x}_i\v\ba_j)
%					\,\prod_{d=1}^{D}
%					\,\left(\lambda_{jd}^2\right)^{-\frac{p}{2}}
%					\exp\left\{-\frac{1}{2} \;\theta_{jd}^\T \;\lambda_{jd}^{-2}\; \theta_{jd}\right\}
%					\,|V_i|^\frac{1}{2}
%					\exp\left\{\frac{1}{2} \,m_i^\T V_i \,m_i\right\},
%				\end{aligned}
%			\end{equation}
%			
%			\vspace{.1in}
%			\noindent with \;$V_i^{-1}\!=\!H_i^\T H_i\sigma_{id}^{-2} + \lambda_{jd}^{-2}I_p$\; and \;$m_i\!=\!H_i^\T(\byid\!-\beta_{0id}\bm{1}_{n_i})\sigma_{id}^{-2} + \lambda_{jd}^{-2}I_p\bth_{jd}$\,.
			
		\end{itemize}
		
	\end{enumerate}

	\vspace{.1in}
	\noindent In the absence of covariates:

	\begin{enumerate}
	
		\vspace{.05in}
		\item[i.] for $j=1,\dots,J$, update $\pi_j$ from $\pi_j\v- \,\sim \text{Dirichlet}\big(\gamma_j+n_j\big)$;
		
%		\vspace{-.2in}
%		\begin{equation}\label{full_pi}
%		\pi_j\v- \;\sim\; \text{Dirichlet}\big(\gamma_j+n_j\big)
%		\end{equation}
		
%		%\vspace{.05in}
%		\item[ii.] for $i=1,\dots,m$,  update $z_i$ from \eqref{full_zi_Int} or \eqref{full_zi_Int_d}, with $\pi_j(\bm{x}_i\v\ba_j)$ replaced by $\pi_j$.

		\vspace{.05in}
		\item[ii.] for $i=1,\dots,m$, update $z_i$ from \eqref{full_zi_Int}, with $\pi_j(\bm{x}_i\v\ba_j)$ replaced by $\pi_j$.

	\end{enumerate}

	\vskip24pt
	\item For $j=1,\dots,J$, update $\Theta_j$ from the full conditional (marginalized with respect to $\mathcal{B}_1,\dots,\mathcal{B}_m$)
	
%	%\vspace{-.1in}
%	\begin{equation}\label{full_theta_dep}
%	\begin{aligned}
%		f(\bth_{jd} \bv -)
%		&\;\propto\;
%		\N(\bth_{jd} \,; V_j\,m_j \,, V_j )\; h(\Theta_d),
%		\quad d=1,\dots,D,
%%		\\[.1in]
%%		V_j^{-1}
%%		&\;=\;
%%		\tau^{-2}_{jd} K \;+\; \lambda^{-2}_{jd} n_j I_p
%%		\\[.1in]
%%		m_j
%%		&\;=\;
%%		\tau^{-2}_{jd} K\bmu_d \;+\; \lambda^{-2}_{jd}\sum_{i:z_i=j}\!\bbt_{id}
%	\end{aligned}
%	\end{equation}
%	
%	\vspace{.1in}
%	\noindent with \;$V_j^{-1}\!=\!\tau^{-2}_{jd}K + \lambda^{-2}_{jd}n_j I_p$\; and
%	\;$m_j\!=\!\tau^{-2}_{jd} K\bmu_d + \lambda^{-2}_{jd}\!\displaystyle\sum_{i:z_i=j}\!\bbt_{id}$
	
	\vspace{-.1in}
	\begin{equation}\label{full_Theta_dep_Int}
	\begin{aligned}
		f(\vect(\Theta_j) \v -)
		&\;\propto\;
		\N\left(\vect(\Theta_j) \,; V_j \, m_j \,,\, V_j \right)\;\prod_{d=1}^{D}h(\Theta_d)\,,\\
		&\hspace{-1.2in}\text{with}
		\\[-.1in]
		V_j^{-1}\!
		&\;=\;
		%\Bigg(
		T_j^{-1} + n_j L_j^{-1} - L_j^{-1}\Bigg(\sum_{i:z_i=j}V_i\Bigg)L_j^{-1}\,,
		%\Bigg)^{-1},
		\\%[.05in]
		m_j
		&\;=\;
		T_j^{-1}\vect(M) + L_j^{-1}\Bigg(\sum_{i:z_i=j}\!V_i\,(\Si^{-1}\otimes H_i^\T)\,\vect(\Yi-\Bzi)\Bigg)\,,
		\\[.05in]
		V_i^{-1}\!
		&\;=\;
		(\Si^{-1}\!\otimes H_i^\T) + L_j^{-1}\,.
	\end{aligned}
	\end{equation}

	\vskip12pt
	\item For $i=1,\dots,m$, update $\Bi$ from the full conditional

	\begin{equation}\label{full_beta_dep}
	\begin{aligned}
		\vect(\Bi) \bv - %\Bzi,\Si,z_i,\Theta_{z_i},\blam_{z_i},\Yi
		&\;\sim\;
		\text{N}( V_i\,m_i , V_i )\,,
		\\[.01in]
		&\hspace{-1.2in}\text{with}
		\\[-.06in]
		V_i^{-1}&\;=\;
		\Si^{-1}\!\otimes H_i^\T H_i + (\text{diag}(\blam_{z_i})\!\otimes I_p)^{-1}\,,
		\\[.1in]
		m_i&\;=\;
		(\Si^{-1}\!\otimes H_i^\T)\vect(\Yi-\Bzi) + (\text{diag}(\blam_{z_i})\!\otimes I_p)^{-1} \vect(\Theta_{z_i})\,.
	\end{aligned}
	\end{equation}
	
	%\noindent where $\text{diag}(\blam_{z_i})$ denotes a $D{\times}D$ diagonal matrix with main diagonal elements $\lambda^{2}_{z_i 1},\dots,\lambda^{2}_{z_i D}$.
	
%	\noindent For the univariate or independent specifications of the proposed model, {we replace \eqref{full_beta_dep} by the following, cheaper to evaluate expression:} 
%	
%	\vspace{-.1in}
%	\begin{equation}\label{full_beta_ind}
%		\bbt_{id} \bv -
%		%\beta_{0id}\,, \sigma^{2}_{id}\,, \bth_{z_i d}\,, \lambda^{2}_{z_i d}\,, \byid
%		\;\sim\;
%		\text{N}( V_i\,m_i \,, V_i ),
%	\end{equation}
%	where $V_i^{-1}\;=\; H_i^\T H_i\sigma^{-2}_{id} \;+\; I_p\lambda^{-2}_{z_i d}$ and $m_i\;=\; H_i^\T(\byid-\beta_{0id}\bm{1}_p)\sigma^{-2}_{id} \;+\; \bth_{z_i d}\lambda^{-2}_{z_i d}$.	

	%\vskip24pt
	\vskip12pt
	\item For $j=1,\dots,J$, update $\tau^{2}_{jd}$ from the full conditional

	\vspace{-.05in}
	\begin{equation}\label{full_tau2}
		\tau^{2}_{jd} \bv - %\bth_{jd}, \bmu_d
		\;\sim\;
		\text{IG}\big(
		a_\tau + 0.5p \;\;,\;\;
		b_\tau + 0.5\,(\bth_{jd}-\bmu_d)^\T K (\bth_{jd}-\bmu_d)
		\big),
		\quad d=1,\dots,D.
	\end{equation}

	%\vskip24pt
	\vskip12pt
	\item For $j=1,\dots,J$, update $\lambda^{2}_{jd}$ from the full conditional

	\begin{equation}\label{full_lam2}
		\lambda^{2}_{jd} \bv - %\bth_{jd}\,, \bm{z}, \{\bbt_{id}: z_i\!=\!j\}\\[.1in]
		\;\sim\;
		\text{IG}\Big(0.5(n_jp-1)
		\;,\;
		0.5\!\displaystyle\sum_{i:z_i=j} (\bbt_{id}-\bth_{jd})^\T (\bbt_{id}-\bth_{jd})\Big)
		\;\bm{1}\big(\lambda^{2}_{jd}\!\in\!(0,A_d^2)\big),
		\quad d=1,\dots,D.
	\end{equation}

	%\vskip24pt
	\vskip12pt
	\item Update $\bmu_d$ from the full conditional

	\vspace{-.1in}
	\begin{equation}\label{full_mu}
	\begin{aligned}
		\bmu_d \v - %\btau_d, \Theta_d
		&\;\sim\;
		\text{N}_p (V_\mu\,m_\mu \,,\, V_\mu),
		\quad d=1,\dots,D,
%		\\[.1in]
%		V_\mu&\;=\;
%		K\sum_{j=1}^{J}\tau^{-2}_{jd}+I_p s^{-2}_\mu
%		\\[.05in]
%		m_\mu&\;=\;
%		\sum_{j=1}^{J}\bth_{jd}\,\tau^{-2}_{jd}
	\end{aligned}
	\end{equation}
	
	\vspace{.05in}
	\noindent where \;$V_\mu \!=\! K\sum_{j=1}^{J}\tau^{-2}_{jd}+I_p s^{-2}_\mu$\; and 
	\;$m_\mu\!=\!\sum_{j=1}^{J}\bth_{jd}\,\tau^{-2}_{jd}$.

	\vskip24pt
	\item For $i=1,\dots,m$, update $\Bzi$ from the full conditional
	
	\begin{equation}\label{full_beta0i_dep}
	\begin{aligned}
		\bbt_{0i} \bv -
		&\;\sim\;
		\text{N}( V_{0i}\,m_{0i} \,, V_{0i} ),
%		\\[.05in]
%		&\hspace{-1.2in}\text{with}
%		\\[-.02in]
%		V_{0i}^{-1}&\;=\;
%		n_i\Si^{-1} \!+ \Sigma_0^{-1},
%		\\[.1in]
%		m_{0i}&\;=\;
%		\Si^{-1}(\Yi-H_i\Bi)^\T{1}_{n_i} + \Sigma_0^{-1}\bmu_0,
	\end{aligned}
	\end{equation}
	
	\vspace{.05in}
	\noindent where \;$V_{0i}^{-1}\!=\!n_i\Si^{-1} \!+ \Sigma_0^{-1}$\;,  \;$m_{0i}\!=\!\Si^{-1}(\Yi-H_i\Bi)^\T{1}_{n_i} + \Sigma_0^{-1}\bmu_0$\,, and  ${(\Yi-H_i\Bi)^\T{1}_{n_i}}$ is a column vector of size $D$ with each element $d$ equal to the sum of the $d$th column of $\Yi-H_i\Bi$, for $d=1,\dots,D$. 
	
%	For the univariate or independent specifications of the proposed model, 	{we replace \eqref{full_beta0i_dep} by the following, cheaper to evaluate expression:} 
%
%	\vspace{-.1in}
%	\begin{equation}\label{full_beta0i_ind}
%		\begin{aligned}
%		\beta_{0id} \bv - %\bbt_{id}, \sigma^{2}_{id}, \mu_{0d}, \sigma^{2}_{0d}, \byid
%		&\;\sim\;
%		\text{N}( V_{0i}\,m_{0i} , V_{0i} ),
%		\quad d=1,\dots,D,
%		\end{aligned}
%	\end{equation}
%	where \;$V_{0i}\!=\!\sigma^{2}_{0d}\,\sigma^{2}_{id} \;\big/\! \left({n_i\,\sigma^{2}_{0d} + \sigma^{2}_{id}}\right)$\;, \;$m_{0i}\!=\!\displaystyle\sum_{t=1}^{n_i}\left(y_{idt}-h^\T_{it}\bbt_{id}\right)/\sigma^{2}_{id} \,+\, \mu_{0d}\big/\sigma^{2}_{0d}$\, and $h_{it}$ represents the $t$th row  of the matrix $H_i$.

	\vspace{.5in}
	\item For $i=1,\dots,m$, update $\Si$ from the full conditional

	\vspace{-.05in}
	\begin{equation}\label{full_Si}
		\Si \bv - %\Bi,\Bzi,\Yi
		\;\sim\;
		%\text{IW}( \omega+n_i \,,\, (\Yi\!-\!H_i\Bi\!-\!\Bzi)^\T(\Yi\!-\!H_i\Bi\!-\!\Bzi) + W_0),
		\text{IW}\left( \omega+n_i \,,\, \E_i^\T \E_i + \Sigma_0\right).
	\end{equation}
	
%	\color{red}
%	\vspace{-.05in}
%	\noindent For the univariate or independent specifications of the proposed model, \eqref{full_Si} may be replaced by an expression that requires lower computational cost, given by
%	
%	\vspace{-.2in}
%	\begin{equation}\label{full_Si_ind}
%		\sigma^{2}_{id} \v \bbt_{id}, \beta_{0id}, \byid
%		\sim
%		\text{IG}\left(
%		0.5\,n_i + \TR{a_\sigma}
%		\,,\,
%		%\bm{\epsilon}^\T_i\bm{\epsilon}_i
%		\textbf{e}^\T_i\textbf{e}_i
%		%	(\byid \!-\! H_i\bbt^{(d)}_i \!-\! \beta^{(d)}_{0i}{1}_{n_i})^\T
%		%	(\byid \!-\! H_i\bbt^{(d)}_i \!-\! \beta^{(d)}_{0i}{1}_{n_i})
%		+ \TR{b_\sigma}\right),
%	\end{equation}
%	
%	\vspace{-.05in}
%	\noindent where $\textbf{e}_i = \byid \!-\! H_i\bbt_{id} \!-\! \beta_{0id}{1}_{n_i}$.
%
%	\color{black}

	\vskip24pt
	\item For $d=1,\dots,D$, update $\mu_{0d}$ from the full conditional

	\begin{equation}\label{full_mu0}
		\mu_{0d} \mid -
		\;\sim\;
		\text{N}\left(\,
		\frac{s_0^2\sum_{i=1}^{m}\beta_{0id}}{m\,s_0^2 + \sigma^{2}_{0d}}
		\,,\,
		\frac{s_0^2 \, \sigma^{2}_{0d}}{m\,s_0^2 + \sigma^{2}_{0d}}
		\,\right).
	\end{equation}

	\vskip24pt
	\item For $d=1,\dots,D$, update $\sigma^{2}_{0d}$ from the full conditional
	
	\begin{equation}\label{full_sig02}
		\sigma^{2}_{0d} \bv - %\beta_{0id}, \mu_{0d}
		\;\sim\;
		\text{IG}\Big( 0.5\,m + a_0 \,,\, 0.5\sum_{i=1}^{m}(\beta_{0id} - \mu_{0d})^2 + b_0 \Big).
	\end{equation}

	\qed

\end{enumerate}

Some specific features regarding the proposed MCMC algorithm are discussed next. The full conditional distribution of the variance parameters $\lambda^{2}_{jd}$ is a Inverse-Gamma truncated in the interval $(0,A_d^2)$, as shown in \eqref{full_lam2}. To sample from $\lambda^{2}_{jd}$, we use the adaptive rejection sampling method (ARS) introduced by \cite{gilks1992}. Since \eqref{full_lam2} is not a log-concave density, we consider a reparameterization of the model by transforming $\lambda^{2}_{jd}$ to obtain a log-concave density. This allows us sample from the posterior distribution of $\lambda^{2}_{jd}$ via the ARS method. The open-source implementations of the ARS did not work properly for small values of the truncation parameter $A_d$, which motivated us to develop our own implementation, as discussed in Section \ref{supp:ARS}.

The sample from the component-specific mean vectors $\bth_{jd}$ is done jointly for all dimensions $d=1,\dots,D$, that is, we sample from the full conditional distribution of $\Theta_j$, $j=1,\dots,J$, available in \eqref{full_Theta_dep_Int}. We use the Metropolis-Hastings method with independent Normal proposal distributions for each element $\bth_{jd}$, $d=1,\dots,D$, in $\Theta_j$, with mean equal to the current values of $\bth_{jd}$ and covariance matrix equal to $\varepsilon K^{-1}$ when component $j$ is occupied or $\varepsilon_0 K^{-1}$ when component $j$ is empty. The covariance matrix $K^{-1}$ is the same defined for the Normal prior distribution of the vectors $\Theta_d$ in Equation \eqref{prior_Theta_d} of the main text. Since an empirical search for appropriate values of $\varepsilon$ and $\varepsilon_0$ could be costly in the case of a large data dimension $D$, we adapt the values of $\varepsilon$ and $\varepsilon_0$ during the burn-in period in order to achieve acceptance rates between $5\%$ and $25\%$ for all the mixture components $j=1,\dots,J$. In practice, the acceptance rates are within the desired range of $5\%$ to $25\%$ for the occupied components, but they are usually between $5\%$ and $50\%$ for the empty components. %It is necessary because each data dimension $d$ could require specific values for $\varepsilon_d$ and $\varepsilon_{0d}$ in the case of the independent version of the proposed model.

%In the case that no repulsion is considered, which is equivalent to assume $h(\Theta_d)=1$ for $d=1,\dots,D$, the component-specific vectors $\bth_{j1},\dots,\bth_{jD}$ in $\Theta_j$ are conditionally independent, and they are directly sampled from the full conditional distribution $\bth_{jd} \bv - \;\sim\; \N(V_j\,m_j \,, V_j ), \quad d=1,\dots,D,$ where $m_j$ and $V_j$ are the same defined in \eqref{full_theta_Int_d}.

In the case that no repulsion is considered, which is equivalent to assume $h(\Theta_d)=1$ for $d=1,\dots,D$, the component-specific vectors $\bth_{j1},\dots,\bth_{jD}$ in $\Theta_j$ are conditionally independent, and they are directly sampled from the following full conditional distribution, for each $j=1,\dots,J$:

%\vspace{-.1in}
\begin{equation}\label{full_theta_Int_d}
	\begin{aligned}
		\bth_{jd}\bv - 
		&\;\sim\;
		\N\left(\bth_{jd} \,; V_j\,m_j \,, V_j \right)\;,
		\quad d=1,\dots,D,
		\\%[.05in]
		&\hspace{-1.2in}\text{with}
		\\[-.1in]
		V_j^{-1}\!
		&\;=\;
		\tau^{-2}_{jd} K \,+\, n_j \lambda^{-2}_{jd}I_p \,-\, \lambda_{jd}^{-2}\left(\sum_{i:z_i=j}V_i\right)\lambda_{jd}^{-2}\,,
		\\[.05in]
		m_j
		&\;=\;
		\tau_{jd}^{-2}K\bmu_d \,+\, \lambda_{jd}^{-2}\sum_{i:z_i=j}\! V_i H_i^\T \sigma_{id}^{-2}(\byid-\beta_{0id}\bm{1}_{n_i})\,,
		\\[.05in]
		V_i^{-1}\!
		&\;=\;
		H_i^\T H_i\sigma_{id}^{-2} + \lambda_{jd}^{-2}I_p\,.
	\end{aligned}
\end{equation}

\vspace{.2in}
The full conditional distributions of $\Theta_j$ and the allocation vector $\bm{z}$ considered in the proposed MCMC algorithm are both marginalized with respect to $\mathcal{B}_1,\dots,\mathcal{B}_m$. This marginalization provided a better mixing in the posterior chain, but the marginalized full conditionals are more computationally expensive. To reduce computational cost, we consider the non-marginalized full conditional distributions of $\bm{z}$ and $\Theta_j$, $j=1,\dots,J$, in the Gibbs sampling scan of the split-merge step. The non-marginalized full conditional distributions of $z_i$, $i=1,\dots,m$, and $\Theta_j$, $j=1,\dots,J$, are given by

\vspace{.05in}
\begin{equation}\label{full_zi}
	\Pr(z_i=j \v -)
	\,=\,
	\frac
	{\pi_j(\bm{x}_i\v\ba_j)\prod_{d=1}^{D}\N(\bbt_{id}\,; \bth_{jd},\lambda^{2}_{jd}I_p)}
	{\sum_{\ell=1}^{J}\pi_\ell(\bm{x}_i\v\ba_\ell)\prod_{d=1}^{D}\N(\bbt_{id}\,; \bth_{\ell d},\lambda^{2}_{\ell d}I_p)}\,,
	\quad j=1,\dots,J,
\end{equation}

\vspace{.1in}
\noindent and
\vspace{-.2in}

\begin{equation}\label{full_theta_dep}
	\begin{aligned}
		\bth_{jd} \bv -
		&\;\sim\;
		\N(\bth_{jd} \,; V_j\,m_j \,, V_j )\,,
		\quad d=1,\dots,D,
		%		\\[.1in]
		%		V_j^{-1}
		%		&\;=\;
		%		\tau^{-2}_{jd} K \;+\; \lambda^{-2}_{jd} n_j I_p
		%		\\[.1in]
		%		m_j
		%		&\;=\;
		%		\tau^{-2}_{jd} K\bmu_d \;+\; \lambda^{-2}_{jd}\sum_{i:z_i=j}\!\bbt_{id}
	\end{aligned}
\end{equation}

\vspace{.1in}
\noindent respectively, with \;$V_j^{-1}\!=\!\tau^{-2}_{jd}K + \lambda^{-2}_{jd}n_j I_p$ and \;$m_j\!=\!\tau^{-2}_{jd} K\bmu_d \,+\, \lambda^{-2}_{jd}\sum_{i:z_i=j}\!\bbt_{id}$.  If no covariates are considered to model the covariate-dependent mixing weights $\pi_j(\bm{x}_i\v\ba_j)$, they are replaced by $\pi_j$. The split-merge step is very important for {an efficient MCMC algorithm. This step speeds up convergence (i.e. less MCMC iterations), and it proved to be an essential step to achieve a good mixing of the MCMC chain, especially in the case of larger data dimensions $D$. The split-merge procedure is discussed in next section.} An \texttt{R} package containing an implmentation of the proposed MCMC algorithm is available at \url{https://github.com/rcpedroso/MHRMMx}. The code is fully written in \texttt{c++}, which significantly speeds up execution.

%\newpage
\subsection{Split-merge algorithm}\label{supp:SM}

In this section, we discuss the split-merge step of the posterior simulation algorithm for the proposed \mbox{MFRMMx} model. This step adapts the method introduced by \cite{jain2007} to our finite mixture model for multivariate functional curves, with a repulsive dependence between the mixture components. Denote by $\gamma=\{\bz,\bTh,\bLam,\bTau\}$ the current state of the parameters in $\bz$, $\bTh$, $\bLam$ and $\bTau$ in the MCMC chain and let $\Arrowvert J\Arrowvert=\{1,\dots,J\}$. The split-merge step in the MCMC algorithm to sample from the MFRMMx model is executed as follows:

%\vspace{.1in}
\begin{enumerate}[leftmargin=*]
	
	\item Randomly select two distinct individual indexes $i_0$ and $i_1$\footnote{If $i_0=i_1$ and all the $J$ mixture components are occupied, the split-merge step is skipped in the current MCMC iteration, because the split proposal requires an empty component.}.
	
	\vspace{.1in}
	\item Define the set $S = \{i: z_i=z_{i_0} \text{ or } z_i=z_{i_1}\}\backslash\{i_0,i_1\}$ \;and\; define $\znew=\min\{j:n_j=0\}$.
	
	\vspace{.1in}
	\item\label{Launch} Define the launch states $\gamma^{L_s}\!=\{\bTh^{L_s},\bTau^{L_s},\bLam^{L_s},\bz^{L_s}\}$ and
	$\gamma^{L_m}\!=\{\bTh^{L_m},\bTau^{L_m},\bLam^{L_m},\bz^{L_m}\}$. $L_s$ and $L_m$ are equivalent here to the upper indexes $L_{split}$ and $L_{merge}$ in \cite{jain2007}, respectively.

	\begin{enumerate}[label=(\alph*),leftmargin=*]
		
		\vspace{.1in}
		\item Obtain $\gamma^{L_s}$ as follows:
		
		\begin{itemize}[leftmargin=*]
			
			\vspace{.05in}
			\item Do $\gamma^{L_s}=\gamma$ (that is, do $\bTh^{L_s}\!=\bTh$, \,$\bTau^{L_s}\!=\bTau$, \,$\bLam^{L_s}\!=\bLam$ and \,$\bz^{L_s}\!=\bz$).
			
			\vspace{.05in}
			\item If $\ziz=\zio$ set $\zizLs=\znew$.
			
			\vspace{.05in}
			\item Modify $\gamma^{L_s}$ by executing $n_{GS}$ iterations of the following \textbf{intermediate} restricted Gibbs sampling scan, always conditional to the newly updated parameter values in $\gamma^{L_s}$ and to the current values of the other related model parameters that are not directly involved in the split-merge step:
			
			\begin{enumerate}[label=(\roman*)]
				
				\vspace{.1in}
				\item update $\bth^{\,L_s}_{\ell d}$ from the Normal density defined in \eqref{full_theta_dep}, for $\ell=\zizLs,\,\zioLs$ and ${d=1,\dots,D}$;
				
				\vspace{.1in}
				\item update $\tau^{2\,L_s}_{\ell d}$ from the Inverse-Gamma density defined in \eqref{full_tau2}, for $\ell=\zizLs,\,\zioLs$ and ${d=1,\dots,D}$;
				
				\vspace{.1in}
				\item update $\lambda^{2\,L_s}_{\ell d}$ from the truncated Inverse-Gamma density defined in \eqref{full_lam2}, for $\ell=\zizLs,\,\zioLs$ and $d=1,\dots,D$;
				
				\vspace{.1in}
				\item update $\ziLs$ from

				\vspace{-.2in}
				\begin{equation*}%\label{full_JN_zi_Ls}
					\Pr\left( \ziLs=\ell \bv \Bi,\Theta^{L_s},\bLam^{L_s},\pi_\ell \right)
					\,=\,
					\frac{\pi_\ell \displaystyle\prod_{d=1}^{D}
						\N(\bbt_{id}\,;\, \bth_{\ell d}^{\,L_s} ,\,\lambda^{2\,L_s}_{\ell d}I_p)}
					{\displaystyle\sum_{j\,\in{\left\{\zizLs,\,\zioLs\right\}}}\!\!\!\!
							\pi_j \prod_{d=1}^{D}
						\N(\bbt_{id}\,;\, \bth_{jd}^{\,L_s} ,\,\lambda^{2\,L_s}_{jd}I_p)},
					\end{equation*}

				\noindent for $i\in S$, with $\ell\in\left\{\zizLs,\,\zioLs\right\}$.

			\end{enumerate}

		\end{itemize}
		
		\vspace{.2in}
		\item Obtain $\gamma^{L_m}$ as follows:
		
			\begin{itemize}[leftmargin=*]
				
			\vspace{.05in}
			\item Do $\gamma^{L_m}=\gamma$ (that is, do $\bTh^{L_m}\!=\bTh$, \,$\bTau^{L_m}\!=\bTau$, \,$\bLam^{L_m}\!=\bLam$ and \,$\bz^{L_m}\!=\bz$).
			
			\vspace{.05in}
			\item If $\ziz\ne\zio$ set $\zizLm=\zio$.
			
			\vspace{.05in}
			\item For $i\in S$, set $\ziLm=\zio$.
			
			\vspace{.05in}
			\item Modify $\gamma^{L_m}$ by executing $n_{GS}$ iterations of the following \textbf{intermediate} restricted Gibbs sampling scan, always conditional to the newly updated parameter values in $\gamma^{L_m}$ and to the current values of the other related model parameters that are not directly involved in the split-merge step:

			\begin{enumerate}[label=(\roman*)]
				
				\vspace{.1in}
				\item update $\bth^{\,L_m}_{\ell d}$ from the Normal density defined in \eqref{full_theta_dep}, for $\ell=\zioLm$ and ${d=1,\dots,D}$;
				
				\vspace{.1in}
				\item update $\tau^{2\,L_m}_{\ell d}$ from the Inverse-Gamma density defined in \eqref{full_tau2}, for $\ell=\zioLm$ and ${d=1,\dots,D}$;
				
				\vspace{.1in}
				\item update $\lambda^{2\,L_m}_{\ell d}$ from the truncated Inverse-Gamma density defined in \eqref{full_lam2}, for $\ell=\zioLm$ and ${d=1,\dots,D}$.
				
			\end{enumerate}
			
		\end{itemize}
	
	\end{enumerate}

	\vspace{.1in}
	\item If $\ziz=\zio$, propose to split the mixture component $\ziz$ in two different mixture components, as follows:
	
	\begin{enumerate}[label=(\alph*),leftmargin=*]
		
		\item\label{GS_Ls_final} Define the candidate state $\gamma^s=(\bTh^s,\bTau^s,\bLam^s,\bz^s)$ and obtain $\gamma^s$ as follows:
		
		\begin{itemize}[leftmargin=*]
			
			\vspace{.05in}
			\item Do $\gamma^s=\gamma^{L_s}$ (that is, do \,$\bTh^s\!=\bTh^{L_s}$, \,$\bTau^s\!=\bTau^{L_s}$, \,$\bLam^s\!=\bLam^{L_s}$ and \,$\bz^s\!=\bz^{L_s}$), where $\gamma^{L_s}$ was obtained in the last iteration of the \textbf{intermediate} restricted Gibbs sampling scan executed in \ref{Launch}. %\TRB{To make it easier to understand the next steps of the algorithm, it may be useful to note that $\zizs=\zizLs=\znew$ and $\zios=\zioLs=\zio$.}
			
			\vspace{.05in}
			\item Modify $\gamma^s$ by running one iteration of the following \textbf{final} restricted Gibbs sampling scan, always conditional to the newly updated parameter values in $\gamma^s$ and to the current values of the other related model parameters that are not directly involved in the split-merge step:
			
			\begin{enumerate}[label=(\roman*)]
				
				\vspace{.1in}
				\item update $\bth^{\,s}_{\ell d}$ from the Normal density defined in \eqref{full_theta_dep}, for $\ell=\zizs,\,\zios$ and ${d=1,\dots,D}$;
				
				\vspace{.1in}
				\item update $\tau^{2\,s}_{\ell d}$ from the Inverse-Gamma density defined in \eqref{full_tau2}, for $\ell=\zizs,\,\zios$ and ${d=1,\dots,D}$;
				
				\vspace{.1in}
				\item update $\lambda^{2\,s}_{\ell d}$ from the truncated Inverse-Gamma density defined in \eqref{full_lam2}, for $\ell=\zizs,\,\zios$ and $d=1,\dots,D$;

				\vspace{.1in}
				\item update $z_i^s$ from
				
				\vspace{-.1in}
				\begin{equation*}%\label{full_JN_zi_s}
					\Pr\left( z_i^s=\ell \bv \Bi,\Theta^{s},\bLam^{s},\pi_\ell \right)
					\,=\,
					\frac{\pi_\ell \displaystyle\prod_{d=1}^{D}
						\N(\bbt_{id}\,;\, \bth_{\ell d}^{\,s} \,,\,\lambda^{2\,s}_{\ell d}I_p)}
					{\displaystyle\sum_{j\,\in{\left\{\zizs,\,\zios\right\}}}\!\!\!\!
						\pi_j \prod_{d=1}^{D}
						\N(\bbt_{id}\,;\, \bth_{jd}^{\,s} \,,\,\lambda^{2\,s}_{jd}I_p)},
				\end{equation*}
				
				\noindent for $i\in S$, with $\ell\in\left\{\zizs,\,\zios\right\}$.
				
			\end{enumerate}
		
		\end{itemize}
		
		\vspace{.2in}
		\item Compute the proposal densities $q(\gamma^s\v\gamma)$ and $q(\gamma\v\gamma^s)$ to be used to calculate the acceptance probability of the proposed state $\gamma^s$:
		
		\begin{itemize}[leftmargin=*]
			
			\vspace{.1in}
			\item $q(\gamma^s\v\gamma)$ is the product of the densities used to update $\gamma^s$ from $\gamma^{L_s}$ in the course of the \textbf{final} restricted Gibbs sampling scan described in \ref{GS_Ls_final}. That is,
			
			\begin{equation}\label{qs_}
				\begin{aligned}
					q(\gamma^s\v\gamma)
					&\,=\,\prod_{d=1}^{D} \; \prod_{\ell\in\left\{\zizs,\zios\right\}}
					\!\!f(\bth^{s}_{\ell d} \v -)
					\;f(\tau^{2\,s}_{\ell d} \v -)
					\;f(\lambda^{2\,s}_{\ell d} \v -)
					%\\[.05in]
					%&\,\times\,
					\;\times\;
					%\prod_{i\in S}\Pr(z_i^s \bv \Bi,\Theta^{s},\bTau^{s},\bLam^{s},\pi_\ell)
					\prod_{i\in S}\Pr(z_i^s \v -)
				\end{aligned}
			\end{equation}

			\noindent where $f(\bth^{s}_{\ell d}\v-)$, $f(\tau^{2\,s}_{\ell d}\v-)$ and $f(\lambda^{2\,s}_{\ell d}\v-)$ are the sampling densities and $\Pr(z_i^s\v-)$ is the sampling probability considered in the \textbf{final} restricted Gibbs sampling scan described in \ref{GS_Ls_final}, evaluated at the respective sampled values $\bth^{s}_{\ell d}$\,, $\tau^{2\,s}_{\ell d}$\,, $\lambda^{2\,s}_{\ell d}$ and $z_i^s$. %These full conditional densities are also conditioned to the current values of other model parameters that are not elements of the candidate state $\gamma^s$, which means that they are not being updated during the split-merge procedure.

			\vspace{.2in}
			\item $q(\gamma\v\gamma^s)$ is the product of the densities that would be used to update the elements in $\gamma^{L_m}$ to their original values $\gamma$ in a hypothetical restricted Gibbs sampling scan. That is,
			
			\vspace{-.1in}
			\begin{equation}\label{q_s}
			\begin{aligned}
				q(\gamma\v\gamma^s)
				&\,=\,\prod_{d=1}^{D}\,
				f(\bth_{\ell d} \v -)
				\,f(\tau^{2}_{\ell d}\v-)
				\,f(\lambda^{2}_{\ell d} \v -)\\
				&\,\times\,\prod_{d=1}^{D}\,
				f(\bth_{\znew\,d}\v\tau^{2}_{\znew\,d})
				\,f(\tau^{2}_{\znew\,d})\,
				\,f(\lambda^{2}_{\znew\,d}),
			\end{aligned}
			\end{equation}

			\noindent with $\ell=\zioLm$, where $f(\bth_{\ell d}\v-)$, $f(\tau^{2}_{\ell d}\v-)$ and $f(\lambda^{2}_{\ell d}\v-)$ are the full conditional densities used to modify $\gamma^{L_m}$ {in step \ref{Launch} above, conditional on} the current values in $\gamma^{L_m}$ (obtained in the last iteration of the \textbf{intermediate} Gibbs sampling scan) and evaluated at the respective original values of $\bth_{\ell d}$\,, $\tau^{2}_{\ell d}$ and $\lambda^{2}_{\ell d}$ in $\gamma$, for $\ell=\zioLm$ and $d=1,\dots,D$. The factor associated to the hypothetical update of $\bz^{L_m}$ to its original value $\bz$ is equal to one because there is an unique way to combine two components into one component. Then, that factor does not appear in \eqref{q_s}. The densities $f(\bth_{\znew\,d}\v\tau^{2}_{\znew\,d})$, $f(\tau^{2}_{\znew\,d})$ and $f(\lambda^{2}_{\znew\,d})$ are the Normal, Inverse-Gamma and truncated Uniform densities defined in \eqref{prior_lam2} and \eqref{prior_Theta_d}, for each respective parameter. These densities are related to a hypothetical update from the candidate parameters of the mixture component $\znew$ to the original empty component.

		\end{itemize}
	
		\vspace{.2in}
		\item Evaluate the $\gamma^s$ proposal by the acceptance probability
		
		\vspace{-.1in}
		\begin{equation}\label{accept_split}
		a(\gamma^s,\gamma)
		\,=\,
		\left\{1\;,\;
		\frac{q(\gamma\v\gamma^s)}{q(\gamma^s\v\gamma)}\,
		\frac{P(\gamma^s)}{P(\gamma)}\,
		\frac{\prod_{i=1}^{m}\prod_{d=1}^{D}
		\N\left(\bbt_{id}\,; \bth^{\,s}_{z_i^{s} d} \,,\, \lambda^{2\,s}_{z_i^{s} d}I_p\right)}
		{\prod_{i=1}^{m}\prod_{d=1}^{D}
		\N\left(\bbt_{id}\,; \bth_{z_i d} \,,\, \lambda^{2}_{z_i d}I_p\right)}
		\,\right\},
		\end{equation}

		\vspace{.1in}
		\noindent where $q(\gamma^s\v\gamma)$ and $q(\gamma\v\gamma^s)$ are defined in \eqref{qs_} and \eqref{q_s}, and $P(\gamma^s)$ and $P(\gamma)$ denote the joint prior distribution of the elements in $\gamma^s$ and $\gamma$, respectively, with

		\begin{equation}\label{prior_split}
			P(\gamma)\,=\,
			\left(\prod_{i=1}^{m} \pi_{z_i}\right)
			\left[\prod_{d=1}^{D}
			\mathrm{C}^{-1}_{Jd}
			\left(\prod_{j=1}^{J}
			\N(\bth_{jd} \,;\, \bmu_d,\tau^2_{jd} K^{-1})
			\;\IG(\tau^2_{jd} \,;\, a_\tau,b_\tau)
			\;f(\lambda^{2}_{jd})
			\right)
			h(\Theta_d)
			\right]
		\end{equation}

		\vspace{-.1in}
		\noindent and
		\vspace{-.1in}
		
		\begin{equation}\label{prior_split_s}
			P(\gamma^s)\,=\,
			\left(\prod_{i=1}^{m} \pi_{z^s_i}\right)
			\left[\prod_{d=1}^{D}
			\mathrm{C}^{-1}_{Jd}
			\left(\prod_{j=1}^{J}
			\N(\bth^s_{jd} \,;\, \bmu_d,\tau^{2\,s}_{jd}\,K^{-1})
			\;\IG(\tau^{2\,s}_{jd} \,;\, a_\tau,b_\tau)
			\;f(\lambda^{2\,s}_{jd})
			\right)
			h(\Theta^s_d)
			\right],
		\end{equation}

		\vspace{.1in}
		\noindent where $f(\lambda^{2}_{jd})$ is defined in \eqref{eq_posterior}. Then, we have that
	
		\vspace{-.1in}
		\begin{equation}\label{prior_split_rate}
		\begin{aligned}
		\frac{P(\gamma^s)}{P(\gamma)}&\;=\;
		\frac
		{\displaystyle\prod_{i\in S\cup\{i_0,i_1\}} \pi_{z^s_i}}
		{\displaystyle\prod_{i\in S\cup\{i_0,i_1\}} \pi_{z_i}}
		\\[.1in]
		&\;\times\;\prod_{d=1}^{D}\;
		\frac{\displaystyle\prod_{j\in\{\zizs,\zios\}}
		\N(\bth^s_{jd}\,;\,\bmu_d,\tau^{2\,s}_{jd}\,K^{-1})
		\;\IG(\tau^{2\,s}_{jd}\,;\,a_\tau,b_\tau)
		\;f(\lambda^{2\,s}_{jd})
		}
		{\displaystyle\prod_{j\in\{\zizs,\zios\}}
		\N(\bth_{jd}\,;\,\bmu_d,\tau^{2}_{jd}\,K^{-1})
		\;\IG(\tau^{2}_{jd}\,;\,a_\tau,b_\tau)
		\;f(\lambda^{2}_{jd})
		}\;
		\\[.1in]
		&\;\times\;\prod_{d=1}^{D}\;
		\frac{g(\dist(\bth^{s}_{\zizs d} \,, \bth^{s}_{\zios d}))}
			{g(\dist(\bth_{\zizs d} \,, \bth_{\zios d}))}\,
		\frac{\displaystyle\prod_{j\,\in\,\Arrowvert J\Arrowvert\backslash\{\zizs,\zios\}}
				g(\dist(\bth^{s}_{jd} \,, \bth^{s}_{\zizs d}))\,
				g(\dist(\bth^{s}_{jd} \,, \bth^{s}_{\zios d}))}
			{\displaystyle\prod_{j\,\in\,\Arrowvert J\Arrowvert\backslash\{\zizs,\zios\}}
				g(\dist(\bth_{jd} \,, \bth_{\zizs d}))\,
				g(\dist(\bth_{jd} \,, \bth_{\zios d}))}.
		\end{aligned}
		\end{equation}

		%\vspace{.1in}
		\noindent {In \eqref{prior_split_rate}, all density components and $g$ functions in \eqref{prior_split} and \eqref{prior_split_s} that are not related to mixture components participating in the computation of the split proposal cancel.} Note that the factors $f(\bth_{\znew\,d}\v\tau^{2}_{\znew\,d})\,f(\tau^{2}_{\znew\,d})\,\,f(\lambda^{2}_{\znew\,d})$ in \eqref{q_s} and $\N(\bth_{\zizs d}\,;\,\bmu_d,\tau^{2}_{\zizs d}\,K^{-1}) \;\IG(\tau^{2}_{\zizs d}\,;\,a_\tau,b_\tau) \;f(\lambda^{2}_{\zizs d})$ in the denominator of the second line of \eqref{prior_split_rate} are equal, given that $\zizs=\zizLs=\znew$, so that they cancel out in the acceptance rate in \eqref{accept_split}. However, the current values of the parameters of the mixture component $\zizs$ in $\gamma$, that is, $\bth_{\zizs d}$\,, $\tau^{2}_{\zizs d}$ and $\lambda^{2}_{\zizs d}$\,, for $d=1,\dots,D$, are still part of the acceptance rate in \eqref{accept_split}, being used in the calculation of the ratio of repulsive factors in the last line of \eqref{prior_split_rate}. The likelihood ratio in \eqref{accept_split} is given by

%		\begin{equation}\label{likelihood_split_rate}
%		\begin{aligned}
%		&\frac{\displaystyle\prod_{i=1}^{m}\,\prod_{d=1}^{D}
%		\N\left(\bbt_{id}\,; \bth^{\,s}_{z_i^{s} d} \,,\, \lambda^{2\,s}_{z_i^{s} d}I_p\right)}
%		{\displaystyle\prod_{i=1}^{m}\,\prod_{d=1}^{D}
%		\N\left(\bbt_{id}\,; \bth_{z_i d} \,,\, \lambda^{2}_{z_i d}I_p\right)}\\
%		&\;=\;
%		%
%		\frac{
%			\displaystyle\prod_{i:z_i^s=\zizs}\prod_{d=1}^{D}
%			%f(\vect(\Bi)\bv\Theta^{s}_{z_i^s},\blam^{s}_{z_i^s})
%			\N\left(\bbt_{id}\,; \bth^{\,s}_{z_i^{s} d} \,,\, \lambda^{2\,s}_{z_i^{s} d}I_p\right)
%			\displaystyle\prod_{i:z_i^s=\zios}\prod_{d=1}^{D}
%			%f(\vect(\Bi)\bv\Theta^{s}_{z_i^s},\blam^{s}_{z_i^s})
%			\N\left(\bbt_{id}\,; \bth^{\,s}_{z_i^{s} d} \,,\, \lambda^{2\,s}_{z_i^{s} d}I_p\right)
%		}
%		{
%			\displaystyle\prod_{i\in S\cup\{i_0,i_1\}}\prod_{d=1}^{D}
%			%f(\vect(\Bi)\bv\Theta_{z_i},\blam_{z_i})
%			\N\left(\bbt_{id}\,; \bth_{z_i d} \,,\, \lambda^{2}_{z_i d}I_p\right),
%		}
%		\end{aligned}
%		\end{equation}
		
		\begin{equation}\label{likelihood_split_rate}
		\hspace{-1.5cm}\begin{aligned}
			\frac{\displaystyle\prod_{i=1}^{m}\,\prod_{d=1}^{D}
				\N\left(\bbt_{id}\,; \bth^{\,s}_{z_i^{s} d} \,,\, \lambda^{2\,s}_{z_i^{s} d}I_p\right)}
			{\displaystyle\prod_{i=1}^{m}\,\prod_{d=1}^{D}
				\N\left(\bbt_{id}\,; \bth_{z_i d} \,,\, \lambda^{2}_{z_i d}I_p\right)}
			\;=\;
			\frac{
				\displaystyle\prod_{i:z_i^s=\zizs}\prod_{d=1}^{D}
				%f(\vect(\Bi)\bv\Theta^{s}_{z_i^s},\blam^{s}_{z_i^s})
				\N\left(\bbt_{id}\,; \bth^{\,s}_{z_i^{s} d} \,,\, \lambda^{2\,s}_{z_i^{s} d}I_p\right)
				\displaystyle\prod_{i:z_i^s=\zios}\prod_{d=1}^{D}
				%f(\vect(\Bi)\bv\Theta^{s}_{z_i^s},\blam^{s}_{z_i^s})
				\N\left(\bbt_{id}\,; \bth^{\,s}_{z_i^{s} d} \,,\, \lambda^{2\,s}_{z_i^{s} d}I_p\right)
			}
			{
				\displaystyle\prod_{i\in S\cup\{i_0,i_1\}}\prod_{d=1}^{D}
				%f(\vect(\Bi)\bv\Theta_{z_i},\blam_{z_i})
				\N\left(\bbt_{id}\,; \bth_{z_i d} \,,\, \lambda^{2}_{z_i d}I_p\right)
			},
		\end{aligned}
		\end{equation}

		\noindent where the likelihood factors related to the individuals $i$ such that $z_i\ne\zizs$ and $z_i\ne\zios$ cancel.
		
	\end{enumerate}

	\vspace{.2in}
	\item If $\ziz\ne\zio$, propose to merge the mixture components $\ziz$ and $\zio$ as follows:
	
	\begin{enumerate}[label=(\alph*),leftmargin=*]
		
		\item\label{GS_Lm_final} Define the candidate state $\gamma^m=(\bTh^m,\bTau^m,\bLam^m,\bz^m)$ and obtain the $\gamma^m$ element values by:
		
		\begin{itemize}[leftmargin=*]
			
			\vspace{.05in}
			\item Do $\gamma^m=\gamma^{L_m}$ (that is, do \,$\bTh^m\!=\bTh^{L_m}$, \,$\bTau^m\!=\bTau^{L_m}$, \,$\bLam^m\!=\bLam^{L_m}$ and \,$\bz^m\!=\bz^{L_m}$), where $\gamma^{L_m}$ was obtained in the last iteration of the \textbf{intermediate} restricted Gibbs sampling scan executed in \ref{Launch}.%\TRB{To make it easier to understand the next steps of the algorithm, it may be useful to note that $\zizm=\zizLm=\zio$ and $\ziom=\zioLm=\zio$.}
			
			\vspace{.05in}
			\item Modify $\gamma^m$ by executing one iteration of the following \textbf{final} restricted Gibbs sampling scan, always conditional to the newly updated parameter values in $\gamma^m$ and to the current values of the other related model parameters that are not directly involved in the split-merge step:

			\begin{enumerate}[label=(\roman*)]
				
				\vspace{.1in}
				\item update $\bth^{\,m}_{\ell d}$ from the Normal density defined in \eqref{full_theta_dep}, for $\ell=\ziom$ and ${d=1,\dots,D}$;
				
				\vspace{.1in}
				\item update $\tau^{2\,m}_{\ell d}$ from the Inverse-Gamma density defined in \eqref{full_tau2}, for $\ell=\ziom$ and ${d=1,\dots,D}$;
				
				\vspace{.1in}
				\item update $\lambda^{2\,m}_{\ell d}$ from the truncated Inverse-Gamma density defined in \eqref{full_lam2}, for $\ell=\ziom$ and ${d=1,\dots,D}$.
				
			\end{enumerate}

			\vspace{.05in}
			\item Modify $\gamma^m$ by updating $\bth^{\,m}_{\ell d}$, $\tau^{2\,m}_{\ell d}$ and $\lambda^{2\,m}_{\ell d}$, for $\ell=\zizm$ and ${d=1,\dots,D}$, from their prior densities, since the mixture component $\zizm$ will be proposed empty:
			
			\begin{enumerate}[label=(\roman*)]
				
				\vspace{.1in}
				\item update $\tau^{2\,m}_{\ell d}$ from the Inverse-Gamma density defined in \eqref{prior_Theta_d}, for $\ell=\zizm$ and ${d=1,\dots,D}$;
				
				\vspace{.1in}
				\item update $\bth^{\,m}_{\ell d}$ from the Normal density defined in \eqref{prior_Theta_d}, for $\ell=\zizm$ and ${d=1,\dots,D}$;
				
				\vspace{.1in}
				\item update $\lambda^{2\,m}_{\ell d}$ from the truncated Inverse-Gamma density defined in \eqref{prior_lam2}, for $\ell=\zizm$ and ${d=1,\dots,D}$.
				
			\end{enumerate}
			
			\vspace{.1in}
			The modification of $\gamma^m$ executed in this item does not happen in practice in the original algorithm in \cite{jain2007}, because the densities associated to $\bth^{\,m}_{\ell d}$, $\tau^{2\,m}_{\ell d}$ and $\lambda^{2\,m}_{\ell d}$, for $\ell=\zizm$ and ${d=1,\dots,D}$, will cancel in the acceptance rate. However, in the proposed model, the values of $\bth^{\,m}_{\ell d}$, for $\ell=\zizm$ and ${d=1,\dots,D}$, sampled here will be used to compute the ratio of repulsive factors in the acceptance rate \eqref{accept_merge}, while the sampled values of $\bth^{\,m}_{\ell d}$, $\tau^{2\,m}_{\ell d}$ and $\lambda^{2\,m}_{\ell d}$ will update the respective model parameters $\bth_{\ell d}$, $\tau^{2}_{\ell d}$ and $\lambda^{2}_{\ell d}$, for $\ell=\zizm$ and ${d=1,\dots,D}$, if the merge step is accepted.

		\end{itemize}

		\vspace{.2in}
		\item Compute the proposal densities $q(\gamma^m\v\gamma)$ and $q(\gamma\v\gamma^m)$ to be used to calculate the acceptance probability of the proposed state $\gamma^m$.
		
		\begin{itemize}[leftmargin=*]
			
			\vspace{.1in}
			\item $q(\gamma^m\v\gamma)$ is the product of the densities used to update $\gamma^m$ from $\gamma^{L_m}$ in the course of the \textbf{final} restricted Gibbs sampling scan described in \ref{GS_Lm_final}. That is,
			
			\vspace{-.1in}
			\begin{equation}\label{qm_}
				\begin{aligned}
					q(\gamma^m\v\gamma)
					&\,=\,\prod_{d=1}^{D}\,
					f(\bth^{\,m}_{\ell d} \v -)
					\;f(\tau^{2\,m}_{\ell d} \v -)
					\;f(\lambda^{2\,m}_{\ell d} \v -)\\
					&\,\times\,\prod_{d=1}^{D}\;
					f(\bth^{\,m}_{\ziz d} \v \tau^{2\,m}_{\ziz d})\,
					\;f(\tau^{2\,m}_{\ziz d})
					\;f(\lambda^{2\,m}_{\ziz d}),
				\end{aligned}
			\end{equation}
			
			\noindent with $\ell=\ziom$\,, where $f(\bth^{\,m}_{\ell d}\v-)$, $f(\tau^{2\,m}_{\ell d}\v-)$ and $f(\lambda^{2\,m}_{\ell d}\v-)$ are the sampling densities considered in the \textbf{final} restricted Gibbs sampling scan described in \ref{GS_Lm_final}, evaluated at the respective sampled values $\bth^{\,ms}_{\ell d}$\,, $\tau^{2\,m}_{\ell d}$ and $\lambda^{2\,m}_{\ell d}$. %These full conditional densities are also conditioned to the current values of other model parameters that are not elements of the candidate state $\gamma^m$, which means that they are not being updated during the split-merge procedure.
			The probability of the transition from $\bz^{L_m}$ to $\bz^m$ is equal to one because there is a unique way to move between the merged components $\bz^{L_m}$ and $\bz^m$. Then, there is no probability factor related to that transition displayed in \eqref{qm_}. The densities $f(\bth^{\,m}_{\ziz d}\v\tau^{2\,m}_{\ziz d})$, $f(\tau^{2\,m}_{\ziz d})$ and $f(\lambda^{2\,m}_{\ziz d})$ in the second line of \eqref{qm_} represent the Normal, Inverse-Gamma and truncated Uniform prior densities defined in \eqref{prior_lam2} and \eqref{prior_Theta_d}, for each respective parameter. These densities are related to the update from the original parameters of the mixture component $\ziz$ to the new parameters of an empty mixture component.

			\vspace{.2in}
			\item $q(\gamma\v\gamma^m)$ is the product of the densities that would be used to update the elements in $\gamma^{L_s}$ to their original values $\gamma$ in a hypothetical restricted Gibbs sampling scan. That is,

			\vspace{-.1in}
			\begin{equation}\label{q_m}
				\begin{aligned}
					q(\gamma\v\gamma^m)
					&\,=\,
					\prod_{d=1}^{D}\;\prod_{\ell\in\{\zizLs,\,\zioLs\}}\!\!\!
					f(\bth_{\ell d} \v -)
					\;f(\tau^{2}_{\ell d} \v -)
					\;f(\lambda^{2}_{\ell d} \v -)
					\;\times\;\prod_{i\in S}^{}\,
					%\Pr(z_i \v \bz^{L_s},\Bi,\Theta^{s},\blam^{s},\bm{\pi}).
					\Pr(z_i \v -).
				\end{aligned}
			\end{equation}

			\vspace{.1in}
			\noindent where $f(\bth_{\ell d}\v-)$, $f(\tau^{2}_{\ell d}\v-)$ and $f(\lambda^{2}_{\ell d}\v-)$ are the full conditional densities used to {modify $\gamma^{L_s}$ in \ref{Launch}, conditional on} the current values in $\gamma^{L_s}$ (obtained in the last iteration of the \textbf{intermediate} Gibbs sampling scan) and evaluated at the respective original values of $\bth_{\ell d}$\,, $\tau^{2}_{\ell d}$ and $\lambda^{2}_{\ell d}$ in $\gamma$, for $\ell\in\{\zizLs,\,\zioLs\}$ and $d=1,\dots,D$.
			
		\end{itemize}
		
		\vspace{.2in}
		\item Evaluate the $\gamma^m$ proposal by the acceptance probability
		
		\begin{equation}\label{accept_merge}
			a(\gamma^m,\gamma)
			\,=\,
			\left\{1\;,\;
			\frac{q(\gamma\v\gamma^m)}{q(\gamma^m\v\gamma)}\,
			\frac{P(\gamma^m)}{P(\gamma)}\,
			\frac{\prod_{i=1}^{m}\prod_{d=1}^{D}
			\N\left(\bbt_{id}\,; \bth^{\,m}_{z_i^{m} d} \,,\, \lambda^{2\,m}_{z_i^{m} d}\,I_p\right)}
			{\prod_{i=1}^{m}\prod_{d=1}^{D}
			\N\left(\bbt_{id}\,; \bth_{z_i d} \,,\, \lambda^{2}_{z_i d}\,I_p\right)}
			\,\right\},
		\end{equation}

		\vspace{.1in}
		\noindent where $q(\gamma^m\v\gamma)$ and $q(\gamma\v\gamma^m)$ are defined in \eqref{qm_} and \eqref{q_m}, $P(\gamma^m)$ and $P(\gamma)$ denote the joint prior distributions of the elements of $\gamma^m$ and $\gamma$, respectively, with $P(\gamma)$ defined in \eqref{prior_split} and

		\begin{equation}\label{prior_merge_m}
			P(\gamma^m)\,=\,
			\left(\prod_{i=1}^{m} \pi_{z^m_i}\right)
			\left[\prod_{d=1}^{D}
			\mathrm{C}^{-1}_{Jd}
			\left(\prod_{j=1}^{J}
			\N(\bth^{\,m}_{jd} \,;\, \bmu_d,\tau^{2\,m}_{jd}\,K^{-1})
			\;\IG(\tau^{2\,m}_{jd} \,;\, a_\tau,b_\tau)
			\;f(\lambda^{2\,m}_{jd})
			\right)
			h(\Theta^m_d)
			\right],
		\end{equation}

		\noindent where $f(\lambda^{2}_{jd})$ is defined in \eqref{eq_posterior}. Then, we have that

		\vspace{-.1in}
		\begin{equation}\label{prior_merge_rate}
			\begin{aligned}
				\frac{P(\gamma^m)}{P(\gamma)}
				&=\,\frac
				{\displaystyle\prod_{i\in S\cup\{i_0,i_1\}} \pi_{z^m_i}}
				{\displaystyle\prod_{i\in S\cup\{i_0,i_1\}} \pi_{z_i}}
				\\[.05in]
				&\hspace{-.5in}\times\; \prod_{d=1}^{D}\;
				\frac{
					\displaystyle\prod_{j\in\{\ziz,\zio\}}
					\N(\bth^{\,m}_{jd}\,;\,\bmu_d,\tau^{2\,m}_{jd}\,K^{-1})
					\;\IG(\tau^{2\,m}_{jd}\,;\,a_\tau,b_\tau)
					\;f(\lambda^{2\,m}_{jd})
				}
				{	
					\displaystyle\prod_{j\in\{\ziz,\zio\}}
					\N(\bth_{jd}\,;\,\bmu_d,\tau^{2}_{jd}\,K^{-1})
					\;\IG(\tau^{2}_{jd}\,;\,a_\tau,b_\tau)
					\;f(\lambda^{2}_{jd})
				}\\[.1in]
				&\hspace{-.5in}\times\; \prod_{d=1}^{D}\;
				\frac{g(\dist(\bth^{m}_{\ziz d}\,,\bth^{m}_{\zio d}))}
				{g(\dist(\bth_{\ziz d}\,,\bth_{\zio d}))}\,
				\frac{\displaystyle\prod_{j\,\in\,\Arrowvert J\Arrowvert\backslash\{\ziz,\zio\}}
					g(\dist(\bth^{m}_{jd}\,,\bth^{m}_{\ziz d}))\,
					g(\dist(\bth^{m}_{jd}\,, \bth^{m}_{\zio d}))}
				{\displaystyle\prod_{j\,\in\,\Arrowvert J\Arrowvert\backslash\{\ziz,\zio\}}
					g(\dist(\bth_{jd}\,,\bth_{\ziz d}))\,
					g(\dist(\bth_{jd}\,,\bth_{\zio d}))}
			\end{aligned}
		\end{equation}
		
		\vspace{.1in}
		\noindent In \eqref{prior_merge_rate}, all density components and $g$ functions that are not related to mixture components participating in the computation of the split proposal cancel. Note that the factors $f(\bth^{\,m}_{\ziz d} \v \tau^{2\,m}_{\ziz d})\;f(\tau^{2\,m}_{\ziz d})\;f(\lambda^{2\,m}_{\ziz d})$ in \eqref{qm_} and $\N(\bth^{\,m}_{\ziz d}\,;\,\bmu_d,\tau^{2\,m}_{\ziz d}\,K^{-1})$ \;$\IG(\tau^{2\,m}_{\ziz d}\,;\,a_\tau,b_\tau)\;f(\lambda^{2\,m}_{\ziz d})$ in the numerator of the second line of \eqref{prior_merge_rate} cancel out in the acceptance rate in \eqref{accept_merge}. However, the proposal sampled values $\bth^{\,m}_{\ziz d}$\,, $\tau^{2\,m}_{\ziz d}$\, and $\lambda^{2\,m}_{\ziz d}$\,, for $d=1,\dots,D$, are still part of the acceptance rate in \eqref{accept_merge}, being used in the calculation of the ratio of repulsive factors in the last line of \eqref{prior_merge_rate}. Finally, the likelihood ratio in \eqref{accept_merge} is given by
		
		%\vspace{.2in}

%		\begin{equation}\label{likelihood_merge_rate}
%		\begin{aligned}
%		&\frac{\displaystyle\prod_{i=1}^{m}\,\prod_{d=1}^{D}
%		\N\left(\bbt_{id}\,;\bth^{\,m}_{z_i^{m} d}\,,\,\lambda^{2\,m}_{z_i^{m} d}I_p\right)}
%		{\displaystyle\prod_{i=1}^{m}\,\prod_{d=1}^{D}
%		\N\left(\bbt_{id}\,; \bth_{z_i d} \,,\, \lambda^{2}_{z_i d}I_p\right)}\\
%		&\;=\;
%		%
%		\frac{
%			\displaystyle\prod_{i\in S\cup\{i_0,i_1\}}\prod_{d=1}^{D}
%			\N\left(\bbt_{id}\,;\bth^{\,m}_{z_i^{m} d}\,,\,\lambda^{2\,m}_{z_i^{m} d}I_p\right)
%		}
%		{
%			\displaystyle\prod_{i:z_i=\ziz}\prod_{d=1}^{D}
%			\N\left(\bbt_{id}\,; \bth_{z_i d} \,,\, \lambda^{2}_{z_i d}I_p\right)
%			\displaystyle\prod_{i:z_i=\zio}\prod_{d=1}^{D}
%			\N\left(\bbt_{id}\,; \bth_{z_i d} \,,\, \lambda^{2}_{z_i d}I_p\right)
%		}
%		\end{aligned}
%		\end{equation}
				
		\begin{equation}\label{likelihood_merge_rate}
		\hspace{-1.5cm}\begin{aligned}
			\frac{\displaystyle\prod_{i=1}^{m}\,\prod_{d=1}^{D}
				\N\left(\bbt_{id}\,;\bth^{\,m}_{z_i^{m} d}\,,\,\lambda^{2\,m}_{z_i^{m} d}I_p\right)}
			{\displaystyle\prod_{i=1}^{m}\,\prod_{d=1}^{D}
				\N\left(\bbt_{id}\,; \bth_{z_i d} \,,\, \lambda^{2}_{z_i d}I_p\right)}
			\;=\;
			\frac{
				\displaystyle\prod_{i\in S\cup\{i_0,i_1\}}\prod_{d=1}^{D}
				\N\left(\bbt_{id}\,;\bth^{\,m}_{z_i^{m} d}\,,\,\lambda^{2\,m}_{z_i^{m} d}I_p\right)
			}
			{
				\displaystyle\prod_{i:z_i=\ziz}\prod_{d=1}^{D}
				\N\left(\bbt_{id}\,; \bth_{z_i d} \,,\, \lambda^{2}_{z_i d}I_p\right)
				\displaystyle\prod_{i:z_i=\zio}\prod_{d=1}^{D}
				\N\left(\bbt_{id}\,; \bth_{z_i d} \,,\, \lambda^{2}_{z_i d}I_p\right)
			},
		\end{aligned}
		\end{equation}

		\noindent where the likelihood factors related to the individuals $i$ such that $z_i\ne\ziz$ and $z_i\ne\zio$ cancel.
		
		\vspace{-.1in}
		\begin{flushright}$\qed$\end{flushright}
		
	\end{enumerate}

\end{enumerate}

Next, we will discuss some key aspects of the proposed split-merge algorithm we just described. If we are not considering repulsion in the model, which is equivalent to assume $h(\Theta_d)=1$, $d=1,\dots,D$, the last line of \eqref{prior_split_rate} and \eqref{prior_merge_rate} must be removed. For the MFPPMx model, that has been defined without a repulsive dependency assumption between the cluster-specific parameter vectors in $\Theta_d$, $d=1,\dots,D$, the last line of \eqref{prior_split_rate} and \eqref{prior_merge_rate} is always removed. Another difference between the split-merge step executed to the MFPPMx and MFPPMx models arises from the fact that the prior and, consequently, the full conditional distributions of the allocation vector $\bz$, both discussed in Section \ref{supp:mfppmx} of this supplementary material, depend on nonparametric cluster cohesions and nonparametric similarity functions, which have to be recomputed in each iteration of the Gibbs sampling scan. Besides, the prior and full conditional distributions of $\bz$ in the MFRMMx depend on the parametric mixing weights $\bm{\pi}$ and the split-merge step is executed conditional on the current mixing weights values, that are not updated during the Gibbs sampling scan. Finally, the third item in 5.\ref{GS_Lm_final}, that generates proposal parameters for the mixture componet $\zizm$, that will be proposed empty, does not need to be executed in the case of the MFPPMx model, because that proposal parameters are only used to compute the ratios of repulsive factors in the last line of \eqref{prior_split_rate} and \eqref{prior_merge_rate}, which does not apply to the MFPPMx model.

%\newpage
\subsection{Adaptive rejection sampling method}
\label{supp:ARS}

Although the some versions of adaptive rejection sampling (ARS) introduced by \cite{gilks1992}  are available (see, for instance, the \texttt{R} packages \texttt{ars} and \texttt{armspp}), when  integrated into the \texttt{C++} implementation of our proposed MCMC algorithm  to sample from the distribution \eqref{full_lam2} these methods presented  numerical problems,  mainly, when  lower values of the truncation parameter $A_d$ were considered. {This motivated us to produce} our own \texttt{C++} implementation of the ARS method, specifically to sample from the full conditional distribution \eqref{full_lam2}. This implementation is available at \url{https://github.com/rcpedroso/MHRMMx}. To this end, we consider the following  parameter transformation that satisfies the log-concave density requirement of the ARS method.

Let  $X$  be a random variable with a  Inverse-Gamma distribution truncated at the interval $(0,L)$, with shape parameter $a>0$ and scale parameter $b>0$ and pdf given by

\begin{equation}\label{f_ig_trunc}
%	f(x) \propto\; x^{-(a+1)} \exp\left(-\frac{b}{x}\right)
	f(x) \propto\; x^{-(a+1)}\, \me^{-b/x}\,
	\bm{1}(\{0 < x < L\}),
\end{equation}

\noindent where $\bm{1}(A)$ denotes the indicator function of event $A$.

\begin{proposition}\label{prop_ars}
If a random variable $X$ has an Inverse-Gamma distribution truncated at the interval $(0,L)$, with shape parameter $a>0$ and scale parameter $b>0$, then the density function of the random variable ${W=\log(\xi X+1)}$, with $\xi>0$, is log-concave if ${\xi>(a+1)/b}-2/L$.

\begin{proof}
Since the logarithm function is strictly increasing, the random variable $W$ will assume values in the interval $(0,\log(\xi L+1))$. Given that $X=\displaystyle(\me^W-1)/\xi$,  for $w\in(0,\log(\xi L+1))$, the pdf  of $W$ is

\begin{equation*}
	\begin{aligned}
		f(w)
		%
%		\propto\;&
%		(\exp(w)-1)^{-(a+1)}
%		\exp\left(-\frac{b\xi}{\exp(w)-1}\right)
%		\left|\frac{d}{dw}\frac{\exp(w)-1}{\xi}\right|\\[.2in]
%		%
%		\propto\;&
%		(\exp(w)-1)^{-(a+1)}
%		\exp\left(w-\frac{b\xi}{\exp(w)-1}\right)\\[.1in]
		%
		\propto\;&
		(\me^w-1)^{-(a+1)}
		\exp\left(-\frac{b\xi}{\me^w-1}\right)
		\left|\frac{d}{dw}\frac{\me^w-1}{\xi}\right|\\[.1in]
		\propto\;&
		(\me^w-1)^{-(a+1)}
		\exp\left(w-\frac{b\xi}{\me^w-1}\right),
	\end{aligned}
\end{equation*}

Consequently, the log-density of $W$ and its second derivative, respectively, are

\begin{equation*}
	\begin{aligned}
		\log f(w)\propto\;&
		-(a+1)\log(\me^w-1) + w - \frac{b\xi}{\me^w-1}\\[.1in]
	\end{aligned}
\end{equation*}

\vspace{-.3in}
\noindent and
\vspace{-.1in}

\begin{equation*}
	\begin{aligned}
		\frac{d^2}{dw^2}\log f(w)\propto\;&
		(a+1)\frac{\me^w}{(\me^w-1)^2}
		- b\xi\me^w\frac{\me^w+1}{(\me^w-1)^3}.\\[.1in]
	\end{aligned}
\end{equation*}

\vspace{-.2in}
\noindent Then,

\vspace{-.1in}
\begin{equation*}
	\begin{aligned}
		\frac{d^2}{dw^2}\log f(w) < 0
		\;\iff\; &(a+1) - b\xi\, \frac{\me^w+1}{\me^w-1} < 0\\[.1in]
		\;\iff\; &\frac{a+1}{b} < \xi\, \frac{\me^w+1}{\me^w-1}\\[.1in]
		\;\iff\; &\xi > \frac{a+1}{b} \cdot \frac{\me^w-1}{\me^w+1}.\\[.1in]
	\end{aligned}
\end{equation*}

\vspace{-.2in}
\noindent Let $h(w)=(\me^w-1)/(\me^w+1)$. Because $\displaystyle\frac{d}{dw}h(w) = 2\me^w/(\me^w+1)^2 > 0$, we have that $h(w)$ is a strictly increasing function. Then, for $w\in(0,\log(\xi L+1))$, $h(w)$ has a maximum value equal to $\xi L/(\xi L+2)$  for $w=\log(\xi L+1)$. Therefore, $f(w)$ is log-concave if 
\begin{equation*}
	\xi > \frac{a+1}{b} \cdot \frac{\xi L}{\xi L+2},
\end{equation*}
or, equivalently, if 
\begin{equation}\label{xi_limit}
	\xi > \frac{a+1}{b} - \frac{2}{L}.
\end{equation}

\end{proof}

\end{proposition}

\noindent In our ARS implementation, we assume $\xi=(a\!+\!1)/b$, which satisfies \eqref{xi_limit}.

%\section{Right-truncated distribution of $\sigma_i$}\label{app_ars_sig}
%
%
%\begin{proposition}
%If a random variable $X$ has a density function given by
%
%\begin{equation}\label{f_sig_trunc}
%	f(x) \propto\; x^{-n}\, \exp\left(-\frac{a}{2}x^{-2} - bx^{-1}\right)\,
%	\bm{1}(\{0 < x < L\}),\\[.1in]
%\end{equation}
%
%\noindent where $\bm{1}(A)$ denotes the indicator function of event $A$, $n\in\mathbb{N}$ and $a,b\in\mathbb{R}$, then $W=X^{-1}$ is a random variable with a log-concave density if $n\ge 2$ and $a>0$\footnote{Prove that $a$ in \eqref{f_sig_trunc} is always positive}.
%
%\begin{proof}
%	
%Given that $X=W^{-1}$, then $W$ has density function given by
%
%\begin{equation*}
%	\begin{aligned}
%		f(w)
%		%
%		\propto\;&
%		w^{n} \exp\left(-\frac{a}{2}w^{2} - bw\right)
%		\left|\frac{d}{dw}w^{-1}\right|\\[.2in]
%		%
%		\propto\;&
%		w^{n-2} \exp\left(-\frac{a}{2}w^{2} - bw\right)\\[.1in]
%	\end{aligned}
%\end{equation*}
%
%\noindent if $W\in(1/L,\infty)$ and 0 otherwise. It follows that the log-density of $W$ and its second derivative are, respectively,
%
%\begin{equation*}
%	\begin{aligned}
%		\log f(w)\propto\;&
%		(n-2)\log(w) -\frac{a}{2}w^{2} - bw\\[.1in]
%	\end{aligned}
%\end{equation*}
%
%\noindent and
%
%\vspace{-.1in}
%\begin{equation*}
%	\begin{aligned}
%		\frac{d^2}{dw^2}\log f(w)\propto\;&
%		-\frac{n-2}{w^2}-a.\\[.1in]
%	\end{aligned}
%\end{equation*}
%
%\noindent Then, $\displaystyle\frac{d^2}{dw^2}\log f(w) <0$ if $n\ge 2$ and $a>0$.
%
%	
%\end{proof}
%	
%\end{proposition}

%\newpage
\section{Multivariate functional PPM\lowercase{x} model (MFPPM\lowercase{x})}\label{supp:mfppmx}

In this section we extend the HPPMx model introduced by \cite{page2015} to the multivariate case. The MFPPMx model identifies clusters of{experimental units} represented by functional curves based on curve shape similarity and individual-specific covariates. It is important to point out that while  the HPPMx model was  defined to cluster possibly unfinished data sequences, the MFPPMx model is restricted to the case of completely observed data sequences. 
 
Unlike the MFRMMx model in \eqref{eq_model}, the MFPMMx model assumes a random partition $\rho$ of the set of observation indexes $\{1,\dots,m\}$ into a random number $J$ of  clusters $S_j$, such that $i\in S_j$ means that individual $i$ is allocated to cluster $S_j$.  Considering the same notation defined  in Section \ref{sec:model} for the MFRMMx model,  for all $i=1,\dots,m$, $j=1,\dots,J$ and $d=1,\dots,D$, the  MFPPMx model is hierarchically defined by

\begin{equation}\label{eq_mfppmx}
	%\centering
	%\left.
	\begin{aligned}
		\vect(\Yi) \v \Bzi,\Bi,\Si &
		\;\overset{\text{ind}}{\sim}\;
		\N_{n_i D}(\vect(\Bzi \!+\! H_i\Bi) \,,\, \Si\otimes I_{n_i}),
		\\%[.0in]
		\Si &
		\;\overset{\text{iid}}{\sim}\;
		\IW(\Sigma_0,\omega),
		\\[.0in]
		\beta_{0id} \v \mu_{0d},\sigma_{0d}^{2} &
		\;\overset{\text{ind}}{\sim}\;
		\N(\mu_{0d},\sigma_{0d}^{2})
		\quad\text{with}\;\;
		\mu_{0d}\!\sim \N(0,s^2_0)
		\;\;\text{and}\;\;
		\sigma_{0d}^{2} \!\sim \IG\left(a_0,b_0\right),
		%\sqrt{\sigma^2_i} \sim \U(0,A_\sigma)
		\\[.0in]
		\bbt_{id} \v z_i,\bth_{z_i d},\lambda^{2}_{z_i d} &
		\;\overset{\text{ind}}{\sim}\;
		\N_p(\bth_{z_i d},\lambda^{2}_{z_i d} I_p)
		\quad\text{with}\;\;
		\sqrt{\lambda^{2}_{jd}}\;\overset{\text{iid}}{\sim}\;\U(0,A_d),
		\\[.0in]
		%
		%			f(\bth_{jd} \v \bm{\mu}_d,\tau^2_{jd}) &
		%			%\;\propto\;
		%			\;\overset{\text{ind}}{\propto}\;
		%			\N_p\left(\bm{\mu}_d,\tau^2_{jd} K^{-1}\right)
		%			h\!\left(\Theta^{(d)}\right)
		%			\\[.05in]
		%			%
		\Theta_d\,,\btau_d \v \bm{\mu}_d &
		%\;\propto\;
		\;\overset{\text{ind}}{\propto}\;
		\prod_{j=1}^{J}
		\N_p(\bth_{jd} \,;\, \bm{\mu}_d,\tau^2_{jd} K^{-1})
		\;\IG(\tau^2_{jd} \,;\, a_\tau,b_\tau),
		\\%[.05in]
		\bm{\mu}_d &
		\;\overset{\text{iid}}{\sim}\;
		\N_p(\bm{0},s^2_\mu I_p),
		\\[.0in]
		%
		%			\tau^2_{jd} &
		%			\;\overset{\text{iid}}{\sim}\;
		%			\IG(a_\tau,b_\tau)
		%			\\[.05in]
		%			%
		\Pr(\rho=\{S_1,\dots,S_J\}) &\;\propto\; \prod_{j=1}^{J}c(S_j)g(\mathcal{X}_j),
		\\[.0in]
		%
	%	&\hspace{-1.5in}\text{for $i=1,\dots,m$, $j=1,\dots,J$ and $d=1,\dots,D$}.
		\\[-.3in]
	\end{aligned}
	%\right\}
\end{equation}

 \noindent for $i\!=\!1,\dots,m$, $j\!=\!1,\dots,J$ and $d\!=\!1,\dots,D$. The MFPPMx model assumes that $\rho$ has a PPMx prior distribution \citep{muller2011} which is a covariate-dependent product distribution determined by cohesion functions $c(S_j)\ge0$ and similarity functions $g(\mathcal{X}_j)\ge0$, where $\mathcal{X}_j$ represents the set of covariate values of the individuals allocated to cluster $j$. The cohesion functions $c(S_j)$, $j=1,\dots, J$, {indicates how likely the elements in $S_j$ are thought co-cluster} \textit{a priori}. The similarity functions $g(\mathcal{X}_j)$ make individuals with similar covariate values more likely to co-cluster, \textit{a priori}. A detailed discussion about similarity functions for continuous or categorical covariates can be found in \cite{muller2011} and \cite{page2018}. The cohesion and the similarity functions (for continuous and categorical covariates) assumed by the MFPPMx model are

\begin{equation}\label{prior_rho_c}
	c(S_j)=M\times(n_j-1)!
	\;,\quad j=1,\dots,J,
\end{equation}

\noindent and

\begin{equation}\label{prior_rho_g}
	g(\mathcal{X}_j) = \prod_{\ell=1}^{\Ell}g_\ell(\bm{x}_{j,\ell})
	\;,\quad j=1,\dots,J,
\end{equation}

\noindent respectively, where $M=1$,  $n_j$ is the number of individuals in the cluster $j$ and $\bm{x}_{j,\ell}=(x_{i,\ell})_{i\in S_j}$ is the vector of values $x_{i,\ell}$ of the $\ell$th-covariate for the individuals $i$ allocated to cluster $j$, for $\ell=1,\dots,\Ell$. In particular, the MFPPMx assumes

\begin{equation}\label{prior_rho_g_con}
	\begin{aligned}
		g_\ell(\bm{x}_{j,\ell}) &= \!\int\prod_{i\in S_j}
		\N(x_{i,\ell} \,;\, m_j,1)
		\N(m_j \,;\, 0,10)
		\d m_j\\
		&=
		\frac{s}{\sqrt{10(2\pi)^{n_j}}}
		\exp\bigg\{
		\!-0.5\bigg[\sum_{i\in S_j}x_{i,\ell}^2 - s^2\bigg(\sum_{i\in S_j}x_{i,\ell}\bigg)^2\bigg]
		\bigg\}
	\end{aligned}
\end{equation}

\noindent if the $\ell$th-covariate is continuous, with $s^2=\big(n_j+10^{-1}\big)^{-1}$, and

\begin{equation}\label{prior_rho_g_cat}
	\begin{aligned}
		g_\ell(\bm{x}_{j,\ell}) &= \!\int
		\bigg( \prod_{i\in S_j} \prod_{c\in\mathcal{C}_\ell} \xi_{c}^{\bm{1}(x_{i,\ell}=c)}\bigg)
		\Dir(\bm{\xi} \,;\, \{\gamma_c\}_{c\in\mathcal{C}_\ell})
		\d\bm{\xi}
		\\[.1in]
		&=
		\frac
		{\Gamma(\sum_{c\in\mathcal{C}_\ell}\gamma_c)}
		{\prod_{c\in\mathcal{C}_\ell}\Gamma(\gamma_c)}
		\cdot\frac
		{\prod_{c\in\mathcal{C}_\ell}\Gamma(n_{j,c}+\gamma_c)}
		{\Gamma(n_j + \sum_{c\in\mathcal{C}_\ell}\gamma_c)}
	\end{aligned}
\end{equation}

\noindent if the $\ell$th-covariate is categorical, where $\mathcal{C}_\ell$ is the set of categories of the $\ell$th-covariate, $\bm{\xi}=(\xi_{c})_{c\in\mathcal{C}_\ell}$, with $\xi_{c}\in(0,1)$ for all $c\in\mathcal{C}_\ell$ and $\sum_{c\in\mathcal{C}_\ell}\xi_{c}=1$, and $n_{j,c}=\sum_{i\in S_j}\bm{1}(x_{i,\ell}=c)$ for all $c\in\mathcal{C}_\ell$.

The MFPPMx in \eqref{eq_mfppmx} and the MFRMMx model in \eqref{eq_model}  assume the same observational model and the same prior distributions for almost all parameters, except for  the allocation of individuals to clusters
for which we consider a random partition $\rho$ and for $\Theta_d$, $d=1,\dots,D$. The MFPPMx assumes that, for each dimension $d$, the prior distribution for the location parameter vectors $\bth_{1d},\dots,\bth_{Jd}$ in $\Theta_d$ does not include repulsion or any other type of dependency by assuming  $\bth_{jd}\v\bm{\mu}_d,\tau^{2}_{jd} \overset{\text{iid}}{\sim} \N_p\left(\bm{\mu},\tau^{2}_j K^{-1}\right)$ for $j=1,\dots,J$. For the parameters in $\bth_{jd}=(\theta_{jd,1},\dots,\theta_{jd,p})^\T$, we assume a first-order random walk  $\theta_{jd,\ell}=\theta_{jd,\ell-1}+u_{jd,\ell}$\,, for $\ell=2,\dots,p$, and $\theta_{jd,1}\sim\N(0,\tau^{2}_{jd})$, which  defines the matrix $K$. Assuming autocorrelation between adjacent elements of B-splines coefficients characterizes the Bayesian penalized B-splines, the P-splines \citep{lang2004, page2015}. {This strategy allows for a better control of the influence of} number and position of the knots of B-splines in the shape and smoothness of the curve. The P-spline method permits desirable smoothness to be achieved under a moderate number of equidistant knots. A small number of knots may result in excessively smooth curve shapes, while a large number of knots may cause overfitting. Additionally, treating the number and positions of the B-spline knots as unknown, parallel to a clustering estimation procedure, could be {computationally expensive.} 
 
 %Regarding the allocation of individuals into clusters, while the proposed model assumes covariate-dependent mixing weights $\pi_j(\bm{x}_i\v\ba_j)$ as the prior probability of the individual $i$ belongs to cluster $j$, in the MFPPMx the covariate-dependet product prior distribution assumed over the partition $\rho$ implies a covariate-dependet prior distribution on the allocation vector $\bm{z}$.

Based on \eqref{eq_mfppmx}, the joint posterior distribution of the MFPPMx model is given by

\begin{equation}\label{eq_posterior_MFPPMx}
	\begin{aligned}
		f(\bm{\Psi} \v \mathcal{Y}_1,\dots,\mathcal{Y}_m,\bm{x}_1,\dots,\bm{x}_m)
		%		f(\bm{\Psi} \v \bY,\bm{x}_1,\dots,\bm{x}_m)
		&\propto\;
		\Bigg[\prod_{i=1}^{m}
		\,\N_{n_i D}(\vect(\Yi) \,;\, \vect(\Bzi \!+\! H_i\Bi) \,,\, \Si)
		\Bigg]\\[.1in]
		&\hspace{-1.8in}\times\;
		\Bigg[\prod_{i=1}^{m}
		\,\IW(\Si \,;\, \Sigma_0,\omega)
		\Bigg]
		\times
		\Bigg[\prod_{d=1}^{D}
		\Bigg(\prod_{i=1}^{m}\,\N(\beta_{0id} ;\, \mu_{0d},\sigma^{2}_{0d})\Bigg)
		\,\N(\mu_{0d} ;\, 0,s^2_0)
		\,\IG(\sigma^{2}_{0d} ;\, a_0,b_0)
		\Bigg]
		\\[.1in]
		&\hspace{-1.8in}\times\;
		\Bigg[\prod_{i=1}^{m}\prod_{j=1}^{J}
		\left(\,\prod_{d=1}^{D}
		\N_p(\bbt_{id} ;\, \bth_{jd} \,, \lambda^{2}_{jd} I_p)
		\right)^{\bm{1}(z_i=j)}
		\Bigg]
		\times
		\Bigg[\prod_{j=1}^{J}
		c(S_j)g(\mathcal{X}_j)
		\Bigg]
		\\[.1in]
		&\hspace{-1.8in}\times\;
		\Bigg[\prod_{d=1}^{D}
		\Bigg(\prod_{j=1}^{J}
		\,\N_p(\bth_{jd} ;\, \bm{\mu}_d ,\, \tau^2_{jd} K^{-1})\,
		\,f(\lambda^{2}_{jd})\,
		\,\IG(\tau^2_{jd} ;\, a_\tau,b_\tau)
		\Bigg)
		%h\!\left(\Theta_d\right)
		\N_p(\bm{\mu}_d ;\, \bm{0},s^2_\mu I_p)
		\Bigg],
		%\\[.1in]
	\end{aligned}
\end{equation}

\noindent where $\bm{\Psi}$ represents here the set of all the parameters of the MFPPMx model. The MCMC algorithm presented in Section \ref{supp:mcmc} of this supplementary material can be used to sample from the joint posterior distribution in \eqref{eq_posterior_MFPPMx}. The MCMC algorithm  changes only in the split-merge step (see Section \ref{supp:SM} of this supplementary material) and in the sample of the cluster allocation vector $\bm{z}$. Based on \cite{page2015}, the samples from $\bm{z}$ can be done in the MCMC algorithm considering the following multinomial weights:

%\begin{equation}\label{full_zi_Int_MFPPMx}
%	\begin{aligned}
%		&\hspace{-2.52cm}\Pr(z_i\!=\!j \v -)\propto\!
%		\;|L_j|^{-\frac{1}{2}}
%		\exp\left\{-\frac{1}{2}\bigg[
%		\vect(\Theta_j)^\T L_j^{-1}\vect(\Theta_j)-m_i^\T V_i \,m_i
%		\bigg]\right\}
%		|V_i|^\frac{1}{2}
%		\,\frac
%		{c(S_j^{-i}\!\cup\{i\})\,g(\mathcal{X}_j^{-i}\!\cup \{\bm{x}_i\})}
%		{c(\{i\})\,g(\{\bm{x}_i\})},\\[.1in]
%		&\hspace{-.5in}\text{for $j\!=\!1,\dots,J$,}\\[.1in]
%		\Pr(z_i\!=\!J\!+\!1 \v -)
%		&\;\propto\;
%		|L_\ell|^{-\frac{1}{2}}
%		\exp\left\{-\frac{1}{2}\bigg[
%		\vect(\Theta_\ell)^\T L_\ell^{-1}\vect(\Theta_\ell)-m_i^\T V_i\,m_i
%		\bigg]\right\}
%		|V_i|^\frac{1}{2}
%		\,c(\{i\})g(\{\bm{x}_i\}),\\[.1in]
%		&\hspace{-.6in}\quad\text{with } \ell\!=\!J\!+\!1,
%	\end{aligned}
%\end{equation}

\begin{equation}\label{full_zi_Int_MFPPMx}
	\hspace{-0.8cm}\begin{aligned}
		\Pr(z_i\!=\!j \v -)
		&\propto\!
		\;|L_j|^{-\frac{1}{2}}
		\exp\left\{-\frac{1}{2}\bigg[
		\vect(\Theta_j)^\T L_j^{-1}\vect(\Theta_j)-m_i^\T V_i \,m_i
		\bigg]\right\}
		|V_i|^\frac{1}{2}
		\,\frac
		{c(S_j^{-i}\!\cup\{i\})\,g(\mathcal{X}_j^{-i}\!\cup \{\bm{x}_i\})}
		{c(\{i\})\,g(\{\bm{x}_i\})},
	\end{aligned}
\end{equation}

\noindent for $j=1,\dots,J$, and

\begin{equation}\label{full_zi_Int_MFPPMx_new}
	\begin{aligned}
		\Pr(z_i\!=\!J\!+\!1 \v -)
		&\;\propto\;
		|L_\ell|^{-\frac{1}{2}}
		\exp\left\{-\frac{1}{2}\bigg[
		\vect(\Theta_\ell)^\T L_\ell^{-1}\vect(\Theta_\ell)-m_i^\T V_i\,m_i
		\bigg]\right\}
		|V_i|^\frac{1}{2}
		\,c(\{i\})g(\{\bm{x}_i\}),
	\end{aligned}
\end{equation}

\noindent with $\ell\!=\!J\!+\!1$, for $i=1,\dots,m$, where $J$ is the number of clusters at the currently imputed state, $m_i$ and $V_i$ are as defined in \eqref{full_zi_Int}, the cohesion function $c$ and the similarity function $g$ are defined in \eqref{prior_rho_c} and \eqref{prior_rho_g}, respectively, $S_j^{-i}$ is the cluster $j$ with individual $i$ removed, $\mathcal{X}_j^{-i}$ represents the covariates for those individuals in cluster $j$ with the covariates of individual $i$, denoted by $\bm{x}_i$, removed, and $J\!+\!1$ denotes the index of an empty candidate cluster with parameters being sampled from their respective prior distributions. The full conditional distribution presented in \eqref{full_zi_Int_MFPPMx} and \eqref{full_zi_Int_MFPPMx_new} is a multivariate version of the full conditional distribution of $\bm{z}$ considered by \cite{page2015}, with the matrices $\Bi$, $i=1,\dots,m$, integrated out.

%\newpage
\newpage
\FloatBarrier
\section{Additional figures from the simulation study}\label{supp:sim4}

This section presents additional figures related to the simulated data of Section \ref{sec:simulation}.

\begin{figure}[!h]
\centering
	{$d=1$\hspace{2.4cm}$d=2$\hspace{2.4cm}$d=3$\hspace{2.4cm}$d=4$}\\[-.25in]
	%%% j = 1
	\raisebox{1.4cm}{\rotatebox[origin=c]{90}{$j=1$}}
	\hspace{-.7cm}\includegraphics[width=4.0cm, height=3.6cm, trim=0cm 0cm 0 0]{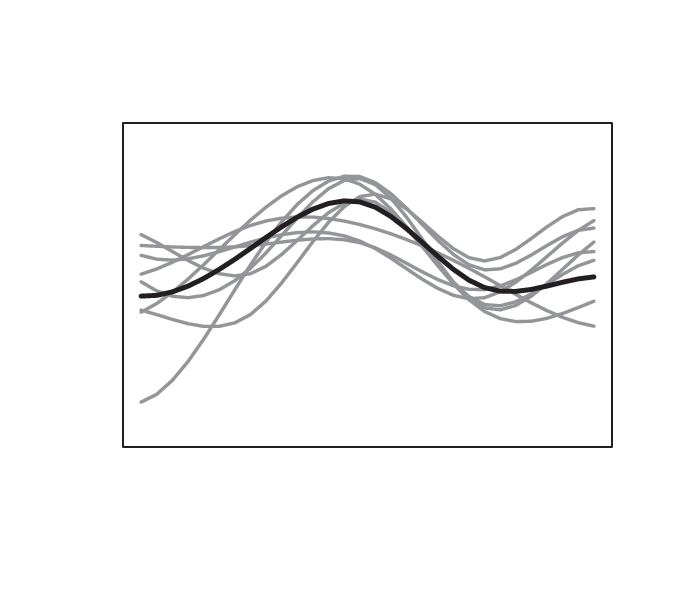}
	\label{fig_sim4_indiv_curves_d1_j1}\hspace{-.37in}
	\includegraphics[width=4.0cm, height=3.6cm, trim=0cm 0cm 0 0]{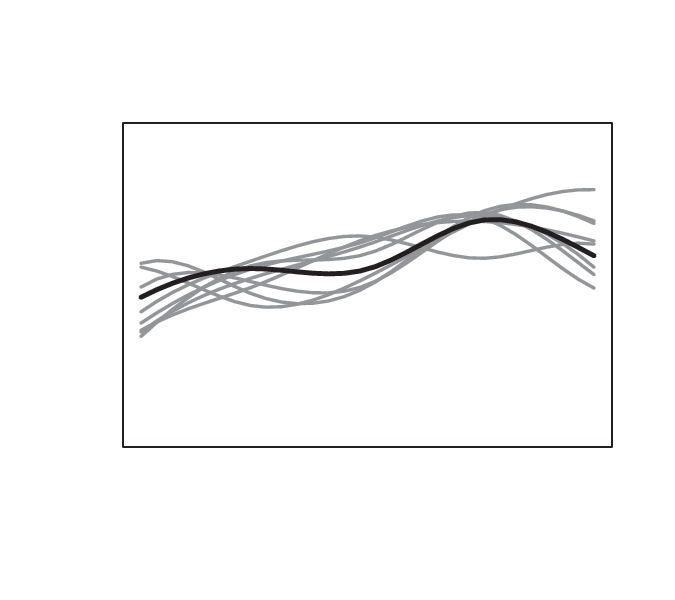}
	\label{fig_sim4_indiv_curves_d2_j1}\hspace{-.37in}
	\includegraphics[width=4.0cm, height=3.6cm, trim=0cm 0cm 0 0]{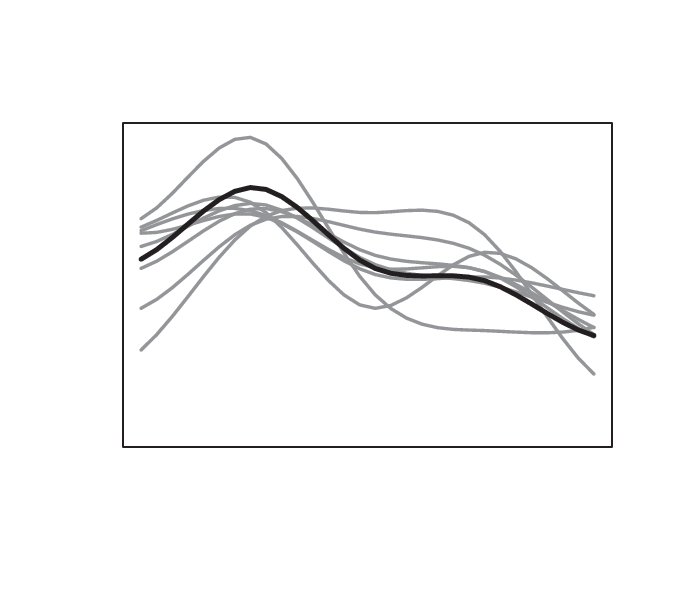}
	\label{fig_sim4_indiv_curves_d3_j1}\hspace{-.37in}
	\includegraphics[width=4.0cm, height=3.6cm, trim=0cm 0cm 0 0]{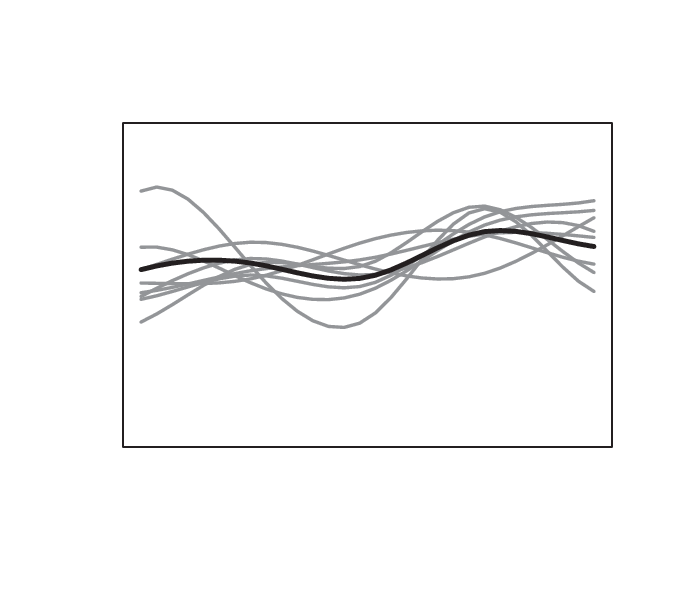}
	\label{fig_sim4_indiv_curves_d4_j1}
	\\[-.55in]
	%%% j = 2
	\raisebox{1.8cm}{\rotatebox[origin=c]{90}{$j=2$}}
	\hspace{-.7cm}\includegraphics[width=4.0cm, height=3.6cm, trim=0cm 0cm 0 0]{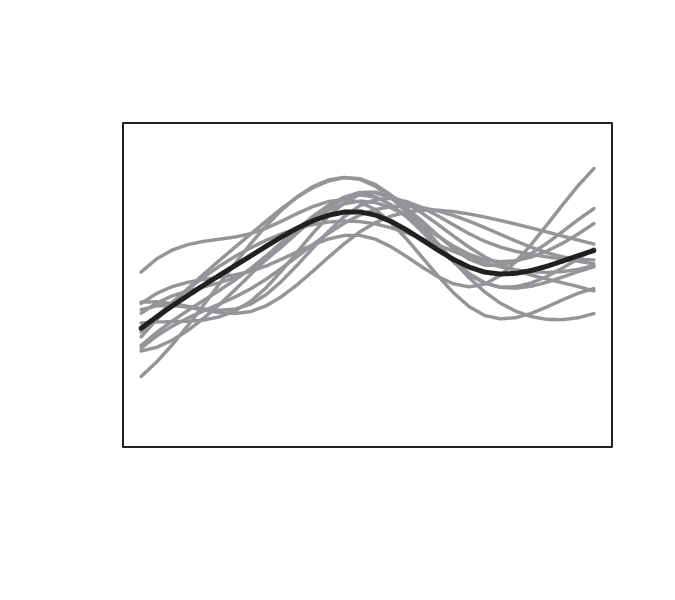}
	\label{fig_sim4_indiv_curves_d1_j2}\hspace{-.37in}
	\includegraphics[width=4.0cm, height=3.6cm, trim=0cm 0cm 0 0]{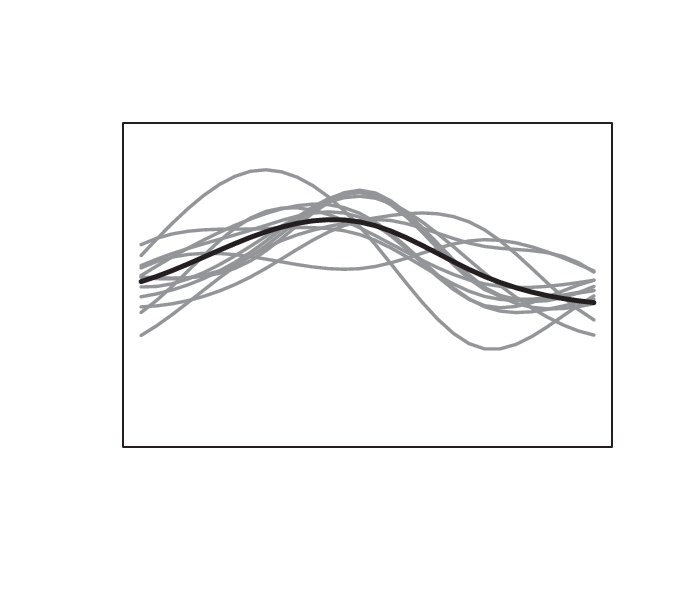}
	\label{fig_sim4_indiv_curves_d2_j2}\hspace{-.37in}
	\includegraphics[width=4.0cm, height=3.6cm, trim=0cm 0cm 0 0]{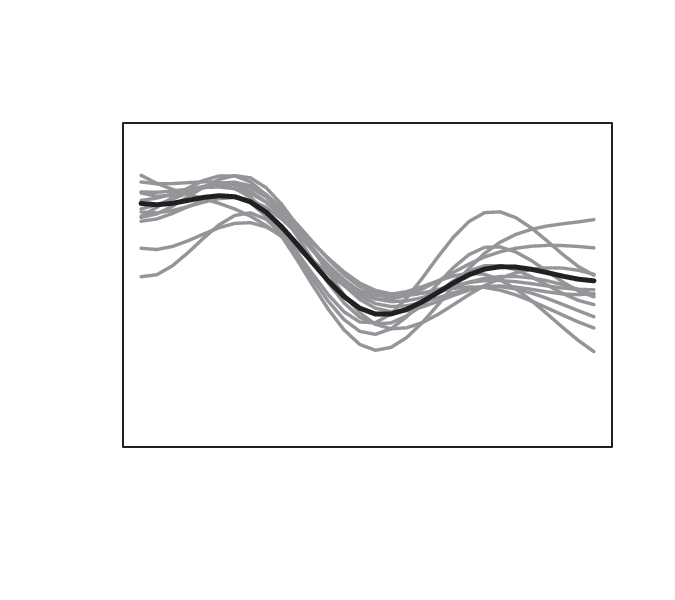}
	\label{fig_sim4_indiv_curves_d3_j2}\hspace{-.37in}
	\includegraphics[width=4.0cm, height=3.6cm, trim=0cm 0cm 0 0]{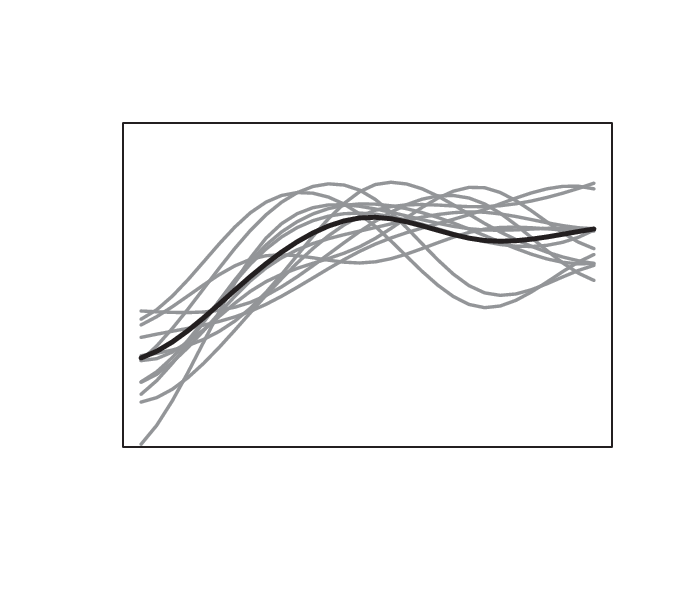}
	\label{fig_sim4_indiv_curves_d4_j2}
	\\[-.55in]
	%%% j = 3
	\raisebox{1.8cm}{\rotatebox[origin=c]{90}{$j=3$}}
	\hspace{-.7cm}\includegraphics[width=4.0cm, height=3.6cm, trim=0cm 0cm 0 0]{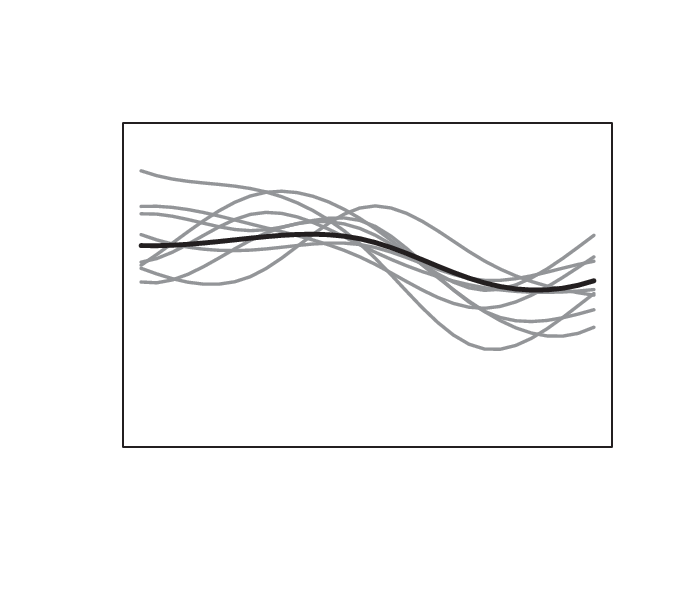}
	\label{fig_sim4_indiv_curves_d1_j3}\hspace{-.37in}
	\includegraphics[width=4.0cm, height=3.6cm, trim=0cm 0cm 0 0]{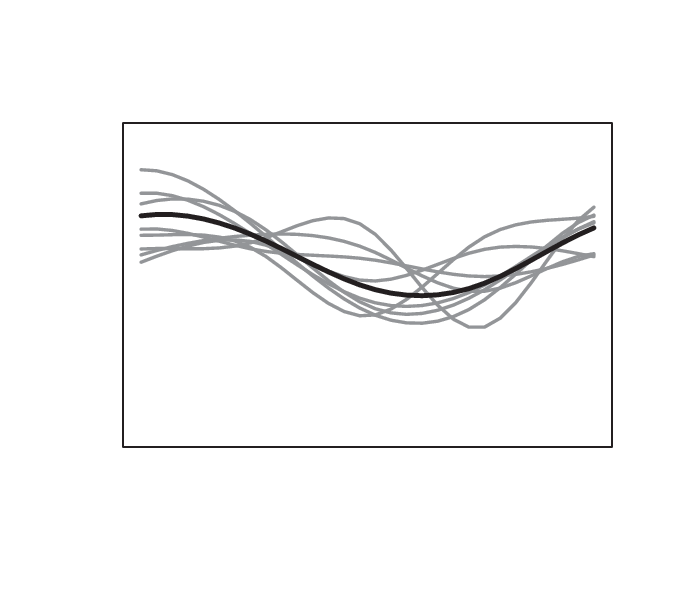}
	\label{fig_sim4_indiv_curves_d2_j3}\hspace{-.37in}
	\includegraphics[width=4.0cm, height=3.6cm, trim=0cm 0cm 0 0]{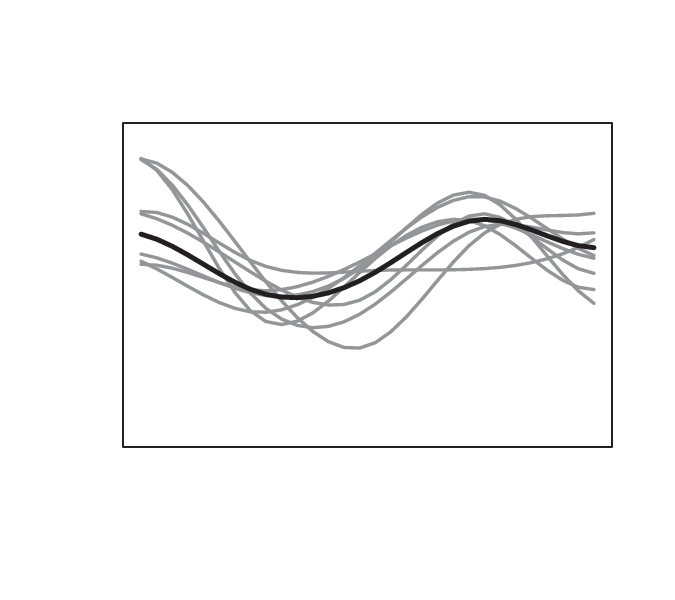}
	\label{fig_sim4_indiv_curves_d3_j3}\hspace{-.37in}
	\includegraphics[width=4.0cm, height=3.6cm, trim=0cm 0cm 0 0]{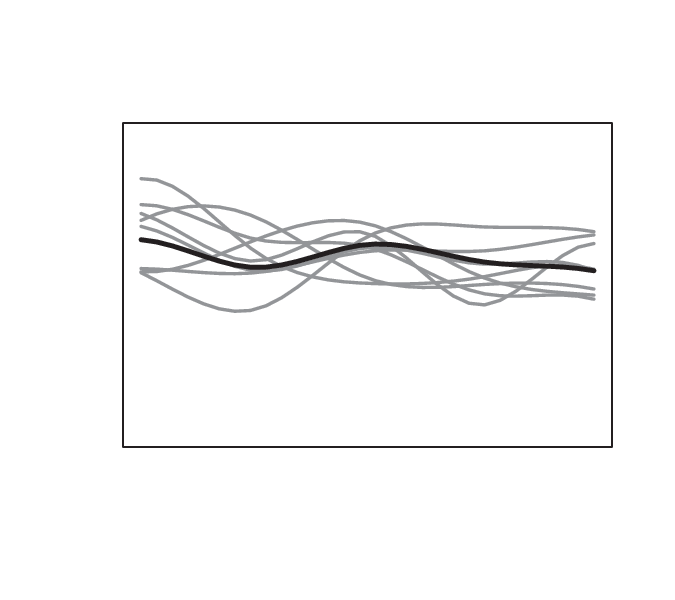}
	\label{fig_sim4_indiv_curves_d4_j3}
	\\[-.55in]
	%%% j = 4
	\raisebox{1.8cm}{\rotatebox[origin=c]{90}{$j=4$}}
	\hspace{-.7cm}\includegraphics[width=4.0cm, height=3.6cm, trim=0cm 0cm 0 0]{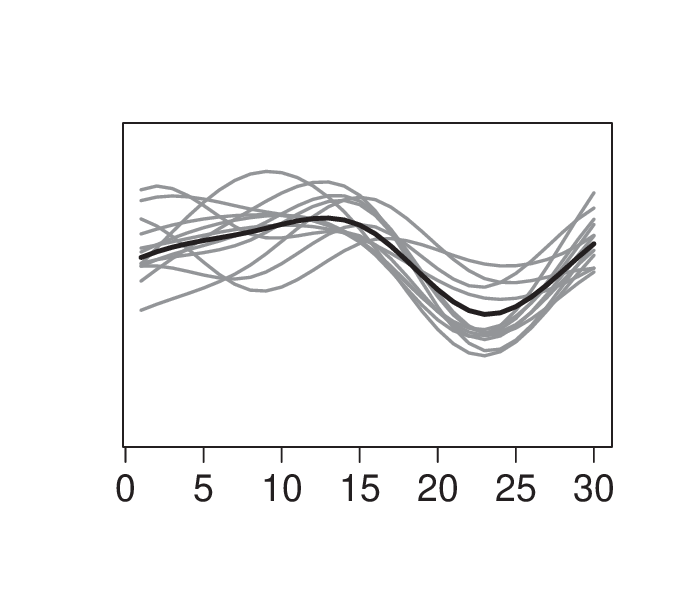}
	\label{fig_sim4_indiv_curves_d1_j4}\hspace{-.37in}
	\includegraphics[width=4.0cm, height=3.6cm, trim=0cm 0cm 0 0]{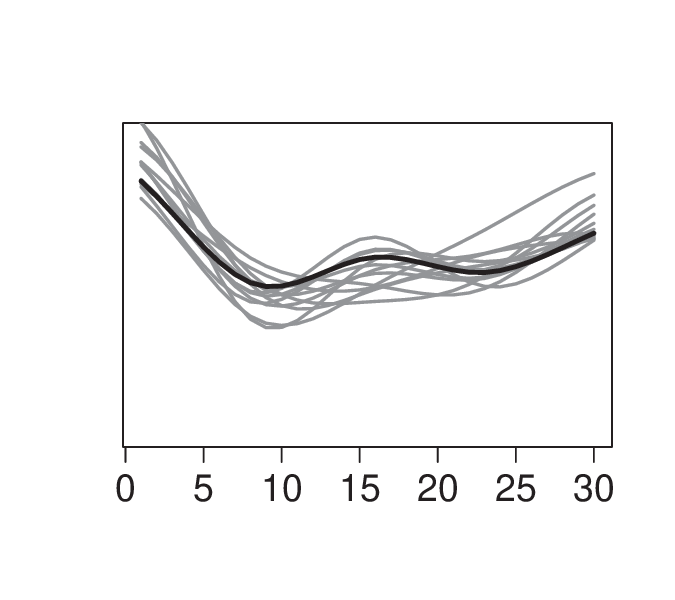}
	\label{fig_sim4_indiv_curves_d2_j4}\hspace{-.37in}
	\includegraphics[width=4.0cm, height=3.6cm, trim=0cm 0cm 0 0]{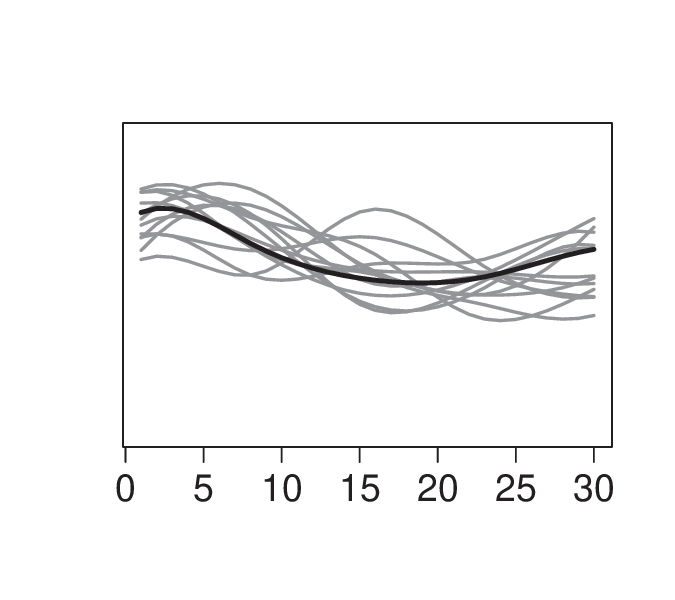}
	\label{fig_sim4_indiv_curves_d3_j4}\hspace{-.37in}
	\includegraphics[width=4.0cm, height=3.6cm, trim=0cm 0cm 0 0]{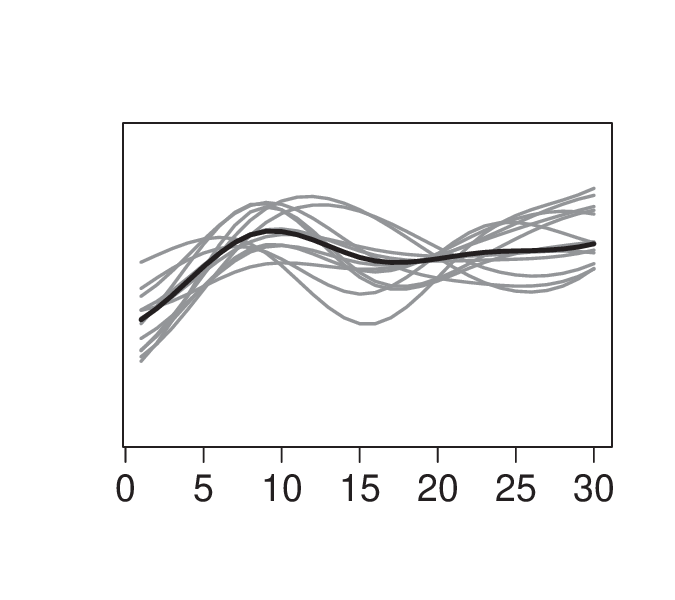}
	\label{fig_sim4_indiv_curves_d4_j4}
	%\\[-.55in]
\vspace{-.35in}
\caption{Cluster-specific mean curves (black) and respective individual curves (gray) of each cluster $j\!=\!1,\dots,4$ and dimension $d\!=\!1,\dots,4$.}
\label{fig:sim4_indiv_curves}
%\end{figure}
%
\vskip12pt
%
%\begin{figure}[b]
\centering
	{$d=1$\hspace{2.4cm}$d=2$\hspace{2.4cm}$d=3$\hspace{2.4cm}$d=4$}\\[-.25in]
	%%% j = 1
	\raisebox{1.8cm}{\rotatebox[origin=c]{90}{$j=1$}}
	\hspace{-.7cm}\includegraphics[width=4.0cm, height=3.6cm, trim=0cm 0cm 0 0]{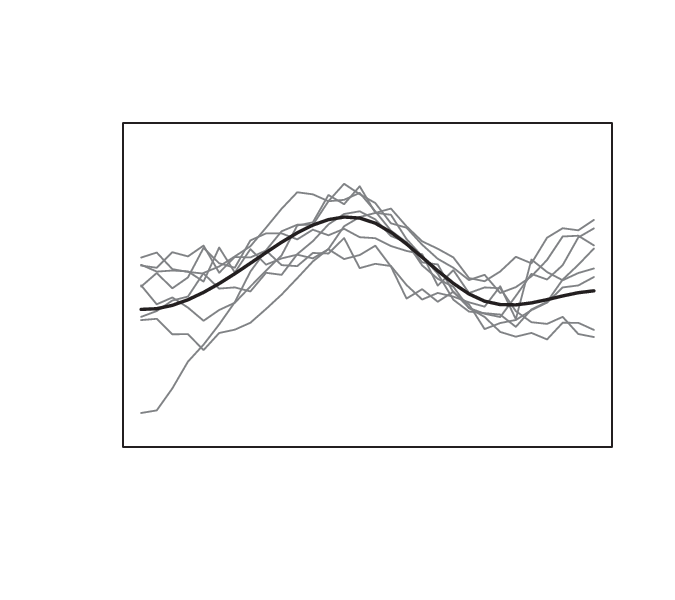}
	\label{fig_sim4_indiv_data_d1_j1}\hspace{-.37in}
	\includegraphics[width=4.0cm, height=3.6cm, trim=0cm 0cm 0 0]{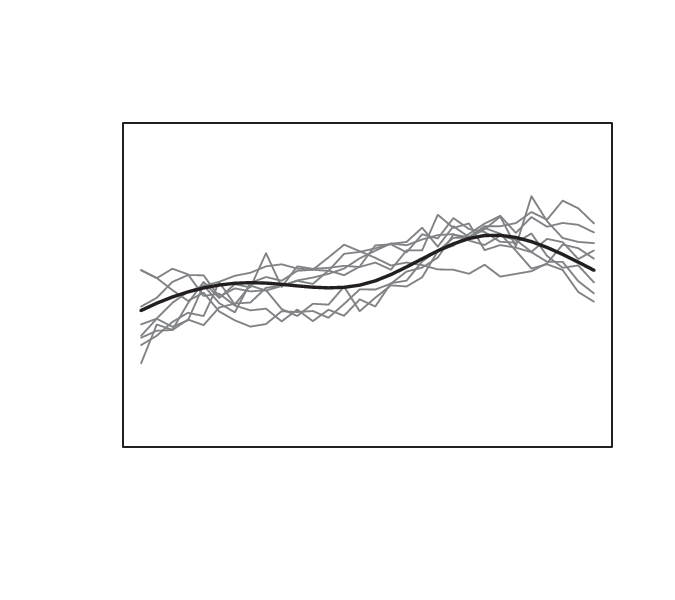}
	\label{fig_sim4_indiv_data_d2_j1}\hspace{-.37in}
	\includegraphics[width=4.0cm, height=3.6cm, trim=0cm 0cm 0 0]{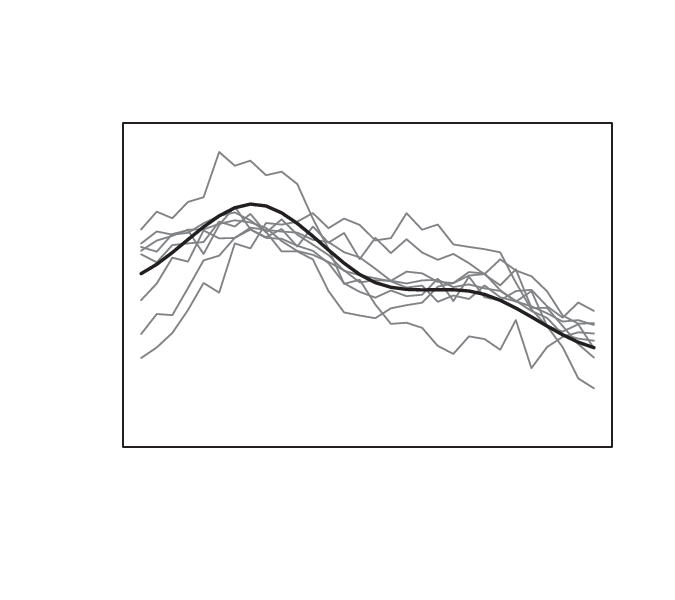}
	\label{fig_sim4_indiv_data_d3_j1}\hspace{-.37in}
	\includegraphics[width=4.0cm, height=3.6cm, trim=0cm 0cm 0 0]{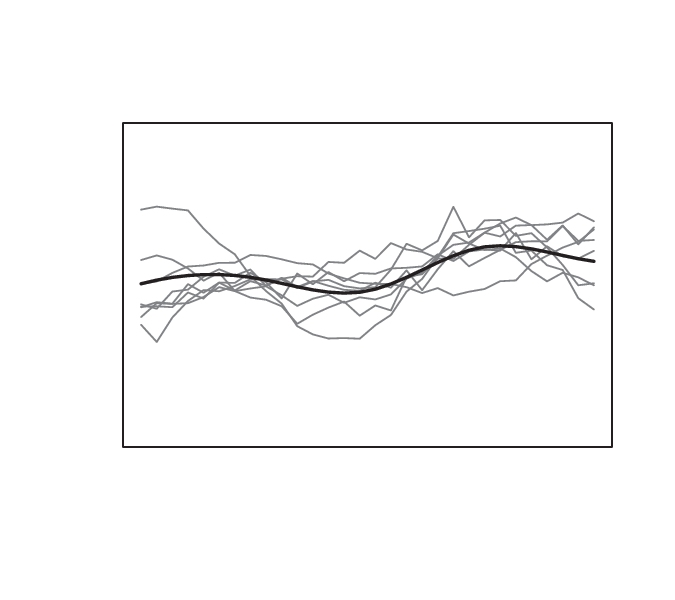}
	\label{fig_sim4_indiv_data_d4_j1}
	\\[-.55in]
	%%% j = 2
	\raisebox{1.8cm}{\rotatebox[origin=c]{90}{$j=2$}}
	\hspace{-.7cm}\includegraphics[width=4.0cm, height=3.6cm, trim=0cm 0cm 0 0]{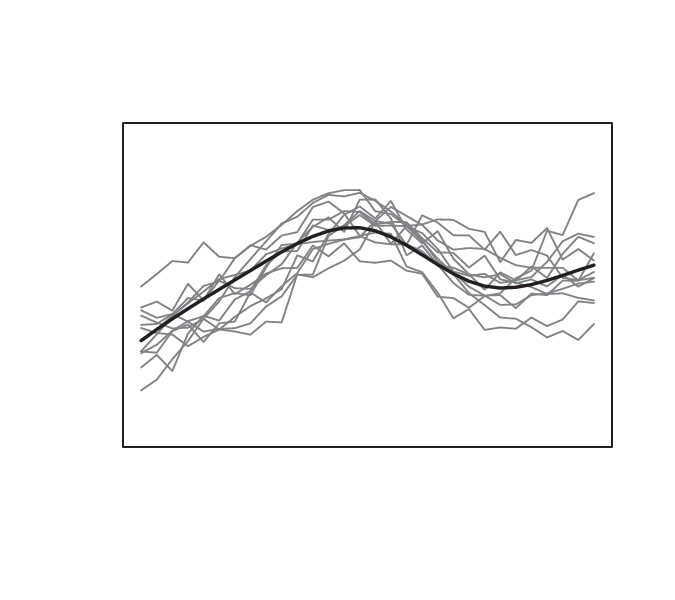}
	\label{fig_sim4_indiv_data_d1_j2}\hspace{-.37in}
	\includegraphics[width=4.0cm, height=3.6cm, trim=0cm 0cm 0 0]{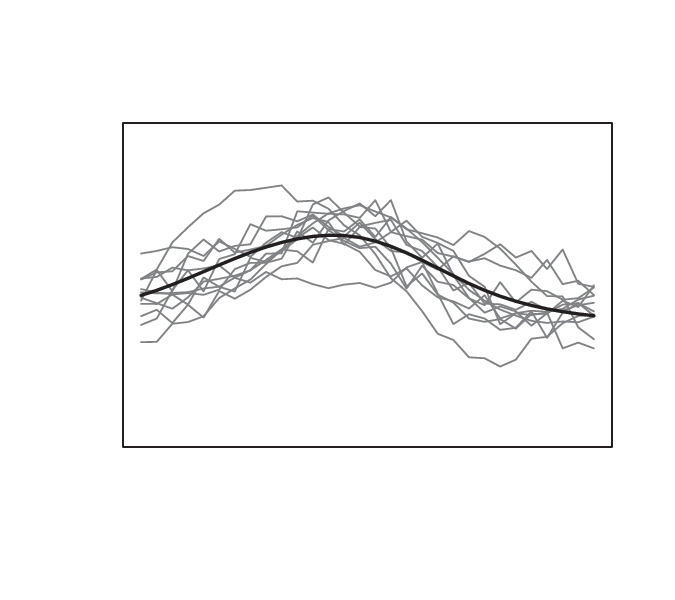}
	\label{fig_sim4_indiv_data_d2_j2}\hspace{-.37in}
	\includegraphics[width=4.0cm, height=3.6cm, trim=0cm 0cm 0 0]{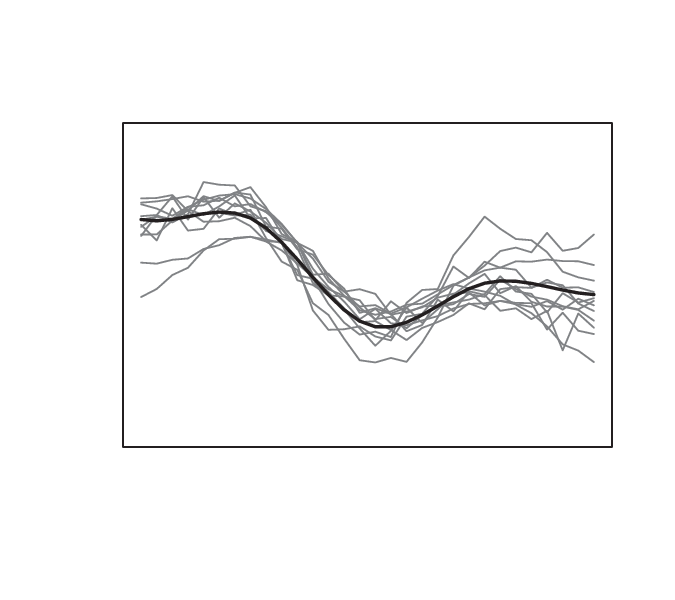}
	\label{fig_sim4_indiv_data_d3_j2}\hspace{-.37in}
	\includegraphics[width=4.0cm, height=3.6cm, trim=0cm 0cm 0 0]{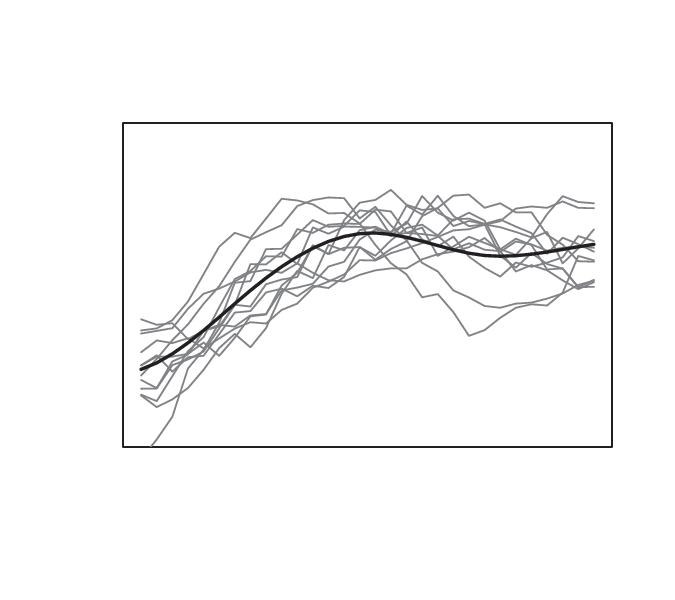}
	\label{fig_sim4_indiv_data_d4_j2}
	\\[-.55in]
	%%% j = 3
	\raisebox{1.8cm}{\rotatebox[origin=c]{90}{$j=3$}}
	\hspace{-.7cm}\includegraphics[width=4.0cm, height=3.6cm, trim=0cm 0cm 0 0]{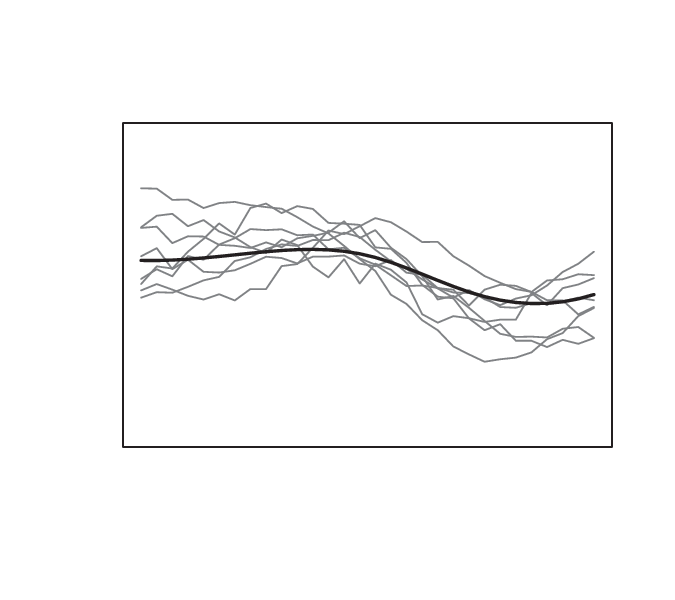}
	\label{fig_sim4_indiv_data_d1_j3}\hspace{-.37in}
	\includegraphics[width=4.0cm, height=3.6cm, trim=0cm 0cm 0 0]{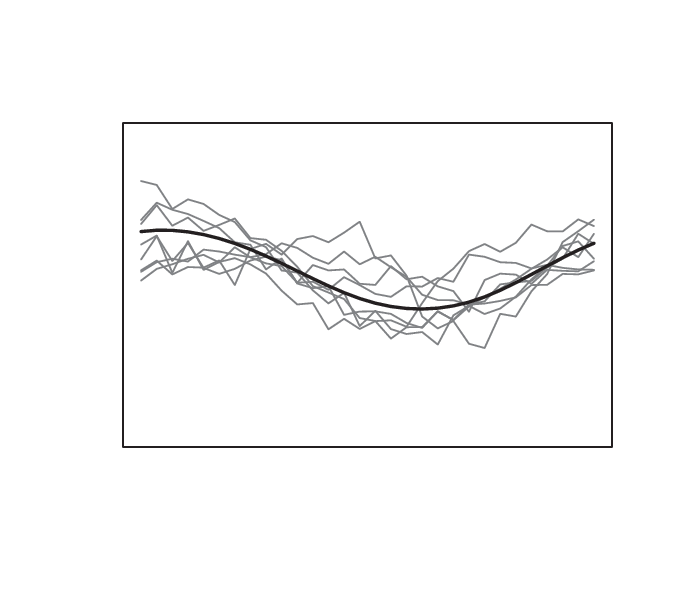}
	\label{fig_sim4_indiv_data_d2_j3}\hspace{-.37in}
	\includegraphics[width=4.0cm, height=3.6cm, trim=0cm 0cm 0 0]{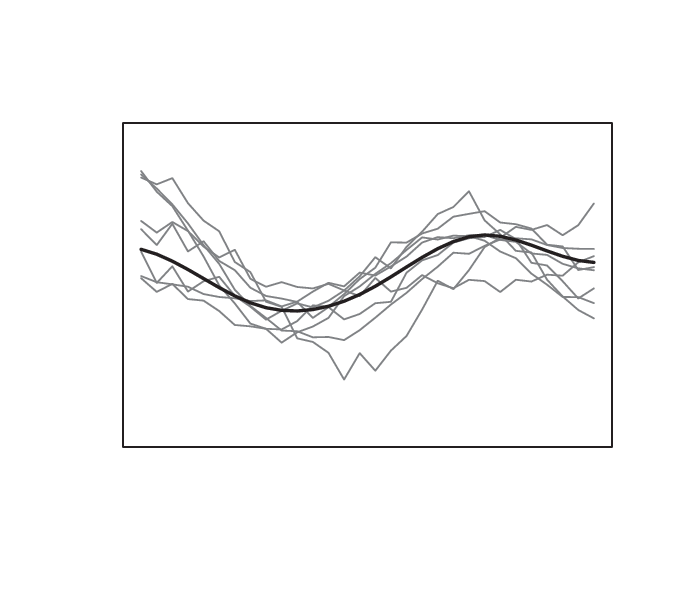}
	\label{fig_sim4_indiv_data_d3_j3}\hspace{-.37in}
	\includegraphics[width=4.0cm, height=3.6cm, trim=0cm 0cm 0 0]{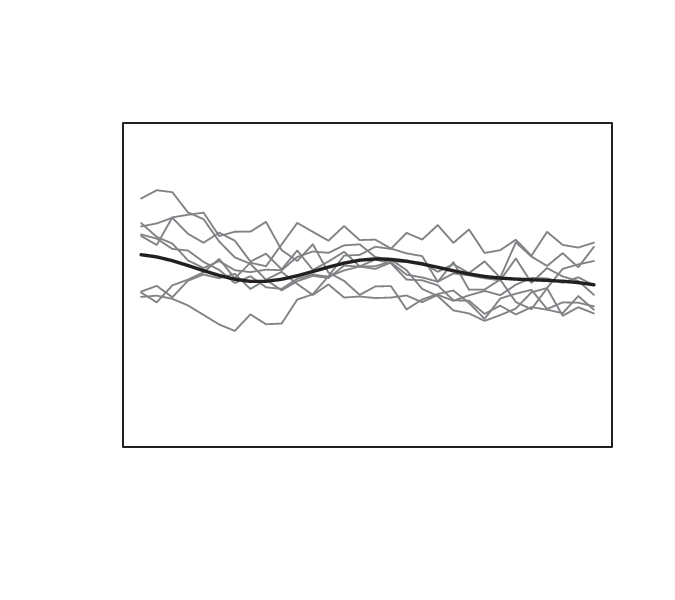}
	\label{fig_sim4_indiv_data_d4_j3}
	\\[-.55in]
	%%% j = 4
	\raisebox{1.8cm}{\rotatebox[origin=c]{90}{$j=4$}}
	\hspace{-.7cm}\includegraphics[width=4.0cm, height=3.6cm, trim=0cm 0cm 0 0]{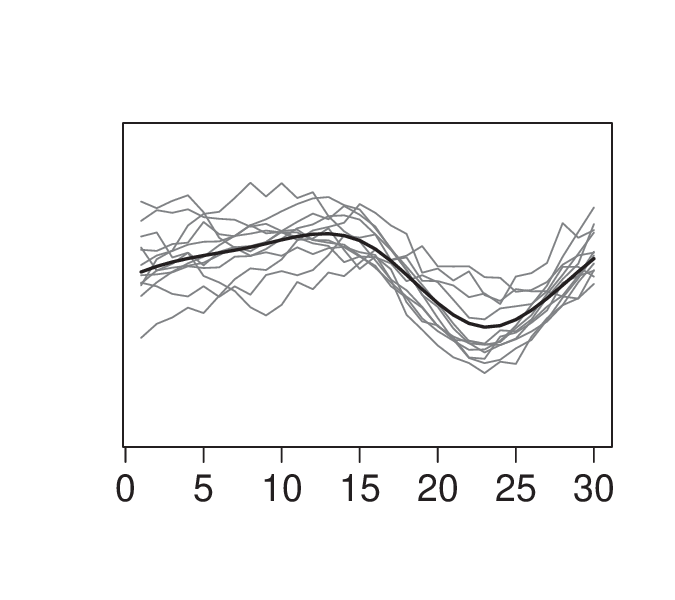}
	\label{fig_sim4_indiv_data_d1_j4}\hspace{-.37in}
	\includegraphics[width=4.0cm, height=3.6cm, trim=0cm 0cm 0 0]{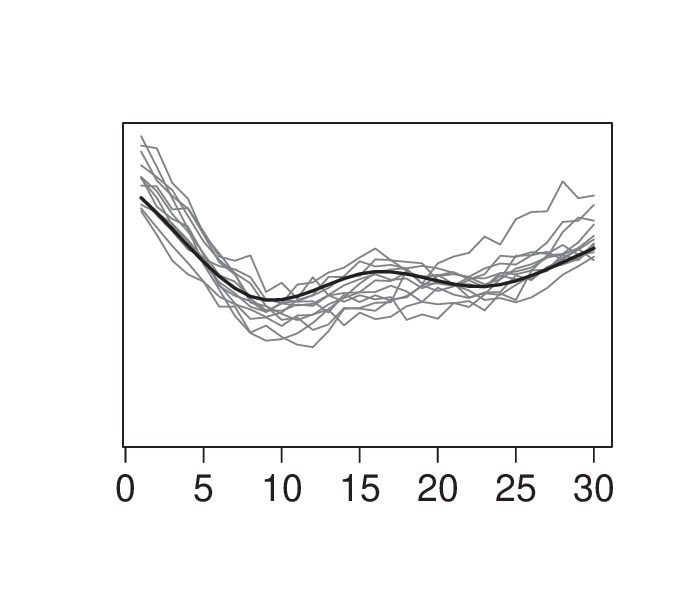}
	\label{fig_sim4_indiv_data_d2_j4}\hspace{-.37in}
	\includegraphics[width=4.0cm, height=3.6cm, trim=0cm 0cm 0 0]{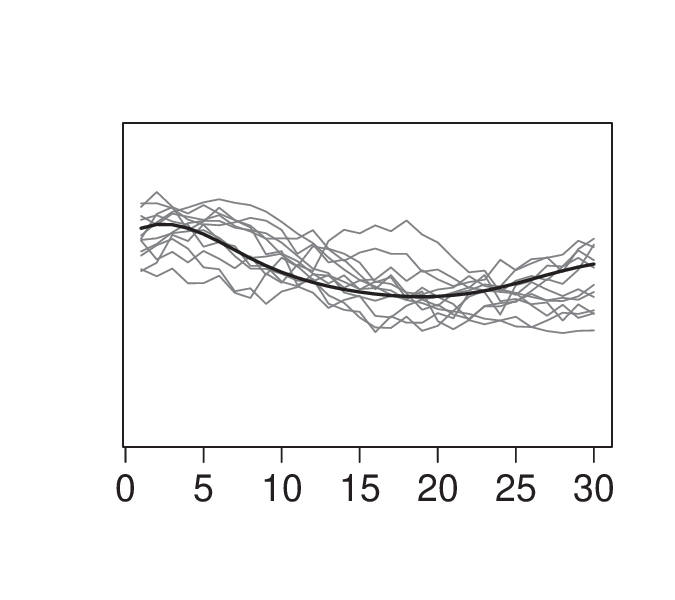}
	\label{fig_sim4_indiv_data_d3_j4}\hspace{-.37in}
	\includegraphics[width=4.0cm, height=3.6cm, trim=0cm 0cm 0 0]{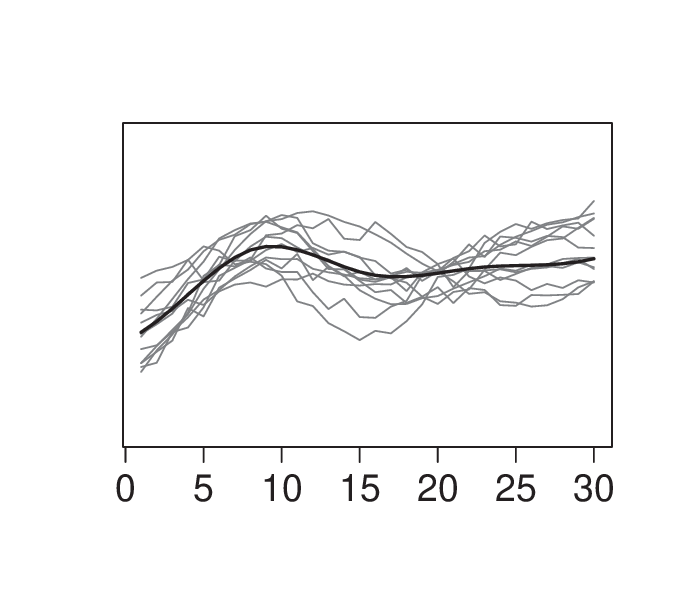}
	\label{fig_sim4_indiv_data_d4_j4}
	%\\[-.55in]
\vspace{-.35in}
\caption{Cluster-specific mean curves (black) and respective individual data sequences (gray) of each cluster $j\!=\!1,\dots,4$ and dimension $d\!=\!1,\dots,4$.}
\label{fig:sim4_indiv_data}
\end{figure}

%\newpage
\FloatBarrier
\newpage
\FloatBarrier
\section{Univariate case: NBA dataset}\label{supp:NBA}

We discuss in this section the application of the proposed model to the data set analysed by \cite{page2015}, refered to as NBA dataset, which consists of individual univariate sequences of game-by-game performance measurements, the so called Hollinger's Game Score (HGS) introduced by \cite{hollinger2002}, of active and retired players of the National Basketball Association (NBA). \cite{page2015} applied a univariate functional clustering model, which is equivalent to the MFPPMx model described in Section \ref{supp:mfppmx} with $D=1$, to cluster 408 NBA players, of which 263 are retired players and the remaining 145 are active players, that had data sequences considered unfinished. Their main goal was to predict the future performance of beginners and active NBA players by borrowing information from the curves of the retired players who were allocated to the same cluster. Our goal here is to evaluate the effect of repulsion on clustering of functional data with high noise level. We only consider the $m=263$ retired players for whom we have full data sequences. Just as was done by \cite{page2015}, three individual-specific covariates were considered to influence the clustering: age of the player in his first game, order of selection of the players in their respective NBA drafts (draft order) and the players experience before being drafted into the NBA. As illustrative examples, Figure \ref{fig:NBA_i} shows the sequences of game-by-game HGS values observed during the entire careers of four retired NBA players.

\begin{figure}[!htb]
	\centering
	\stackunder[5pt]{
		\raisebox{1.6cm}{\rotatebox[origin=c]{90}{\scriptsize HGS}}
		\includegraphics[width=7cm, height=3cm, trim=0 0.5cm 0 0]{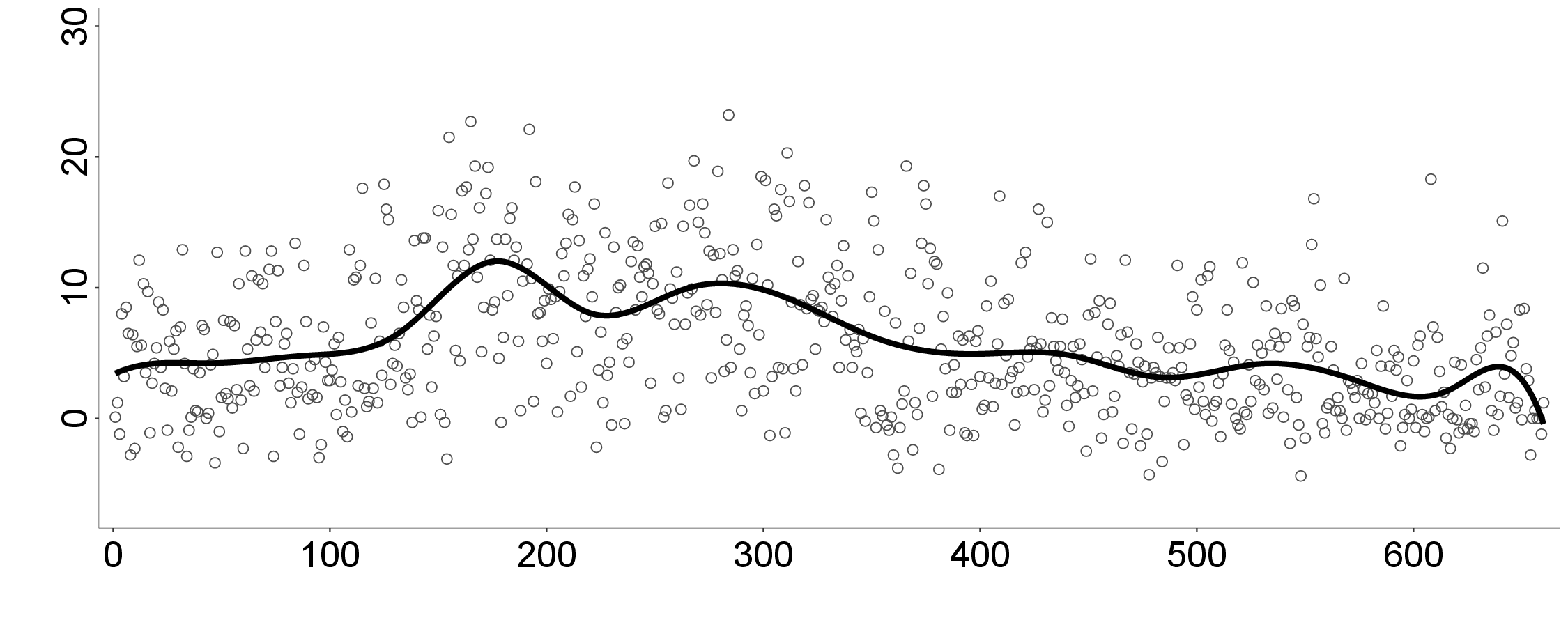}}{}
	\stackunder[5pt]{
		\includegraphics[width=7cm, height=3cm, trim=0 0.5cm 0 0]{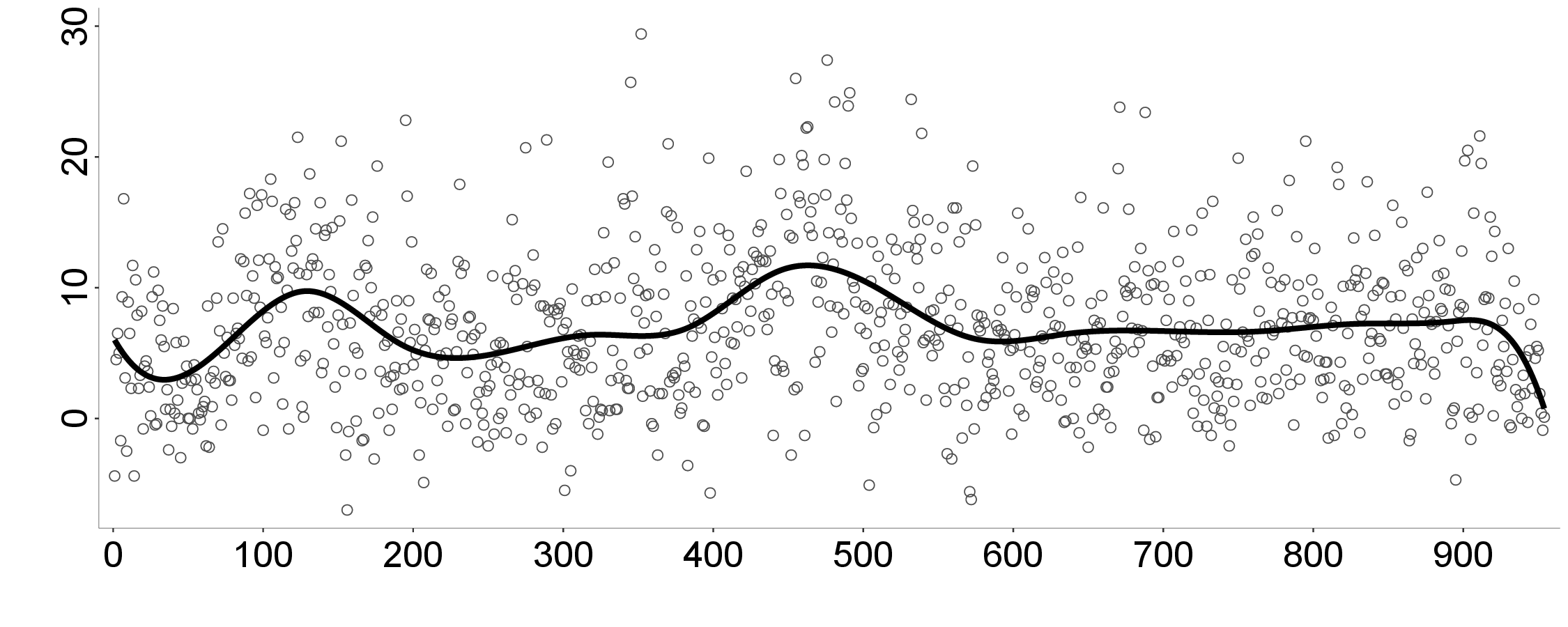}}{}
	\\%[-.1in]
	\stackunder[5pt]{
		\raisebox{1.6cm}{\rotatebox[origin=c]{90}{\scriptsize HGS}}
		\includegraphics[width=7cm, height=3cm, trim=0 0.5cm 0 0]{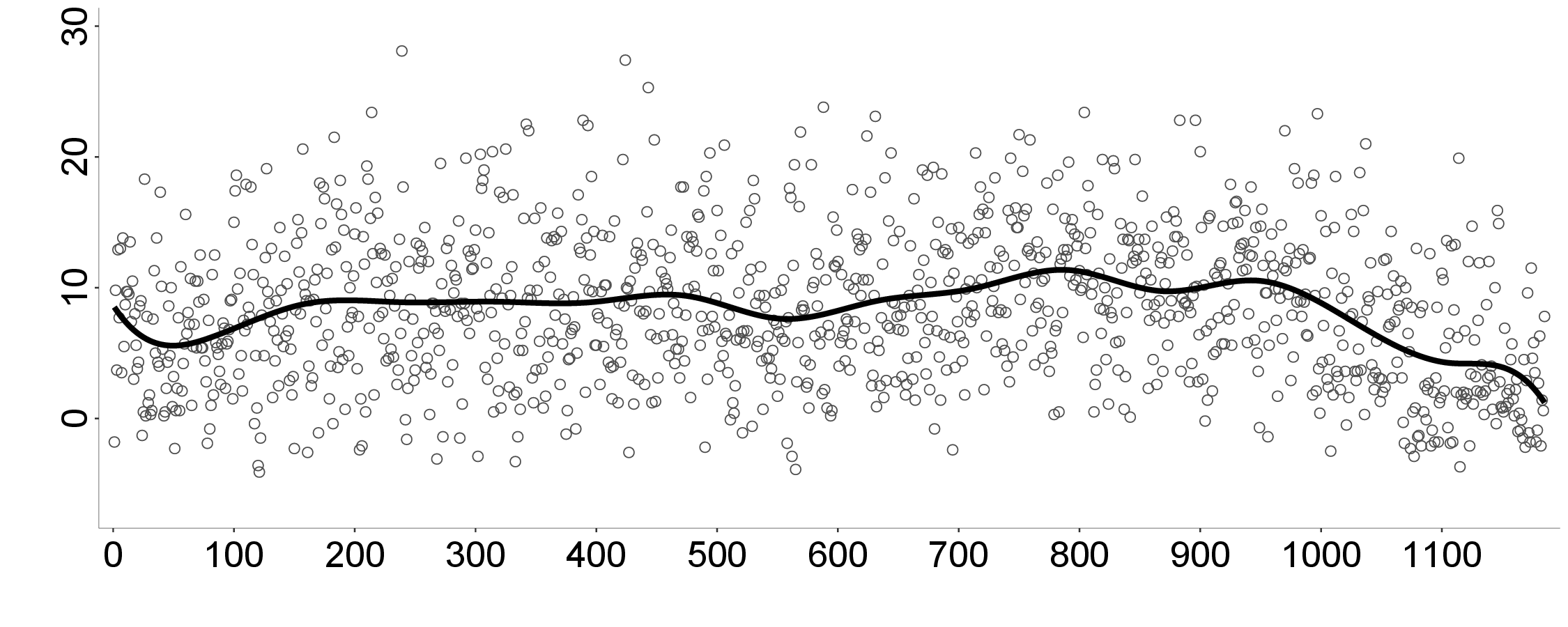}}{}
	\stackunder[5pt]{
		\includegraphics[width=7cm, height=3cm, trim=0 0.5cm 0 0]{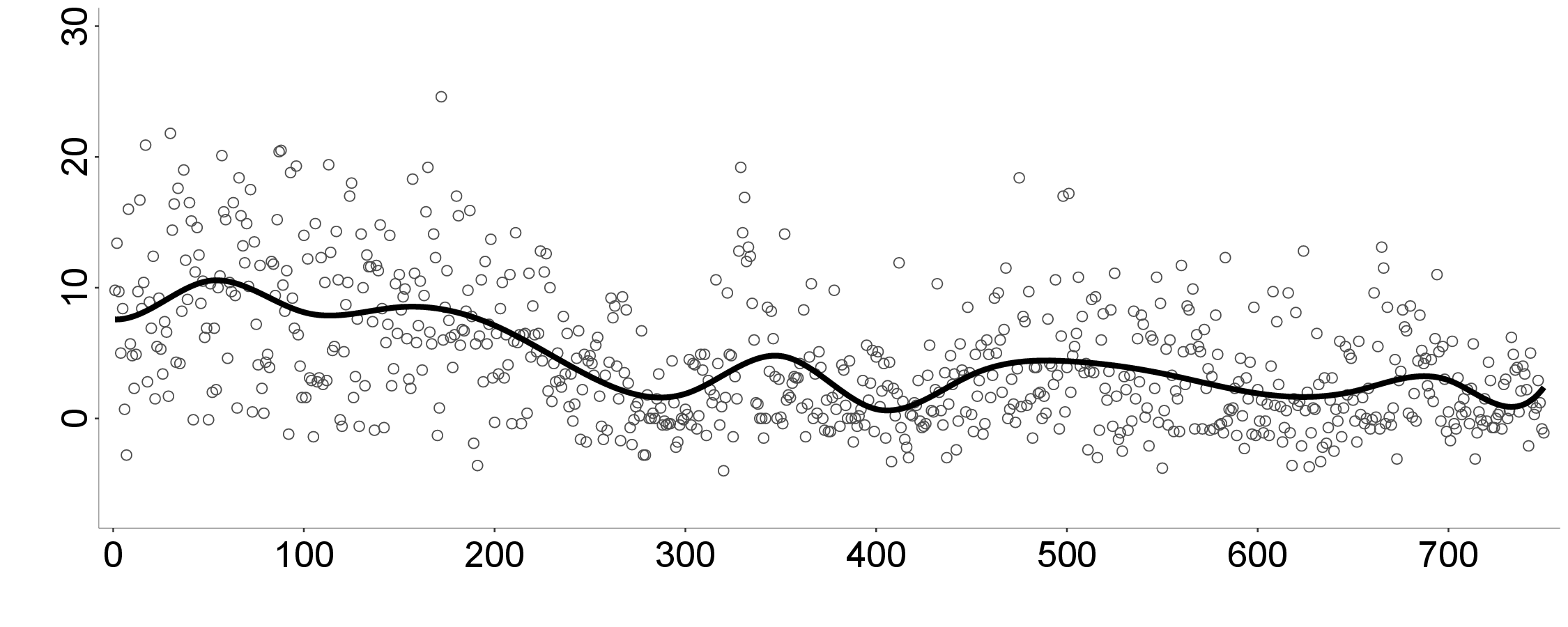}}{}
	\\[-.1in]
	{\scriptsize \hspace{0.8cm} Games \hspace{6.2cm} Games}\\
	%\vspace{-.05in}
	\caption{Game-by-game HGS measurements (dots) observed throughout the career of four NBA retired players and the respective cubic B-spline functions (solid lines).}
	\label{fig:NBA_i}
\end{figure}

\vspace{-.2in}
We fit the proposed MFRMMx model assuming different levels of repulsion $\phi=0,10^{-3},10^{-2},10^{-1},1,5,$ $10,20,50,100,200$, $D=1$ and $J=60$ mixture components. We also considered different number of knots $p=10,15$ and of the truncation parameter $A=0.1,1,10$. The other hyperparameter values are the same of the simulation study discussed in Section \ref{sec:simulation} of the main text, except for $b_\tau$, which we fixed at $b_\tau=0.05$, following \cite{page2015}. 

Table \ref{tab:NBA_nclr} shows that, as expected, the number of clusters reduces when the repulsive parameter $\phi$ increases. Clusters with only one observation, so called singletons, tend to disappear for the higher values of $\phi$. Table \ref{tab:NBA_nclr} also shows that the truncation parameter $A$ also influences the number of clusters. As observed in the simulation study (Section \ref{sec:simulation}), increasing the values of $A$ reduces the number of clusters. The truncation parameter $A$ also influences the heterogeneity level of the curves within each cluster. To illustrate, Figures \ref{fig:NBA_clusters_A01_p15_phi20} and \ref{fig:NBA_clusters_A1_p15_phi20} show the clusters identified by the MFRMMx with $A=0.1$ and $A=1$, respectively, both obtained with $p=15$ and $\phi=20$. In the case where $A=0.1$, Figure \ref{fig:NBA_clusters_A01_p15_phi20} shows that the individual curves and the cluster-specific mean curves are very similar, with no visually noticeable differences in the curve shapes within almost all clusters (except cluster $j\!=\!1$). On the other hand, when $A=1$, we can see in Figure \ref{fig:NBA_clusters_A1_p15_phi20} that the individual-specific mean curves present some heterogeneity within each cluster, being the individual and the cluster-specific mean curves different. Depending on the application, some heterogeneity in the curve shapes within each cluster may be interesting, along with the identification of well-differentiated clusters. For example, in \cite{page2015}, the future performance predictions of beginners and active NBA players who shared a common cluster would vary more or less around the mean curve of their respective cluster, under higher or lower values of the truncation parameter $A$, respectively. Moreover, well-differentiated future predictions between different clusters and a reduced number of singletons could be obtained by appropriate calibration of the repulsive parameter $\phi$.

\begin{table}\setlength{\tabcolsep}{0.05cm}%[!htb]
%\begin{adjustwidth}{-1.0cm}{0cm}
	\centering
	%\small
	\footnotesize
	\begin{tabular}{
%		l@{\hskip .2in}
%		rr@{\hskip .3in}rr@{\hskip .3in}
%		rr@{\hskip .3in}rr@{\hskip .3in}
%		rr@{\hskip .3in}rr
		l@{\hskip .2in}
		r@{\hskip .1in}r@{\hskip .2in}r@{\hskip .1in}r@{\hskip .5in}
		r@{\hskip .1in}r@{\hskip .2in}r@{\hskip .12in}r@{\hskip .5in}
		r@{\hskip .1in}r@{\hskip .2in}r@{\hskip .1in}r
		}
		\toprule\\[-.15in]
		\multicolumn{1}{c}{}
		& \multicolumn{4}{c}{\hspace{-.5in} $A\!=\!0.1$}
		& \multicolumn{4}{c}{\hspace{-.5in} $A\!=\!1$}
		& \multicolumn{4}{c}{\hspace{-.1in} $A\!=\!10$}
		\\[-.02in]
		\cmidrule(l{-.05in}r{.45in}){2-5}
		\cmidrule(l{-.05in}r{.45in}){6-9}
		\cmidrule(l{-.05in}r{.0in}){10-13}\\[-.15in]
		\multicolumn{1}{c}{}
		& \multicolumn{2}{c}{\hspace{-.25in} NC}
		& \multicolumn{2}{c}{\hspace{-.55in} NS}
		& \multicolumn{2}{c}{\hspace{-.25in} NC}
		& \multicolumn{2}{c}{\hspace{-.55in} NS}
		& \multicolumn{2}{c}{\hspace{-.25in} NC}
		& \multicolumn{2}{c}{\hspace{-.08in} NS}
		\\[-.02in]
		\cmidrule(l{-.05in}r{.2in}){2-3}
		\cmidrule(l{-.02in}r{.45in}){4-5}
		\cmidrule(l{-.05in}r{.18in}){6-7}
		\cmidrule(l{-.02in}r{.45in}){8-9}
		\cmidrule(l{-.05in}r{.18in}){10-11}
		\cmidrule(l{-.02in}r{.0in}){12-13}\\[-.15in]
		Model\hspace{1.0cm}$p\!=$
		& \multicolumn{2}{c}{\hspace{-.2in}10\hspace{.2in}15}
		& \multicolumn{2}{c}{\hspace{-.5in}10\hspace{.2in}15}
		& \multicolumn{2}{c}{\hspace{-.2in}10\hspace{.2in}15}
		& \multicolumn{2}{c}{\hspace{-.5in}10\hspace{.2in}15}
		& \multicolumn{2}{c}{\hspace{-.2in}10\hspace{.1in}15}
		& \multicolumn{2}{c}{\hspace{-.0in}10\hspace{.15in}15}\\[-.01in]
		\midrule\\[-.15in]
		MFPPMx    			       & 41.7 & 44.7 &  1.0 &  0.9 & 22.3 & 24.9 & 1.2 & 0.9 & 7.0 & 7.6 & 1.0 & 0.7 \\[-.01in]
		\midrule\\[-.15in]
		MFRMMx   	      & \multicolumn{12}{c}{} \\%[.05in]
		no repulsion      & 51.1 & 56.4 &  8.9 &  8.8 & 20.7 & 25.0 & 2.5 & 2.9 & 7.8 & 9.0 & 1.2 & 1.0 \\%[.01in]
		$\phi\!=\!0.001$  & 52.8 & 58.8 &  9.3 & 10.5 & 18.4 & 22.7 & 2.8 & 3.8 & 8.6 & 9.6 & 1.2 & 1.0 \\%[.01in]
		$\phi\!=\!0.01$   & 54.1 & 57.4 &  9.5 &  9.0 & 18.0 & 23.9 & 3.3 & 3.8 & 8.5 & 9.2 & 1.2 & 1.0 \\%[.01in]
		$\phi\!=\!0.1$    & 53.1 & 57.8 &  9.5 &  9.2 & 20.8 & 24.9 & 4.3 & 5.8 & 8.0 & 9.0 & 1.4 & 1.0 \\%[.01in]
		$\phi\!=\!1$      & 32.0 & 35.0 &  2.3 &  2.0 & 12.4 & 15.1 & 2.1 & 2.6 & 7.2 & 8.9 & 1.0 & 1.1 \\%[.01in]
		$\phi\!=\!5$      & 21.0 & 29.0 &    - &  1.0 & 11.9 & 10.1 & 1.2 & 1.9 & 6.5 & 8.2 & 1.4 & 1.6 \\%[.01in]
		$\phi\!=\!10$     & 21.4 & 22.0 &  0.8 &  2.0 & 11.3 & 12.6 & 2.5 & 3.9 & 5.6 & 6.9 & 1.0 & 1.2 \\%[.01in]
		$\phi\!=\!20$     & 16.0 & 15.0 &  1.0 &  1.0 & 10.6 &  9.5 & 3.0 & 1.8 & 4.9 & 5.7 & 0.8 & 1.1 \\%[.01in]
		$\phi\!=\!50$     & 10.0 & 11.0 &  1.0 &    - &  5.7 &  5.1 & 1.0 & 1.0 & 3.8 & 3.9 & 0.5 & 0.5 \\%[.01in]
		$\phi\!=\!100$    &  6.0 & 10.0 &    - &    - &  5.0 &  5.8 &   - &   - & 3.3 & 3.3 & 0.2 & 0.2 \\%[.01in]
		$\phi\!=\!200$    &  5.0 &  5.0 &    - &    - &  3.0 &  4.0 &   - &   - & 2.3 & 2.5 & 0.2 & 0.3 \\%[.02in]
		\bottomrule
	\end{tabular}
%\end{adjustwidth}
\vspace{.1in}
\caption{Number of clusters (column NC) and number of singletons (column NS) obtained by fitting the MFRMMx and MFPPMx models to the NBA dataset with $D=1$, under the influence of covariates, for different values of $p$, $A$ and $\phi$.}
\label{tab:NBA_nclr}
\end{table}

%\newpage
%\input{fig_NBA_clusters_A01_p15_phi20-ylim.tex}
%\input{fig_NBA_clusters_A01_p15_phi20.tex}
\begin{figure}[!h]
\begin{adjustwidth}{-2.5cm}{-2cm}
\centering
	\stackunder[-29pt]{
	%\raisebox{1.3cm}{\rotatebox[origin=c]{90}{\scriptsize $\phi=0.1$}}
		\includegraphics[width=4.0cm, height=4.4cm, trim=0 0cm 0 0]
		{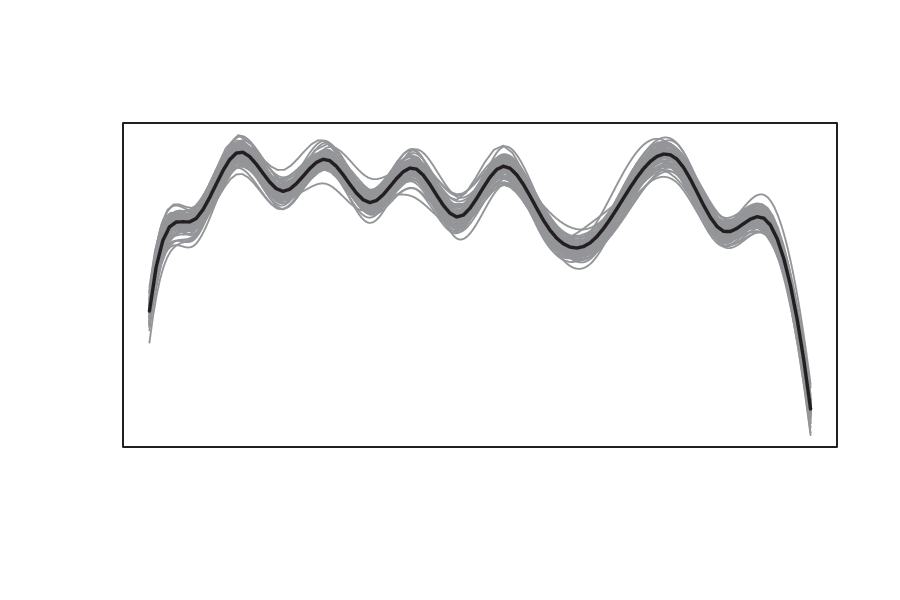}}{\quad$n_1=59$}
	\hspace{-.35in}
	\stackunder[-29pt]{
		\includegraphics[width=4.0cm, height=4.4cm, trim=0 0cm 0 0]
		{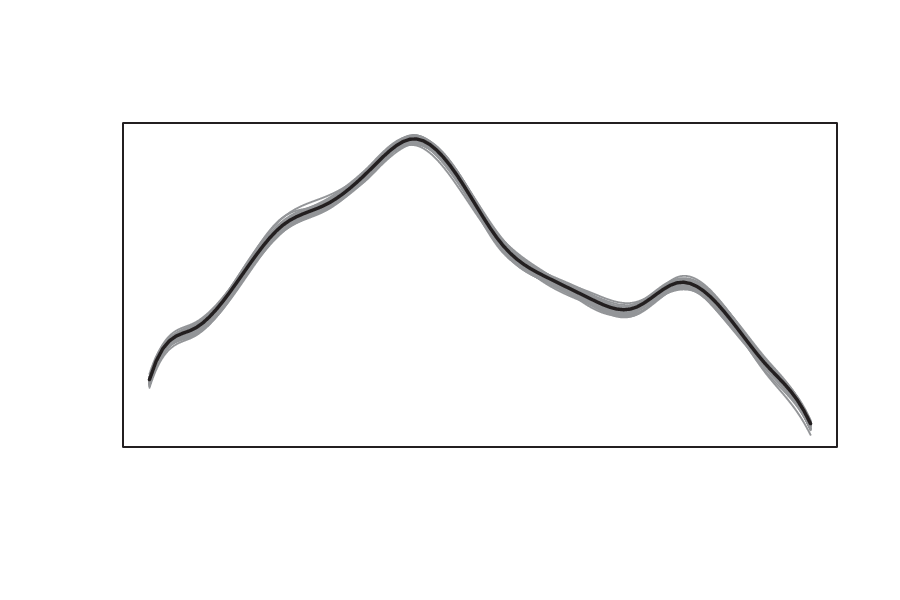}}{\quad$n_2=47$}
	\hspace{-.35in}
	\stackunder[-29pt]{
		\includegraphics[width=4.0cm, height=4.4cm, trim=0 0cm 0 0]
		{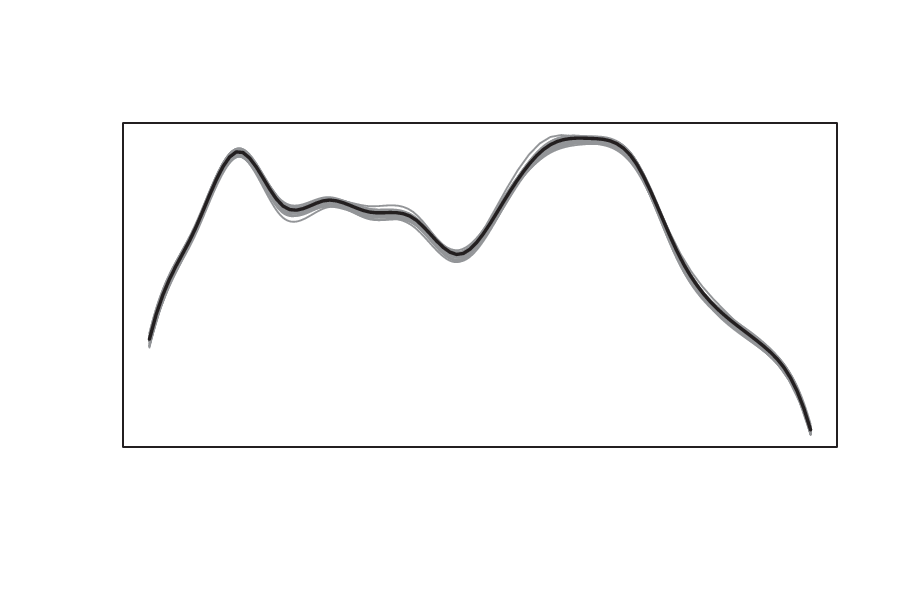}}{\quad$n_3=34$}
	\hspace{-.35in}
	\stackunder[-29pt]{
	%\raisebox{1.3cm}{\rotatebox[origin=c]{90}{\scriptsize $\phi=0.1$}}
		\includegraphics[width=4.0cm, height=4.4cm, trim=0 0cm 0 0]
		{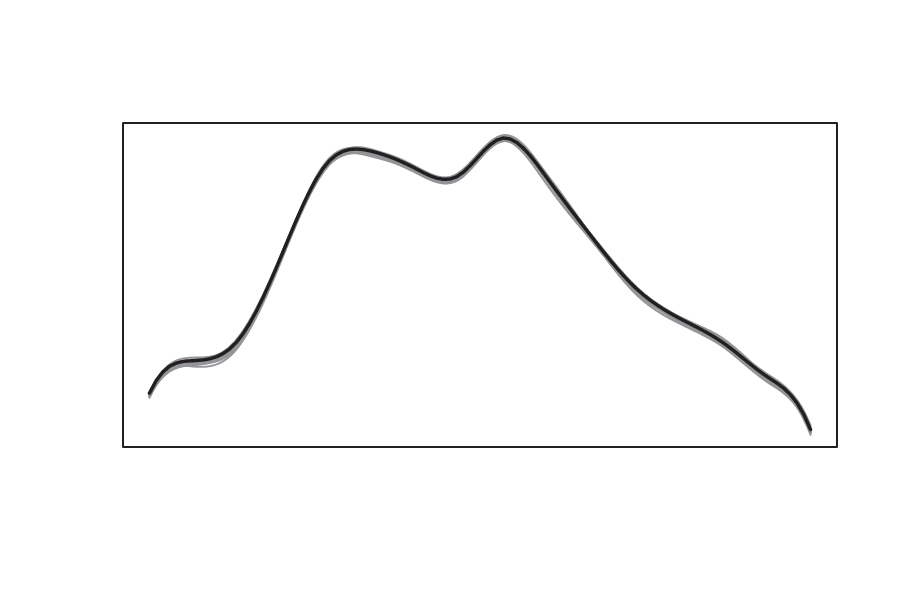}}{\quad$n_4=22$}
	\hspace{-.35in}
	\stackunder[-29pt]{
		\includegraphics[width=4.0cm, height=4.4cm, trim=0 0cm 0 0]
		{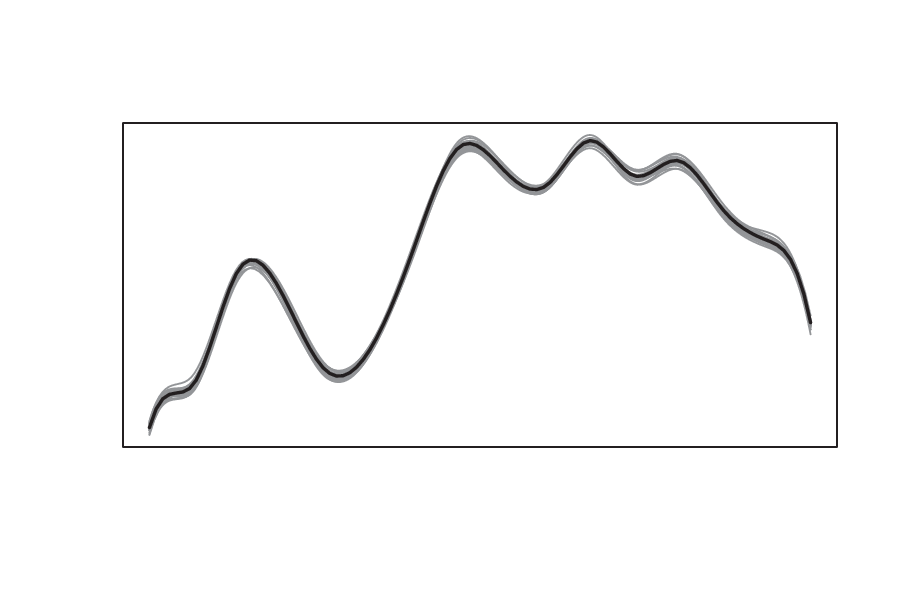}}{\quad$n_5=22$}
	\\[-.6in]
	\stackunder[-29pt]{
		\includegraphics[width=4.0cm, height=4.4cm, trim=0 0cm 0 0]
		{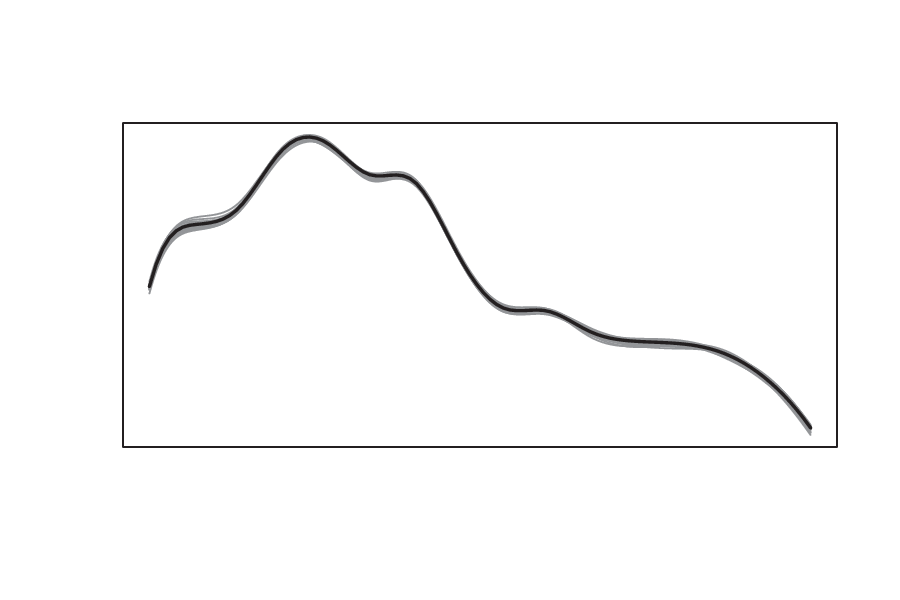}}{\quad$n_6=20$}
	\hspace{-.35in}
	\stackunder[-29pt]{
		\includegraphics[width=4.0cm, height=4.4cm, trim=0 0cm 0 0]
		{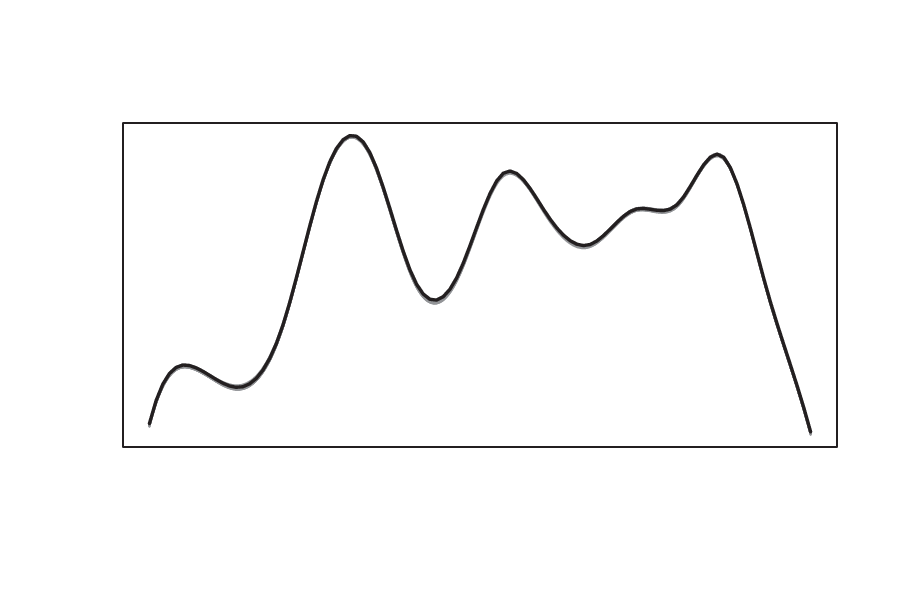}}{\quad$n_7=17$}
	\hspace{-.35in}
	\stackunder[-29pt]{
		\includegraphics[width=4.0cm, height=4.4cm, trim=0 0cm 0 0]
		{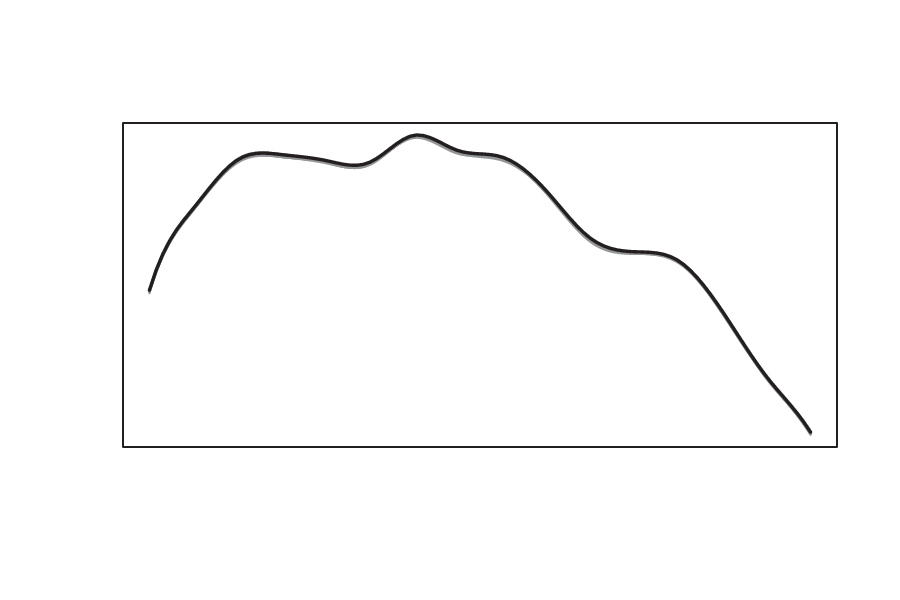}}{\quad$n_8=12$}
	\hspace{-.35in}
	\stackunder[-29pt]{
		\includegraphics[width=4.0cm, height=4.4cm, trim=0 0cm 0 0]
		{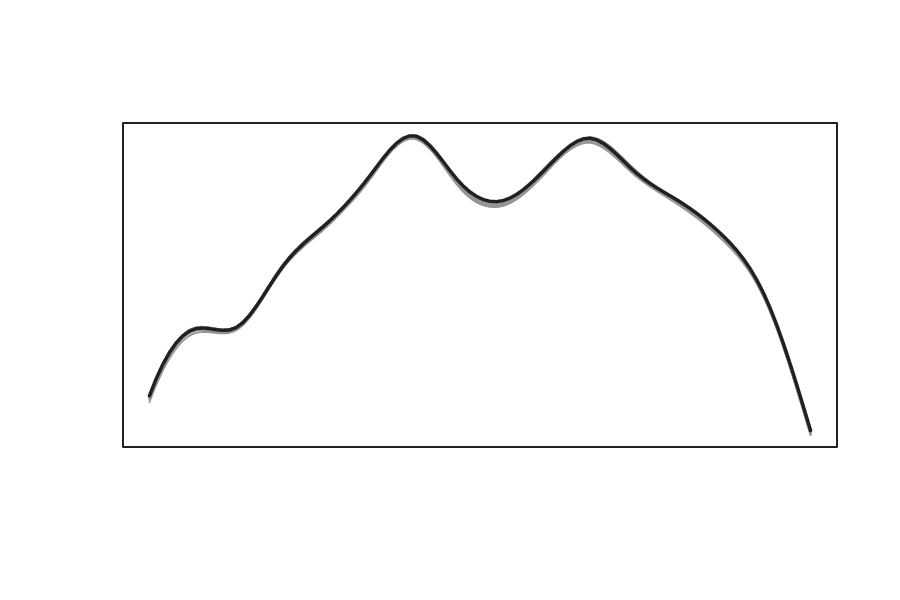}}{\quad$n_9=10$}
	\hspace{-.35in}
	\stackunder[-29pt]{
		\includegraphics[width=4.0cm, height=4.4cm, trim=0 0cm 0 0]
		{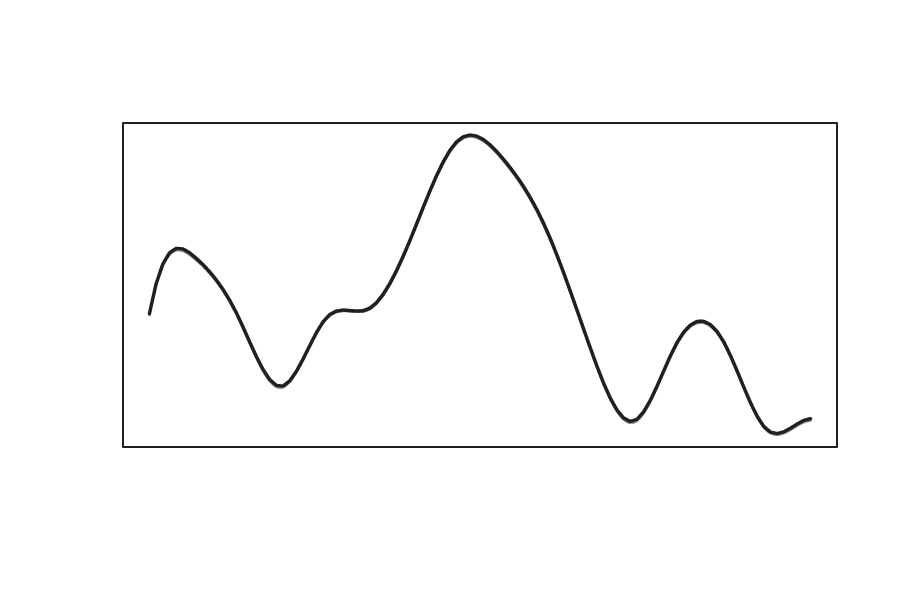}}{\quad$n_{10}=8$}
	\\[-.6in]
	\stackunder[-29pt]{
		\includegraphics[width=4.0cm, height=4.4cm, trim=0 0cm 0 0]
		{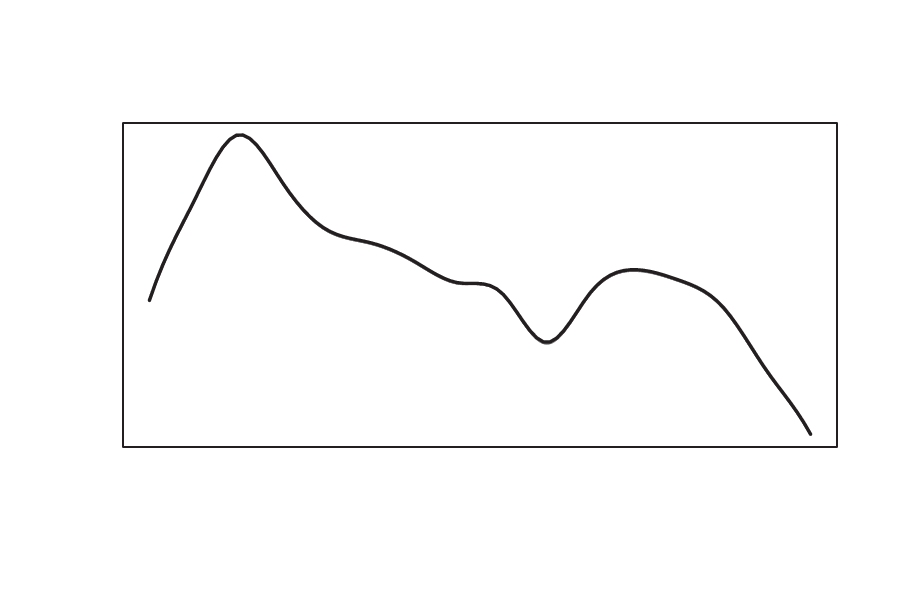}}{\quad$n_{11}=4$}
	\hspace{-.35in}
	\stackunder[-29pt]{
		\includegraphics[width=4.0cm, height=4.4cm, trim=0 0cm 0 0]
		{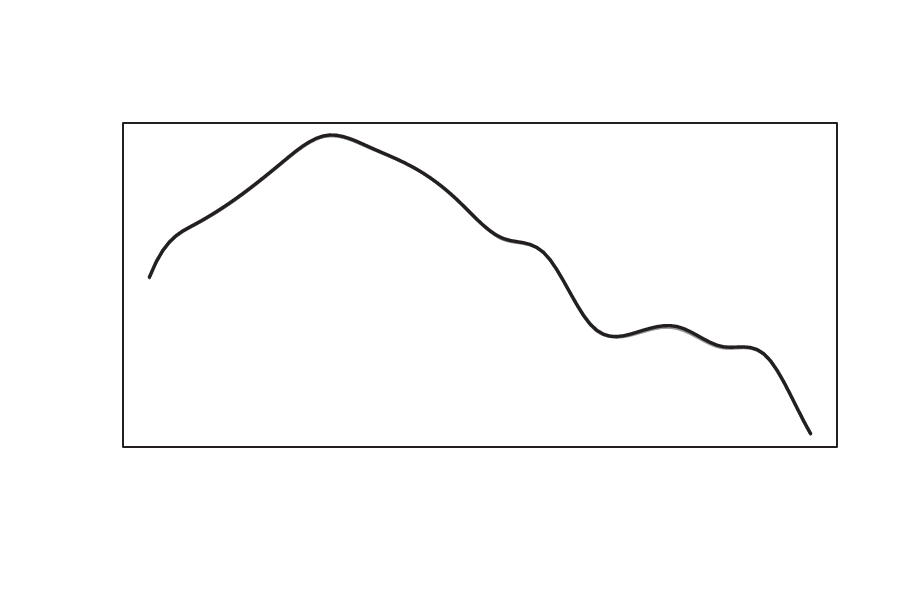}}{\quad$n_{12}=3$}
	\hspace{-.35in}
	\stackunder[-29pt]{
		\includegraphics[width=4.0cm, height=4.4cm, trim=0 0cm 0 0]
		{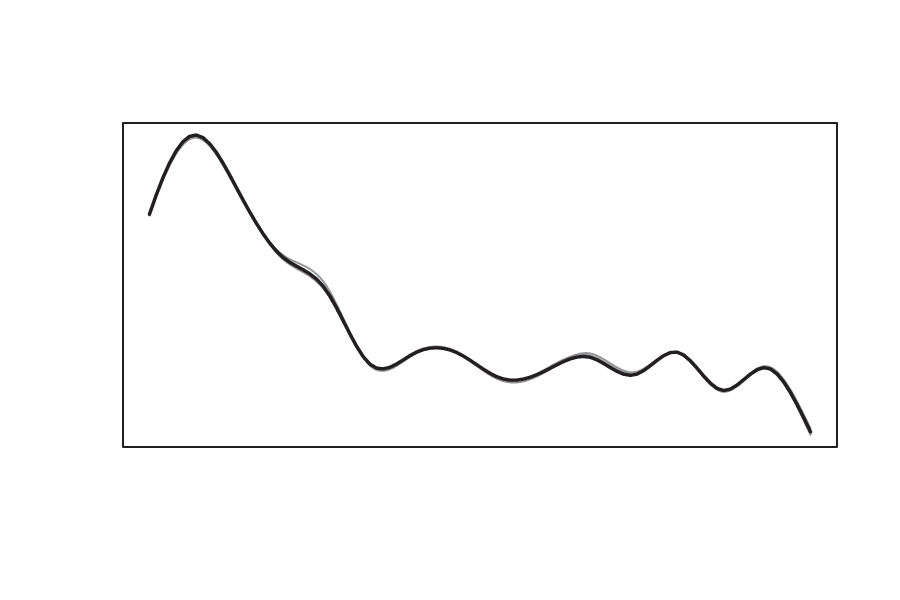}}{\quad$n_{13}=2$}
	\hspace{-.35in}
	\stackunder[-29pt]{
	%\raisebox{1.3cm}{\rotatebox[origin=c]{90}{\scriptsize $\phi=0.1$}}
		\includegraphics[width=4.0cm, height=4.4cm, trim=0 0cm 0 0]
		{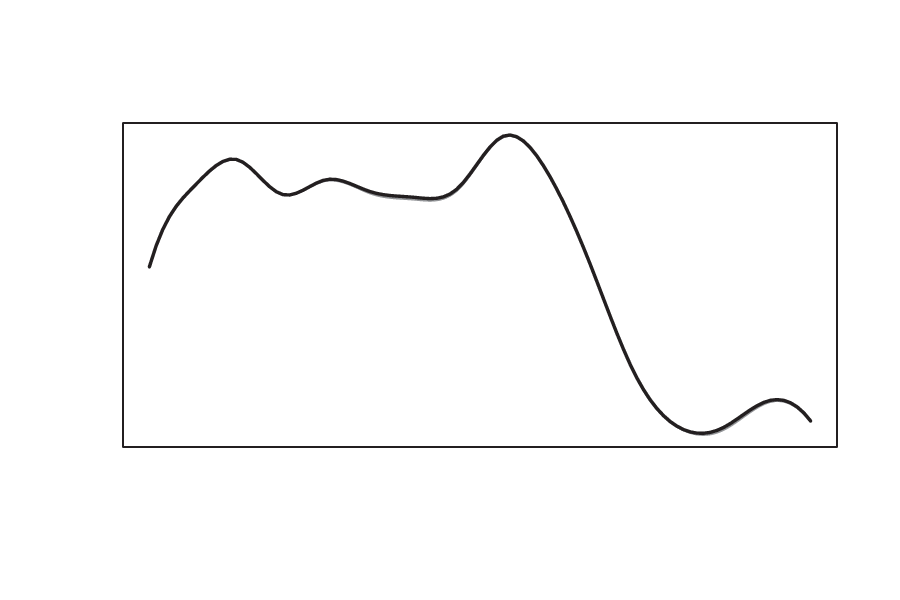}}{\quad$n_{14}=2$}
	\hspace{-.35in}
	\stackunder[-29pt]{
		\includegraphics[width=4.0cm, height=4.4cm, trim=0 0cm 0 0]
		{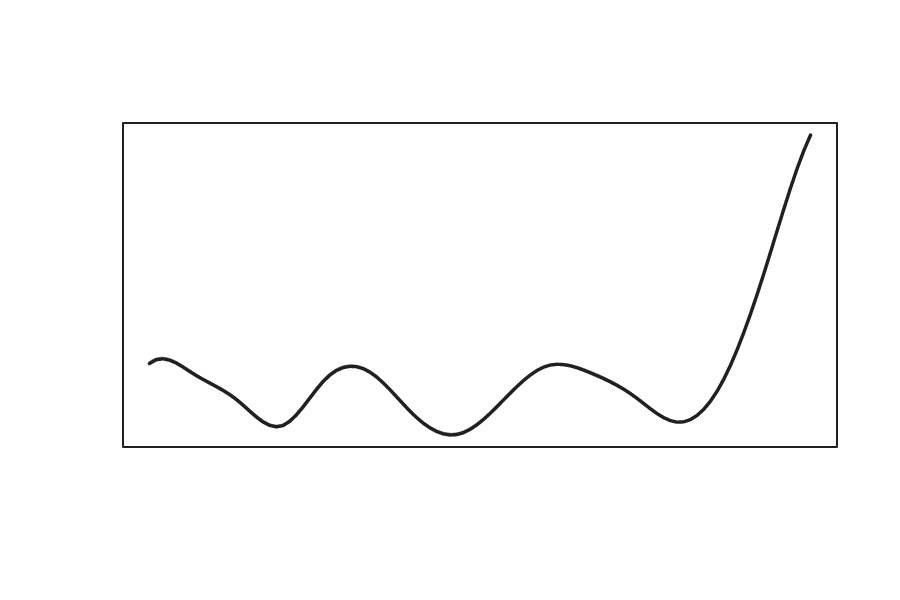}}{\quad$n_{15}=1$}
	\\[-.3in]
	\vspace{-.15in}
\end{adjustwidth}
\begin{adjustwidth}{-0.5cm}{-0.5cm}
\caption{Clusters estimated by fitting MFRMMx with influence of covariates, assuming $A=0.1$, $p=15$ and $\phi=20$, for the NBA dataset. The gray lines are the estimated individual-specific curves and the black lines are the cluster-specific mean curves. The values $n_j$ denote the number of individuals in cluster $j$.}
\end{adjustwidth}
\label{fig:NBA_clusters_A01_p15_phi20}
\end{figure}

%\newpage
%\input{fig_NBA_clusters_A1_p15_phi20-ylim.tex}
\begin{figure}[h]%[!htb]
\centering
	\stackunder[-29pt]{
	%\raisebox{1.3cm}{\rotatebox[origin=c]{90}{\scriptsize $\phi=0.1$}}
		\includegraphics[width=4cm, height=4.4cm, trim=0 0cm 0 0]
		{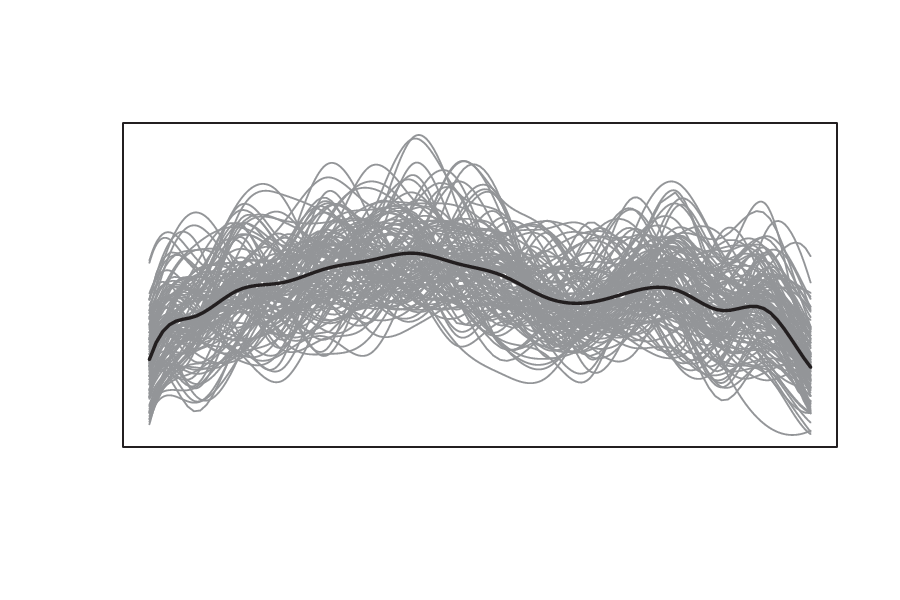}}{\quad$n_1=113$}
	\hspace{-.35in}
	\stackunder[-29pt]{
		\includegraphics[width=4cm, height=4.4cm, trim=0 0cm 0 0]
		{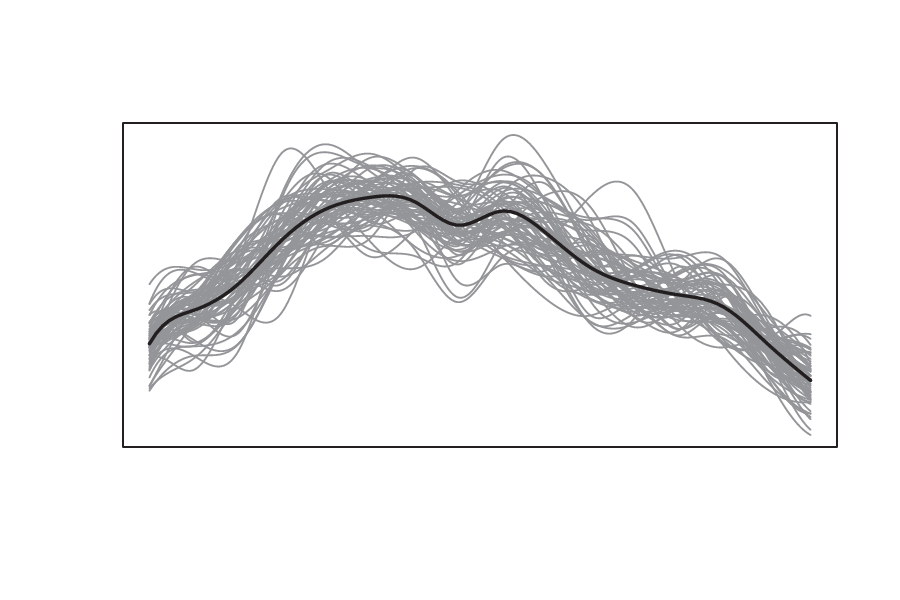}}{\quad$n_2=66$}
	\hspace{-.35in}
	\stackunder[-29pt]{
		\includegraphics[width=4cm, height=4.4cm, trim=0 0cm 0 0]
		{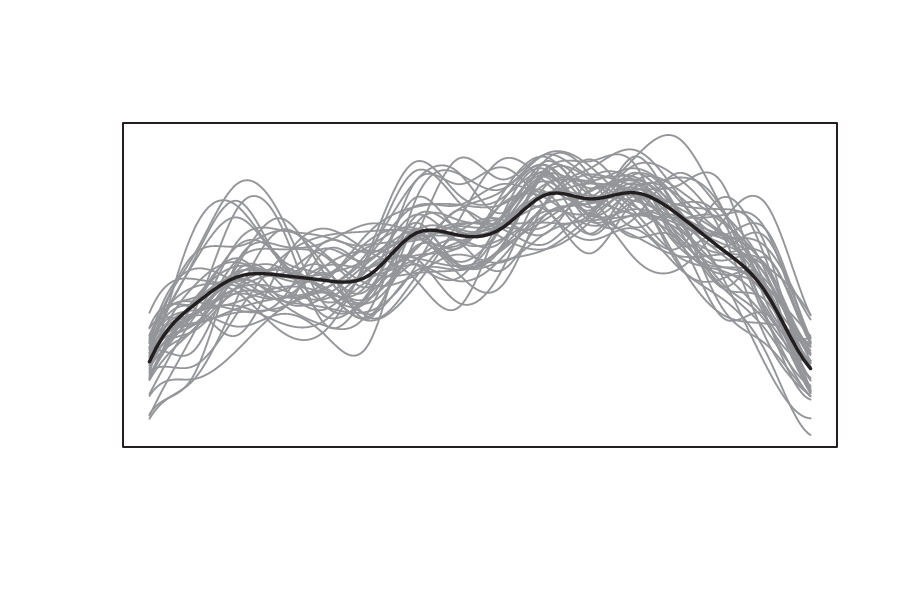}}{\quad$n_3=41$}
	\\[-.6in]
	\stackunder[-29pt]{
%			\raisebox{1.3cm}{\rotatebox[origin=c]{90}{\scriptsize $\phi=0.1$}}
		\includegraphics[width=4cm, height=4.4cm, trim=0 0cm 0 0]
		{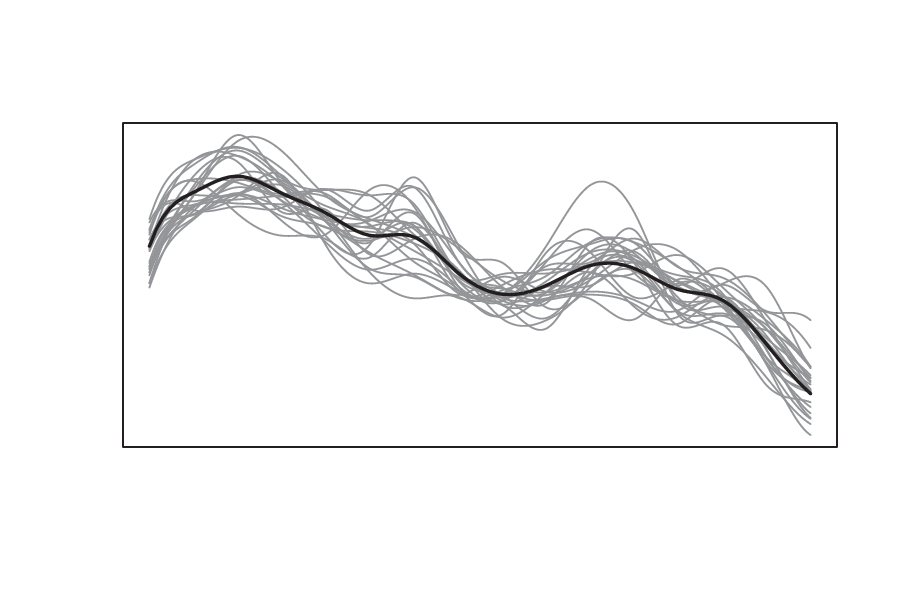}}{\quad$n_4=24$}
	\hspace{-.35in}
	\stackunder[-29pt]{
		\includegraphics[width=4cm, height=4.4cm, trim=0 0cm 0 0]
		{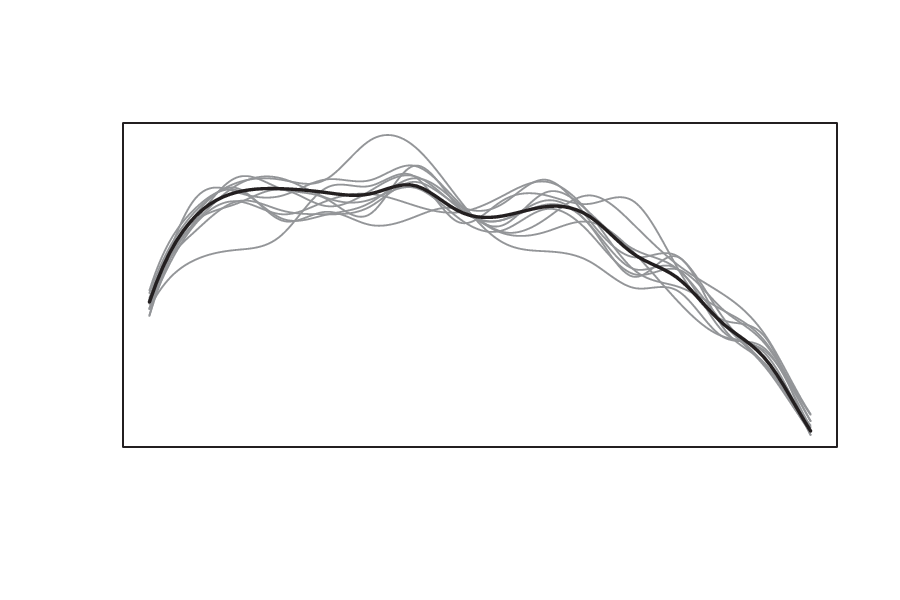}}{\quad$n_5=10$}
	\hspace{-.35in}
	\stackunder[-29pt]{
		\includegraphics[width=4cm, height=4.4cm, trim=0 0cm 0 0]
		{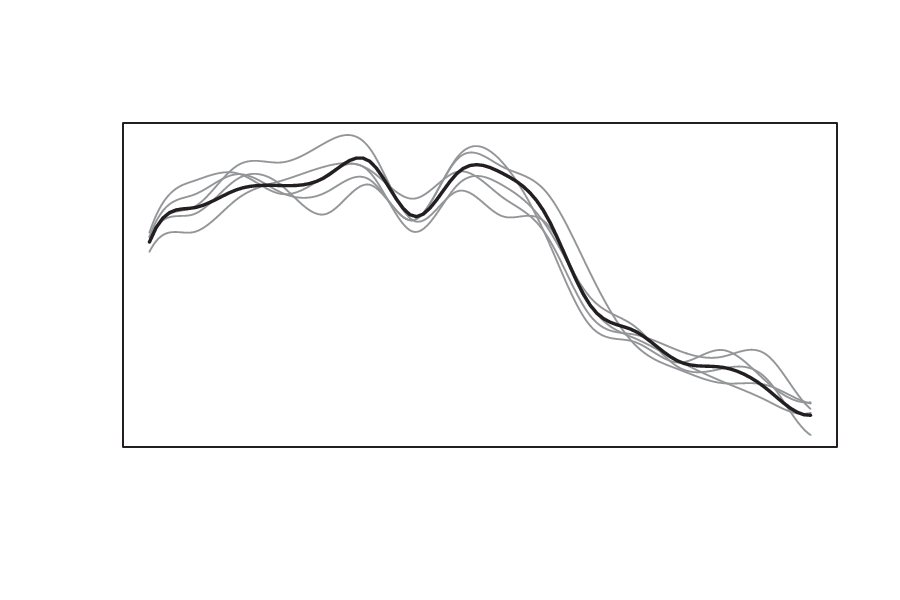}}{\quad$n_6=5$}
	\\[-.6in]
	\stackunder[-29pt]{
	%\raisebox{1.3cm}{\rotatebox[origin=c]{90}{\scriptsize $\phi=0.1$}}
		\includegraphics[width=4cm, height=4.4cm, trim=0 0cm 0 0]
		{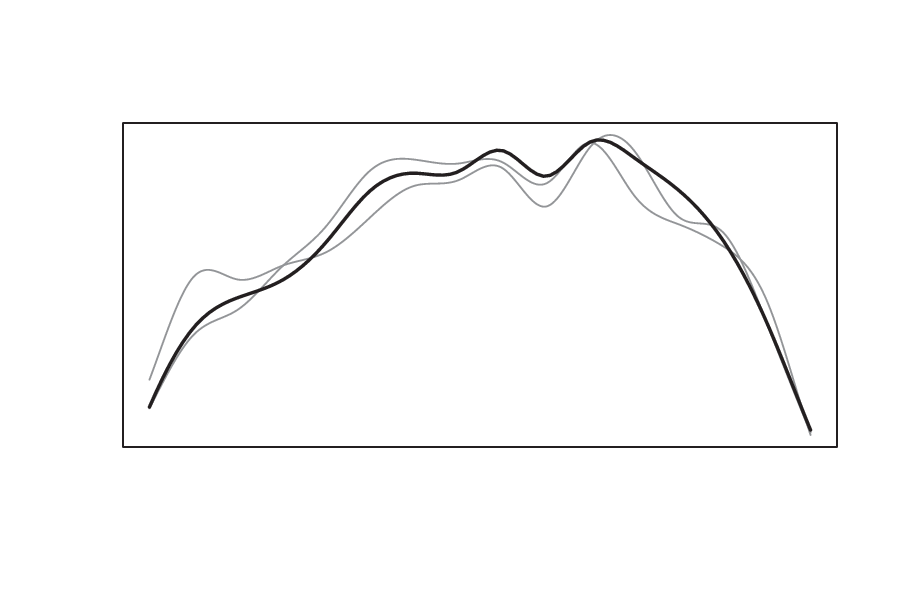}}{\quad$n_7=2$}
	\hspace{-.35in}
	\stackunder[-29pt]{
		\includegraphics[width=4cm, height=4.4cm, trim=0 0cm 0 0]
		{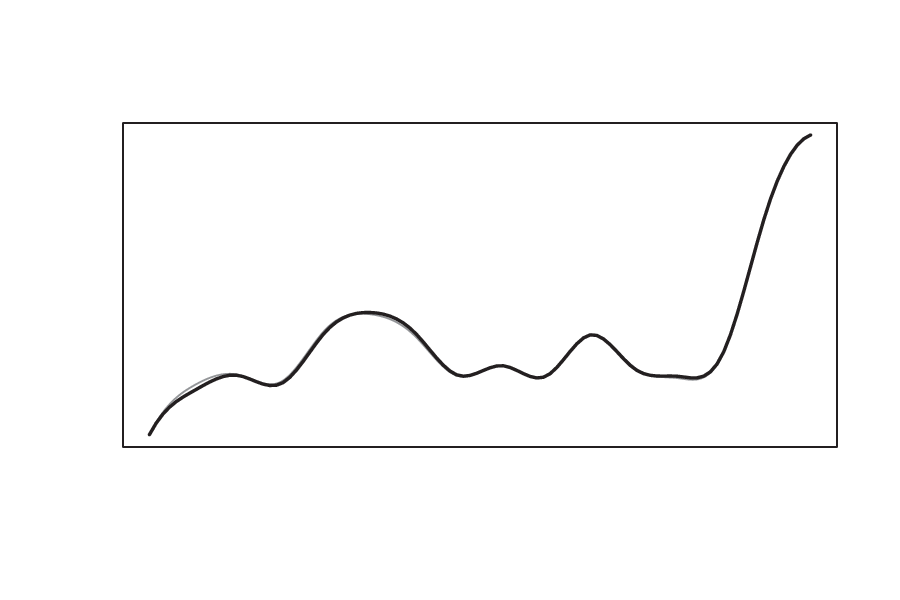}}{\quad$n_8=1$}
	\hspace{-.35in}
	\stackunder[-29pt]{
		\includegraphics[width=4cm, height=4.4cm, trim=0 0cm 0 0]
		{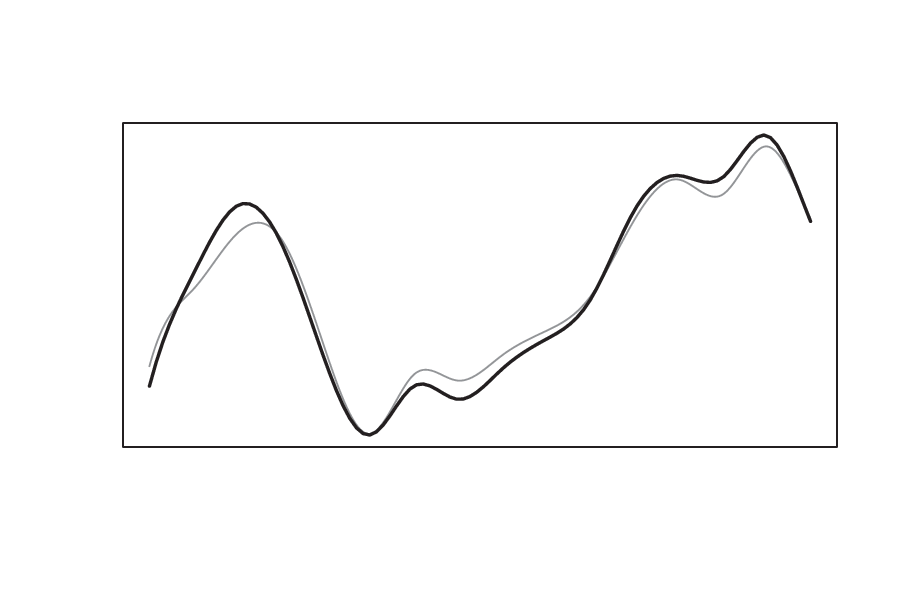}}{\quad$n_9=1$}
	\\[-.3in]
	\vspace{-.1in}
	\caption{Clusters estimated by fitting MFRMMx with influence of covariates, assuming $A=1$, $p=15$ and $\phi=20$, for the NBA dataset. The gray lines are the estimated individual-specific curves and the black lines are the cluster-specific mean curves. The values $n_j$ denote the number of individuals in cluster $j$.}
\label{fig:NBA_clusters_A1_p15_phi20}
\end{figure}

\end{document}